\documentclass[preprintnumbers,amsmath,amssymb,superscriptaddress]{revtex4}
\usepackage{amsmath,amssymb,braket,bm}
\usepackage{times}
\usepackage{graphicx}
\usepackage{dcolumn}
\usepackage{dsfont}
\usepackage{mathrsfs}
\usepackage[landscape,papersize={297.1mm,210mm},left=1.9cm,right=1.6cm,top=2.7cm,bottom=2.8cm]{geometry}
\usepackage{amsthm}
\usepackage{xcolor}
\usepackage{hyperref}
\makeatletter

\newcommand{\Rmnum}[1]{\expandafter\@slowromancap\romannumeral #1@}

\newtheorem{theorem}{Theorem}
\makeatother
\begin{document}

\title{Entanglement-Assisted Discrimination of Nonlocal Sets of Orthogonal States}
\author{Ziying Hou}
\affiliation{School of Mathematics and Science, Hebei GEO University, Shijiazhuang 052161, China}

\author{Huaqi Zhou}
\email{zhouhuaqilc@163.com}
\affiliation{School of Mathematics and Science, Hebei GEO University, Shijiazhuang 052161, China}

\author{Limin Gao}
\email{gaoliminabc@163.com}
\affiliation{School of Mathematics and Science, Hebei GEO University, Shijiazhuang 052161, China}
\affiliation{Intelligent Sensor Network Engineering Research Center of Hebei Province, Hebei GEO University, Shijiazhuang 052161, China}

\begin{abstract}
Entanglement-assisted discrimination of orthogonal quantum states exhibiting quantum nonlocality is a frontier topic in quantum information theory. In this paper, we investigate the role of multipartite entanglement and develop resource-efficient LOCC discrimination protocols for nonlocal sets of orthogonal states, including multipartite orthogonal product-state sets and entangled-state sets with different nonlocal features. By incorporating controlled-NOT (CNOT) operations into the discrimination procedure, we construct protocols for genuinely nonlocal GHZ bases in four- and five-qubit systems that require only a single EPR pair. For the same target sets, we compare different entanglement-assisted schemes and identify those with lower entanglement consumption. We further observe that, on average, protocols avoiding teleportation consume fewer resources than teleportation-based approaches. In addition, when higher-partite GHZ-type resources (with $n>3$) are available among suitable subsystems, they can in some cases reduce the overall entanglement cost. Our results highlight the operational significance of multipartite entanglement and provide practical protocols for the local discrimination of orthogonal state sets exhibiting quantum nonlocality.
\end{abstract}

\maketitle
\section{Introduction}
Quantum nonlocality is one of the fundamental features of quantum mechanics. Early studies mainly focused on the nonlocality of entangled pure states. As a central resource in quantum information science, quantum entanglement has attracted widespread attention, and extensive progress has been made in this direction \cite{Ref1,Ref2,Ref3}. Local indistinguishability refers to the phenomenon that a set of orthogonal quantum states cannot be perfectly distinguished by local operations and classical communication (LOCC). In 1999, Bennett \textit{et al.}~\cite{Ref4} constructed a class of complete orthogonal product bases that are locally indistinguishable, thereby revealing for the first time that a set of product states can also exhibit quantum nonlocality. Unlike Bell nonlocality, which originates from entangled pure states, quantum nonlocality based on local indistinguishability is not restricted to entangled states \cite{Ref5,Ref6,Ref7,Ref8,Ref9,Ref10}. Since then, substantial progress has been made in proving the existence of and constructing orthogonal entangled sets and orthogonal product sets with quantum nonlocality \cite{Ref11,Ref12,Ref13,Ref14,Ref15,Ref16,Ref17,Ref18,Ref19,Ref20,Ref21,Ref22,Ref23,Ref24,Ref25}.

Recently, Horodecki \textit{et al.}~\cite{Ref26} showed that one can have ``more nonlocality with less entanglement.'' In multipartite systems, genuine nonlocality has also been extensively studied \cite{Ref27}. A set of orthogonal states is said to be genuinely nonlocal if it cannot be perfectly distinguished by LOCC across any bipartition of the subsystems. Rout \textit{et al.}~\cite{Ref28} constructed genuinely nonlocal orthogonal product sets in multipartite systems such as $\mathbb{C}^{4}\otimes\mathbb{C}^{3}\otimes\mathbb{C}^{3}$ and $\mathbb{C}^{6}\otimes\mathbb{C}^{5}\otimes\mathbb{C}^{5}$. They further generalized such sets to $(m+1)$-partite systems of the form $\mathbb{C}^{m+2}\otimes(\mathbb{C}^{3})^{\otimes m}$, thereby providing an effective way to construct product-state sets with minimal genuine nonlocality in multipartite systems. Xiong \textit{et al.}~\cite{Ref29} showed that any five states in the three-qubit GHZ basis form a genuinely nonlocal set, whereas any four states do not. They also studied the construction of small genuinely nonlocal sets of generalized Greenberger--Horne--Zeilinger (GHZ) states in multipartite systems \cite{Ref30}. In particular, they derived a universal bound on the cardinality required for a GHZ-like state set to possess this form of nonlocality and confirmed the existence of small genuinely nonlocal sets whose cardinality scales linearly with the local dimension $d$. Halder \textit{et al.}~\cite{Ref31} introduced the notion of locally irreducible sets and explicitly defined \emph{strong nonlocality}, meaning that a set of orthogonal states is locally irreducible across all bipartitions. Since local irreducibility is sufficient but not necessary for local indistinguishability, strong nonlocality constitutes a stronger manifestation of genuine nonlocality. In 2020, Shi \textit{et al.}~\cite{Ref32} constructed strongly nonlocal sets containing entangled states based on the Rubik's cube structure. Two years later, in Hilbert spaces $\mathbb{C}^{d_A}\otimes\mathbb{C}^{d_B}\otimes\mathbb{C}^{d_C}$ ($d_A,d_B,d_C\ge 4$) and $\mathbb{C}^{d_A}\otimes\mathbb{C}^{d_B}\otimes\mathbb{C}^{d_C}\otimes\mathbb{C}^{d_D}$ ($d_A,d_B,d_C,d_D\ge 3$), Gao \textit{et al.}~\cite{Ref33} proposed strongly nonlocal orthogonal product sets (OPSs) with fewer states. Li \textit{et al.}~\cite{Ref34} established strongest nonlocal sets of various sizes for four-qudit systems $\mathbb{C}^{d}\otimes\mathbb{C}^{d}\otimes\mathbb{C}^{d}\otimes\mathbb{C}^{d}$ ($d\ge 2$) and, more generally, for four-partite systems. In addition, Zhang \textit{et al.}~\cite{Ref35} investigated \emph{generally strong nonlocality}, namely the property that a set of orthogonal $n$-partite states is locally irreducible in every $(n-1)$-partition. Gao \textit{et al.}~\cite{Ref36} constructed generally strongly nonlocal OPSs with asymmetric and symmetric structures in four-partite systems, and further presented a general construction scheme for strongly nonlocal OPSs in $n$-partite systems with $n\ge 5$.

Even when all parties are spatially separated and restricted to LOCC, a nonlocal set of orthogonal states can still be perfectly distinguished provided that sufficient entanglement is shared among them \cite{Ref37,Ref38,Ref39,Ref40}. A general approach is to transfer all subsystems to a single party via quantum teleportation \cite{Ref41,Ref42,Ref43}, although this may consume a large amount of entanglement. Therefore, how to locally distinguish nonlocal sets while minimizing entanglement consumption is an important problem \cite{Ref44,Ref45,Ref46,Ref47,Ref48,Ref49}. Cohen~\cite{Ref50} proposed an entanglement-efficient discrimination scheme applicable to a class of unextendible product bases. Rout \textit{et al.}~\cite{Ref27} proposed several local discrimination protocols for genuinely nonlocal product bases using Einstein--Podolsky--Rosen (EPR) states and GHZ states, thereby demonstrating the usefulness of genuinely multipartite entangled resources. Zhang \textit{et al.}~\cite{Ref51,Ref52} showed that some unextendible product bases can be locally discriminated with the assistance of multiple EPR states. Gao \textit{et al.}~\cite{Ref33} found that a $d\otimes d$ maximally entangled state can be replaced by $n$ EPR pairs in discrimination whenever $2^n\ge d$.

Entanglement-assisted discrimination is closely related to orthogonal state sets exhibiting quantum nonlocality. Numerous orthogonal product sets and orthogonal entangled sets with different strengths of nonlocality have been constructed. Local indistinguishability also underlies applications such as quantum data hiding and quantum secret sharing \cite{Ref53,Ref54,Ref55,Ref56}. Related applications of quantum entanglement and entanglement-assisted discrimination continue to expand \cite{Ref57,Ref58}. Nevertheless, research on entanglement-assisted discrimination of nonlocal sets remains incomplete: for some orthogonal state sets, no local identification schemes more efficient than teleportation-based approaches are currently known, and more resource-efficient protocols are still needed.

In this paper, we study the entanglement-assisted local distinguishability of several known nonlocal sets. In Sec.~\ref{Q2}, we introduce the necessary preliminaries. In Sec.~\ref{Q3}, for the GHZ bases in four- and five-qubit systems, we propose entanglement-assisted discrimination protocols incorporating controlled-NOT (CNOT) gates such that only one EPR pair is required. In Sec.~\ref{Q4}, we present local discrimination protocols for strongly nonlocal asymmetric and symmetric OPSs in four-partite systems. In Sec.~\ref{Q5}, we further study strongly nonlocal symmetric OPSs in five-partite systems. Finally, we conclude in Sec.~\ref{Q6}.
\section{Preliminary}\label{Q2}
In this section, we introduce the notation and concepts used throughout the paper.

\textbf{\textit{Definition 1} \cite{Ref31}.} A set of orthogonal pure states in a multipartite quantum system is said to be \emph{locally indistinguishable} if it is impossible to perfectly distinguish all states in the set using only LOCC.

\textbf{\textit{Definition 2} \cite{Ref31}.} A set of orthogonal quantum states on \(\bigotimes_{i=1}^{n}\mathbb{C}^{d_i}\), with \(n \ge 2\) and \(d_i \ge 2\) for all \(i=1,\ldots,n\), is said to be \emph{locally irreducible} if it is impossible to eliminate one or more states from the set by means of orthogonality-preserving local measurements.

According to these definitions, local irreducibility is strictly stronger than local indistinguishability.

\textbf{\textit{Definition 3} \cite{Ref30}.} A set of orthogonal multipartite quantum states is said to be \emph{genuinely nonlocal} if it is locally indistinguishable across every bipartition of the subsystems.

\textbf{\textit{Definition 4} \cite{Ref35}.} If a set of orthogonal quantum states is locally irreducible in every $(n-1)$-partition of an $n$-partite system, then it has the \emph{generally strong nonlocality}.

Unless otherwise specified, the term ``generally strong nonlocality'' used below follows Definition~4. For convenience, some displayed states are written in unnormalized form.

To facilitate the discussion of local discrimination protocols, we adopt a unified notation for entanglement-resource configurations. We write
\begin{equation*}
\{(p,|\phi^+(d)\rangle_{AB});(r,|G\rangle_{ABC});(q,|F\rangle_{ABCD})\}
\end{equation*}
to denote an entanglement resource configuration \cite{Ref27}, where \(p\), \(r\), and \(q\) are nonnegative real numbers. This means that, on average, \(p\) copies of the maximally entangled state
\[
|\phi^{+}(d)\rangle_{AB}=\frac{1}{\sqrt{d}}\sum_{i=0}^{d-1}|ii\rangle_{AB}
\]
are consumed between parties \(A\) and \(B\). Similarly, \(r\) copies of
\[
|G\rangle_{ABC}=\frac{1}{\sqrt{2}}\sum_{i=0}^{1}|iii\rangle_{ABC}
\]
and \(q\) copies of
\[
|F\rangle_{ABCD}=\frac{1}{\sqrt{2}}\sum_{i=0}^{1}|iiii\rangle_{ABCD}
\]
are consumed among \(A,B,C\) and among \(A,B,C,D\), respectively.

In entanglement-assisted discrimination protocols, we frequently use orthogonality-preserving local measurements on different subsystems. Throughout the paper, we denote a generic projector by
\begin{equation*}
M_i=P\bigl[(|j_{i1}\rangle,|j_{i2}\rangle)_X;|k_i\rangle_x\bigr]+P\bigl[|j_i'\rangle_X;|k_i'\rangle_x\bigr],
\end{equation*}
with \(\sum_i M_i=I\). Here \(\{|j_{i1}\rangle,|j_{i2}\rangle\}\) and \(\{|j_i'\rangle\}\) are subsets of the computational basis of subsystem \(X\), while \(\{|k_i\rangle\}\) and \(\{|k_i'\rangle\}\) are subsets of the computational basis of the auxiliary system \(x\). The term
\[
P\bigl[(|j_{i1}\rangle,|j_{i2}\rangle)_X;|k_i\rangle_x\bigr]
\]
stands for
\[
\left(|j_{i1}\rangle_X\langle j_{i1}|+|j_{i2}\rangle_X\langle j_{i2}|\right)\otimes |k_i\rangle_x\langle k_i|,
\]
and \(P\bigl[|j_i'\rangle_X;|k_i'\rangle_x\bigr]\) is defined analogously.

\section{Local Discrimination of the GHZ Bases with Genuine Quantum Nonlocality}\label{Q3}
A nonlocal set of orthogonal states can always be perfectly distinguished locally once sufficient entanglement resources are shared among the parties. The key issue is how to reduce the required entanglement consumption. In an $n$-qubit system, the GHZ basis is an orthogonal basis consisting of generalized GHZ states \cite{Ref30}. For $n=4$ and $n=5$, this basis is not only locally irreducible in the full $n$-partition \cite{Ref31}, but also locally indistinguishable across every bipartition \cite{Ref30}. In this section, we construct entanglement-assisted discrimination protocols for the GHZ bases in four- and five-qubit systems, respectively. To save entanglement resources, we incorporate controlled-NOT (CNOT) gates \cite{Ref59} into the protocols, as stated in Theorems~\ref{thm:4qubit-ghz} and~\ref{thm:5qubit-ghz}.

\subsection{Local Distinguishability of the GHZ Basis in a Four-Qubit System}

In the system $\mathcal{C}^{2} \otimes \mathcal{C}^{2} \otimes \mathcal{C}^{2} \otimes \mathcal{C}^{2}$, the genuinely nonlocal GHZ basis in Ref.~\cite{Ref30} is
\begin{equation}\label{set3.1}
\begin{aligned}
&\left |\psi_{\pm}   \right \rangle _{0} =\left | 0000  \right \rangle \pm \left | 1111 \right \rangle  , \\
&\left |\psi_{\pm}   \right \rangle_{1}  =\left | 0010  \right \rangle \pm \left | 1101 \right \rangle , \\
&\left |\psi_{\pm}   \right \rangle _{2} =\left | 0100  \right \rangle \pm \left | 1011 \right \rangle, \\
&\left |\psi_{\pm}   \right \rangle_{3}  =\left | 0110  \right \rangle \pm \left | 1001 \right \rangle , \\
&\left |\psi_{\pm}   \right \rangle_{4}  =\left | 0001  \right \rangle \pm \left | 1110 \right \rangle, \\
&\left |\psi_{\pm}   \right \rangle_{5}  =\left | 0011  \right \rangle \pm \left | 1100\right \rangle ,\\
&\left |\psi_{\pm}   \right \rangle_{6}  =\left | 0101  \right \rangle \pm \left | 1010 \right \rangle, \\
&\left |\psi_{\pm}   \right \rangle_{7}  =\left | 0111  \right \rangle \pm \left | 1000 \right \rangle.
\end{aligned}
\end{equation}

\begin{theorem}\label{thm:4qubit-ghz}
In a four-qubit system, the nonlocal basis given by Eq.~(\ref{set3.1}) can be perfectly discriminated using the entanglement resource configuration $(1,|\phi^{+}(2)\rangle_{AB})$.
\end{theorem}

\begin{figure}[h]
\centering
\includegraphics[width=0.48\textwidth]{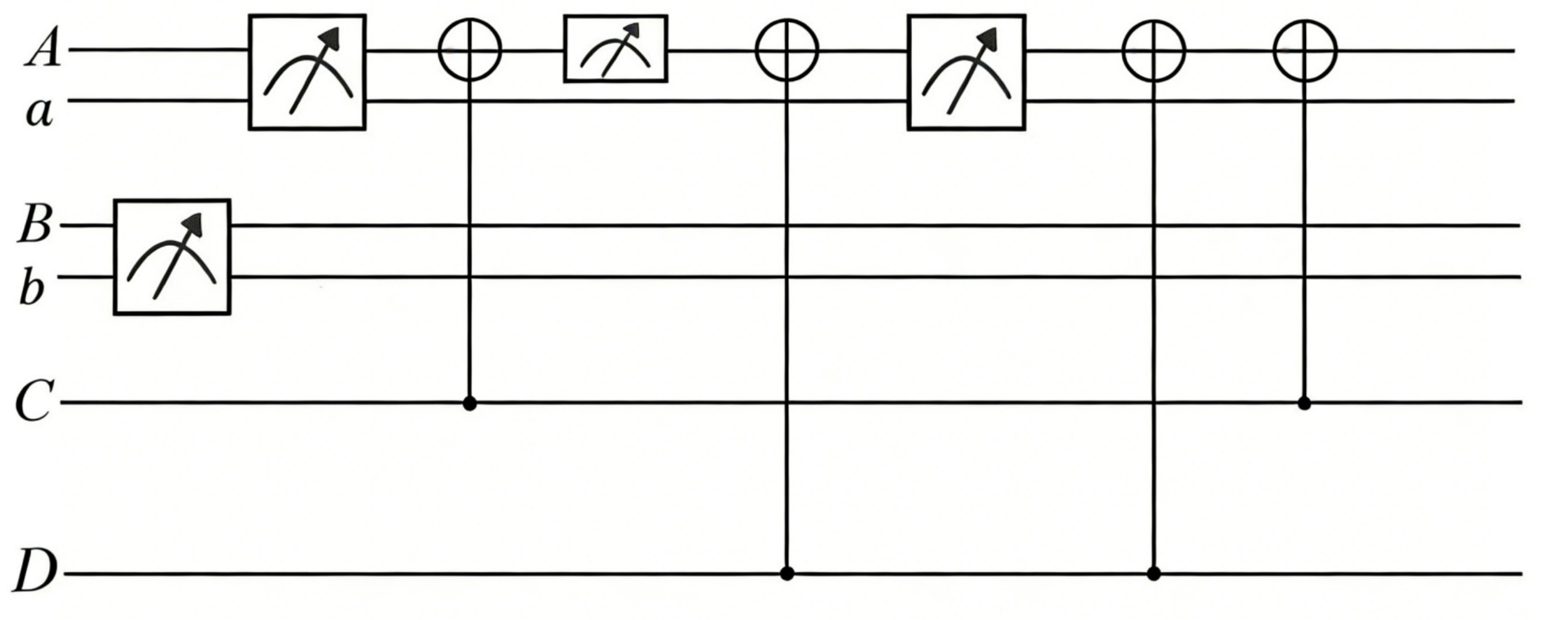}
\caption{Circuit diagram of the discrimination scheme for the GHZ basis~(\ref{set3.1}) in a four-qubit system. Lines $A$ and $a$ correspond to Alice's main subsystem and ancilla, respectively; lines $B$ and $b$ correspond to Bob's main subsystem and ancilla, respectively. The lower lines $C$ and $D$ correspond to Charlie and Dave. The meter denotes a measurement operation.\label{f3.1}}
\end{figure}

\begin{proof}
Figure~\ref{f3.1} presents the procedure for local discrimination of the set~(\ref{set3.1}). The four subsystems are held by Alice, Bob, Charlie, and Dave, respectively. The protocol proceeds as follows.

\medskip\noindent\textbf{Step 1.}
Alice and Bob share an EPR state $|\phi^{+}(2)\rangle_{ab}=|00\rangle_{ab}+|11\rangle_{ab}$. The initial state is
\[
|\psi\rangle_{ABCD} \otimes |\phi^{+}(2)\rangle_{ab},
\]
where $a$ and $b$ are ancillary subsystems of Alice and Bob, respectively. Bob performs the measurement
\[
\mathcal{M}_1 \equiv \bigl\{
M_{1,1} = P[|0\rangle_B; |0\rangle_b] + P[|1\rangle_B; |1\rangle_b],\,
M_{1,2} = I - M_{1,1}
\bigr\}.
\]
Conditioned on the outcome $M_{1,1}$, the states become
\[
\begin{aligned}
&|\psi_{\pm}\rangle_{0} =|0000\rangle |00\rangle \pm |1111\rangle |11\rangle, \\
&|\psi_{\pm}\rangle_{1} =|0010\rangle |00\rangle \pm |1101\rangle |11\rangle, \\
&|\psi_{\pm}\rangle_{2} =|0100\rangle |11\rangle \pm |1011\rangle |00\rangle, \\
&|\psi_{\pm}\rangle_{3} =|0110\rangle |11\rangle \pm |1001\rangle |00\rangle, \\
&|\psi_{\pm}\rangle_{4} =|0001\rangle |00\rangle \pm |1110\rangle |11\rangle, \\
&|\psi_{\pm}\rangle_{5} =|0011\rangle |00\rangle \pm |1100\rangle |11\rangle, \\
&|\psi_{\pm}\rangle_{6} =|0101\rangle |11\rangle \pm |1010\rangle |00\rangle, \\
&|\psi_{\pm}\rangle_{7} =|0111\rangle |11\rangle \pm |1000\rangle |00\rangle,
\end{aligned}
\]
where $|00\rangle$ and $|11\rangle$ denote $|00\rangle_{ab}$ and $|11\rangle_{ab}$, respectively.

\medskip\noindent\textbf{Step 2.}
Alice performs the measurement
\[
\mathcal{M}_2 \equiv \bigl\{
M_{2,1} = P[|0\rangle_A; |0\rangle_a] + P[|1\rangle_A; |1\rangle_a],\,
M_{2,2} = I - M_{2,1}
\bigr\}.
\]
If outcome $M_{2,1}$ occurs, the state is one of
$|\psi_{\pm}\rangle_{0}$, $|\psi_{\pm}\rangle_{1}$, $|\psi_{\pm}\rangle_{4}$, and $|\psi_{\pm}\rangle_{5}$.
If outcome $M_{2,2}$ occurs, the state is one of
$|\psi_{\pm}\rangle_{2}$, $|\psi_{\pm}\rangle_{3}$, $|\psi_{\pm}\rangle_{6}$, and $|\psi_{\pm}\rangle_{7}$.
Denote these two parts by $P_1$ and $P_2$, respectively:
\[
\begin{aligned}
P_1 &= \bigl\{ |\psi_{\pm}\rangle_{0}, |\psi_{\pm}\rangle_{1}, |\psi_{\pm}\rangle_{4}, |\psi_{\pm}\rangle_{5} \bigr\},\\
P_2 &= \bigl\{ |\psi_{\pm}\rangle_{2}, |\psi_{\pm}\rangle_{3}, |\psi_{\pm}\rangle_{6}, |\psi_{\pm}\rangle_{7} \bigr\}.
\end{aligned}
\]

\medskip\noindent\textbf{Step 3.}
Each part cannot be perfectly distinguished by projective measurements alone at this stage, so we apply a CNOT gate with Charlie as the control qubit and Alice as the target qubit. Then the states in $P_1$ and $P_2$ become
\[
\begin{aligned}
P_1:\;
\begin{cases}
|\psi_{\pm}\rangle_{0}=|0000\rangle|00\rangle \pm |0111\rangle|11\rangle,\\
|\psi_{\pm}\rangle_{1}=|1010\rangle|00\rangle \pm |1101\rangle|11\rangle,\\
|\psi_{\pm}\rangle_{4}=|0001\rangle|00\rangle \pm |0110\rangle|11\rangle,\\
|\psi_{\pm}\rangle_{5}=|1011\rangle|00\rangle \pm |1100\rangle|11\rangle,
\end{cases}
\end{aligned}
\]
and
\[
\begin{aligned}
P_2:\;
\begin{cases}
|\psi_{\pm}\rangle_{2}=|0100\rangle|11\rangle \pm |0011\rangle|00\rangle,\\
|\psi_{\pm}\rangle_{3}=|1110\rangle|11\rangle \pm |1001\rangle|00\rangle,\\
|\psi_{\pm}\rangle_{6}=|0101\rangle|11\rangle \pm |0010\rangle|00\rangle,\\
|\psi_{\pm}\rangle_{7}=|1111\rangle|11\rangle \pm |1000\rangle|00\rangle.
\end{cases}
\end{aligned}
\]

\medskip\noindent\textbf{Step 4.}
Alice performs the measurement
\[
\mathcal{M}_3 \equiv \bigl\{
M_{3,1}=|0\rangle_A\langle 0|,\,
M_{3,2}= I - M_{3,1}
\bigr\}.
\]
For the parts $P_1$ and $P_2$ in Step~3, the outcomes correspond to
\[
P_1:\;
\begin{cases}
M_{3,1}\Rightarrow \{|\psi_{\pm}\rangle_{0}, |\psi_{\pm}\rangle_{4}\},\\
M_{3,2}\Rightarrow \{|\psi_{\pm}\rangle_{1}, |\psi_{\pm}\rangle_{5}\},
\end{cases}
\qquad
P_2:\;
\begin{cases}
M_{3,1}\Rightarrow \{|\psi_{\pm}\rangle_{2}, |\psi_{\pm}\rangle_{6}\},\\
M_{3,2}\Rightarrow \{|\psi_{\pm}\rangle_{3}, |\psi_{\pm}\rangle_{7}\}.
\end{cases}
\]
Thus we obtain four subparts:
\[
\begin{aligned}
P_{11} &= \{|\psi_{\pm}\rangle_{0}, |\psi_{\pm}\rangle_{4}\},&
P_{12} &= \{|\psi_{\pm}\rangle_{1}, |\psi_{\pm}\rangle_{5}\},\\
P_{21} &= \{|\psi_{\pm}\rangle_{2}, |\psi_{\pm}\rangle_{6}\},&
P_{22} &= \{|\psi_{\pm}\rangle_{3}, |\psi_{\pm}\rangle_{7}\}.
\end{aligned}
\]

\medskip\noindent\textbf{Step 5.}
Apply a CNOT gate again, now with Dave as the control qubit and Alice as the target qubit. The states become
\[
\begin{aligned}
P_{11}:\;
\begin{cases}
|\psi_{\pm}\rangle_{0}=|0000\rangle|00\rangle \pm |1111\rangle|11\rangle,\\
|\psi_{\pm}\rangle_{4}=|1001\rangle|00\rangle \pm |0110\rangle|11\rangle,
\end{cases}
\qquad
P_{12}:\;
\begin{cases}
|\psi_{\pm}\rangle_{1}=|1010\rangle|00\rangle \pm |0101\rangle|11\rangle,\\
|\psi_{\pm}\rangle_{5}=|0011\rangle|00\rangle \pm |1100\rangle|11\rangle,
\end{cases}
\end{aligned}
\]
and
\[
\begin{aligned}
P_{21}:\;
\begin{cases}
|\psi_{\pm}\rangle_{2}=|0100\rangle|11\rangle \pm |1011\rangle|00\rangle,\\
|\psi_{\pm}\rangle_{6}=|1101\rangle|11\rangle \pm |0010\rangle|00\rangle,
\end{cases}
\qquad
P_{22}:\;
\begin{cases}
|\psi_{\pm}\rangle_{3}=|1110\rangle|11\rangle \pm |0001\rangle|00\rangle,\\
|\psi_{\pm}\rangle_{7}=|0111\rangle|11\rangle \pm |1000\rangle|00\rangle.
\end{cases}
\end{aligned}
\]

\medskip\noindent\textbf{Step 6.}
Alice performs the measurement
\[
\mathcal{M}_4 \equiv \bigl\{
M_{4,1}=P[|0\rangle_A; |0\rangle_a] + P[|1\rangle_A; |1\rangle_a],\,
M_{4,2}=I-M_{4,1}
\bigr\}.
\]
For the states in $P_{11}$, $P_{12}$, $P_{21}$, and $P_{22}$ given by Step~5, the outcomes identify $i\in\{0,1,\dots,7\}$ as follows:
\[
\begin{aligned}
P_{11}:\;
\begin{cases}
M_{4,1} \Rightarrow |\psi_{\pm}\rangle_{0},\\
M_{4,2} \Rightarrow |\psi_{\pm}\rangle_{4},
\end{cases}
\quad
P_{12}:\;
\begin{cases}
M_{4,1} \Rightarrow |\psi_{\pm}\rangle_{5},\\
M_{4,2} \Rightarrow |\psi_{\pm}\rangle_{1},
\end{cases}
\quad
P_{21}:\;
\begin{cases}
M_{4,1} \Rightarrow |\psi_{\pm}\rangle_{6},\\
M_{4,2} \Rightarrow |\psi_{\pm}\rangle_{2},
\end{cases}
\end{aligned}
\]
and
\[
\begin{aligned}
P_{22}:\;
\begin{cases}
M_{4,1} \Rightarrow |\psi_{\pm}\rangle_{3},\\
M_{4,2} \Rightarrow |\psi_{\pm}\rangle_{7}.
\end{cases}
\end{aligned}
\]

Since any two orthogonal pure states can be exactly distinguished by LOCC \cite{Ref60}, each pair $\{|\psi_{+}\rangle_{i},|\psi_{-}\rangle_{i}\}$ is locally distinguishable. Therefore, the GHZ basis~(\ref{set3.1}) can be perfectly distinguished.

During the protocol, the original quantum states are transformed due to the two CNOT gates. Finally, we can recover the original state by applying the two CNOT gates again (see Fig.~\ref{f3.1}), i.e., first CNOT with Charlie controlling Alice, and then CNOT with Dave controlling Alice.

If outcome $M_{1,2}$ occurs in Step~1, an analogous analysis applies.
\end{proof}

Therefore, we obtain a local protocol that perfectly distinguishes this genuinely nonlocal set while consuming only one EPR pair. The same idea can also be extended to the GHZ basis in the five-qubit system.

\subsection{Local Distinguishability of the GHZ Basis in a Five-Qubit System}
In the system $\mathcal{C}^{2} \otimes \mathcal{C}^{2} \otimes \mathcal{C}^{2} \otimes \mathcal{C}^{2} \otimes \mathcal{C}^{2}$, the GHZ basis with genuine nonlocality in Ref.~\cite{Ref30} is
\begin{equation}\label{set3.2}
\begin{aligned}
&\left |\varphi _{\pm}   \right \rangle_{0} =\left | 00000  \right \rangle \pm \left | 11111 \right \rangle , \\
&\left |\varphi _{\pm}   \right \rangle_{1} =\left | 00001 \right \rangle \pm \left | 11110 \right \rangle , \\
&\left |\varphi _{\pm}   \right \rangle_{2} =\left | 00010  \right \rangle \pm \left | 11101 \right \rangle, \\
&\left |\varphi _{\pm}   \right \rangle_{3} =\left |00011  \right \rangle \pm \left | 11100 \right \rangle , \\
&\left |\varphi _{\pm}   \right \rangle_{4} =\left | 00100  \right \rangle \pm \left | 11011 \right \rangle ,\\
&\left |\varphi _{\pm}   \right \rangle_{5} =\left | 00101 \right \rangle \pm \left | 11010\right \rangle, \\
&\left |\varphi _{\pm}   \right \rangle_{6} =\left | 00110  \right \rangle \pm \left |11001\right \rangle ,\\
&\left |\varphi _{\pm}   \right \rangle_{7} =\left | 00111  \right \rangle \pm \left | 11000 \right \rangle ,\\
&\left |\varphi _{\pm}   \right \rangle_{8} =\left | 01000  \right \rangle \pm \left | 10111 \right \rangle,\\
&\left |\varphi _{\pm}   \right \rangle_{9} =\left | 01001 \right \rangle \pm \left | 10110 \right \rangle,\\
&\left |\varphi _{\pm}   \right \rangle_{10}=\left | 01010  \right \rangle \pm \left | 10101 \right \rangle,\\
&\left |\varphi _{\pm}   \right \rangle_{11}=\left | 01011  \right \rangle \pm \left | 10100 \right \rangle,\\
&\left |\varphi _{\pm}   \right \rangle_{12}=\left | 01100 \right \rangle \pm \left | 10011 \right \rangle,\\
&\left |\varphi _{\pm}   \right \rangle_{13}=\left | 01101  \right \rangle \pm \left | 10010\right \rangle,\\
&\left |\varphi _{\pm}   \right \rangle_{14}=\left | 01110  \right \rangle \pm \left | 10001 \right \rangle,\\
&\left |\varphi _{\pm}   \right \rangle_{15}=\left | 01111  \right \rangle \pm \left | 10000 \right \rangle.
\end{aligned}
\end{equation}

\begin{theorem}\label{thm:5qubit-ghz}
The GHZ basis~(\ref{set3.2}) can be locally discriminated using the entanglement resource configuration $(1,|\phi^{+}(2)\rangle_{AB})$.
\end{theorem}

\begin{figure}[h]
\centering
\includegraphics[width=0.48\textwidth]{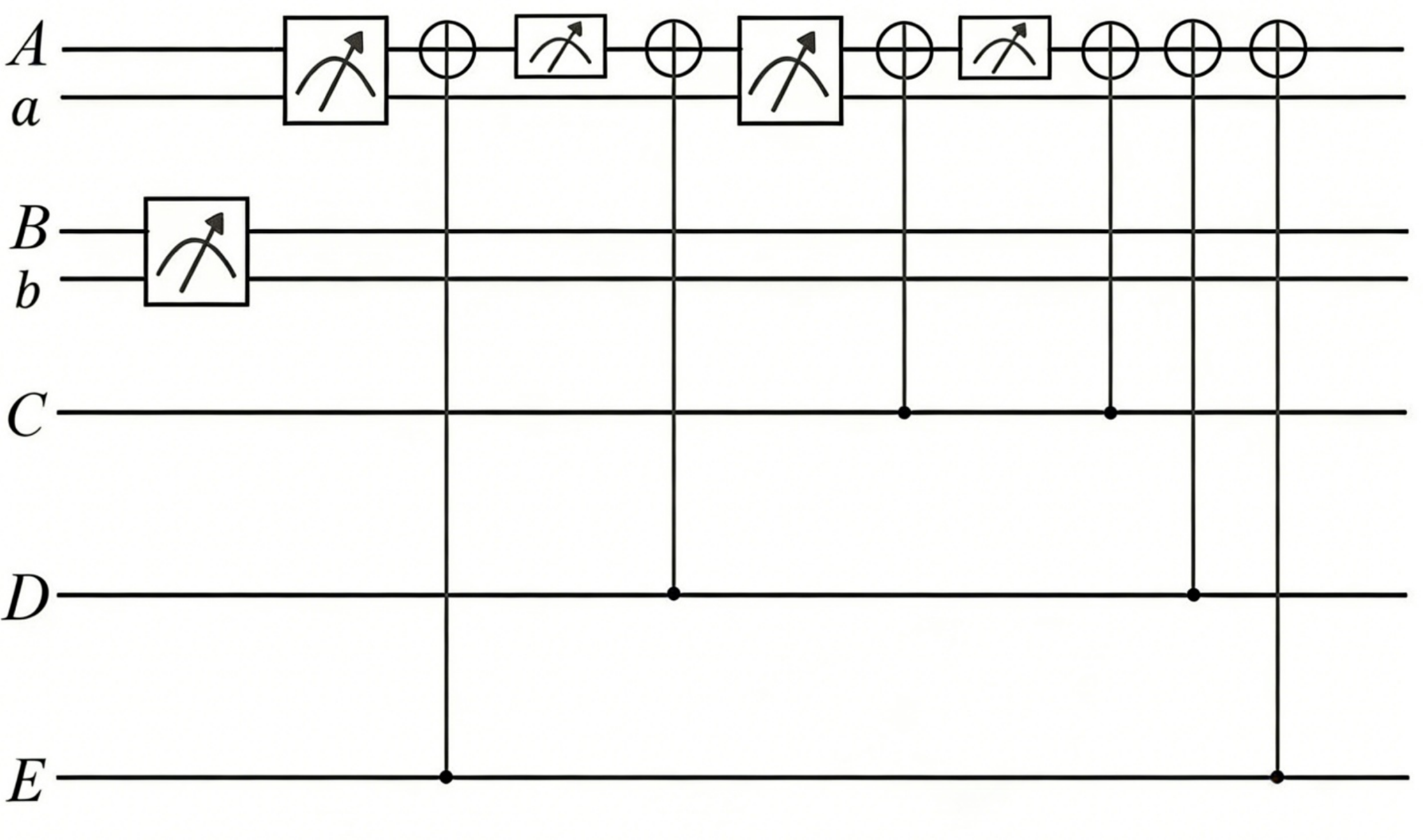}
\caption{Circuit diagram of the discrimination scheme for the genuinely nonlocal set~(\ref{set3.2}). Lines $A$ and $a$ correspond to Alice's main subsystem and ancilla, respectively; lines $B$ and $b$ correspond to Bob's main subsystem and ancilla, respectively. The lower lines $C$, $D$, and $E$ correspond to Charlie, Dave, and Eve. The meter denotes a measurement operation.\label{f3.2}}
\end{figure}

\begin{proof}
The detailed procedure for local discrimination of the set~(\ref{set3.2}) is shown in Fig.~\ref{f3.2}. The protocol proceeds as follows.

\medskip\noindent\textbf{Step 1.}
Alice and Bob share an EPR state $|\phi^{+}(2)\rangle_{ab}=|00\rangle_{ab}+|11\rangle_{ab}$. The initial state is
\[
|\psi\rangle_{ABCDE} \otimes |\phi^{+}(2)\rangle_{ab},
\]
where $a$ and $b$ are ancillary subsystems of Alice and Bob, respectively. Bob performs the measurement
\[
\mathcal{M}_1 \equiv \bigl\{
M_{1,1}=P[|0\rangle_B; |0\rangle_b] + P[|1\rangle_B; |1\rangle_b],\,
M_{1,2}=I - M_{1,1}
\bigr\}.
\]
If outcome $M_{1,1}$ occurs, each state becomes
\[
\begin{aligned}
&|\varphi_{\pm}\rangle_{0} =|00000\rangle|00\rangle \pm |11111\rangle|11\rangle, \\
&|\varphi_{\pm}\rangle_{1} =|00001\rangle|00\rangle \pm |11110\rangle|11\rangle, \\
&|\varphi_{\pm}\rangle_{2} =|00010\rangle|00\rangle \pm |11101\rangle|11\rangle, \\
&|\varphi_{\pm}\rangle_{3} =|00011\rangle|00\rangle \pm |11100\rangle|11\rangle, \\
&|\varphi_{\pm}\rangle_{4} =|00100\rangle|00\rangle \pm |11011\rangle|11\rangle, \\
&|\varphi_{\pm}\rangle_{5} =|00101\rangle|00\rangle \pm |11010\rangle|11\rangle, \\
&|\varphi_{\pm}\rangle_{6} =|00110\rangle|00\rangle \pm |11001\rangle|11\rangle, \\
&|\varphi_{\pm}\rangle_{7} =|00111\rangle|00\rangle \pm |11000\rangle|11\rangle, \\
&|\varphi_{\pm}\rangle_{8} =|01000\rangle|11\rangle \pm |10111\rangle|00\rangle, \\
&|\varphi_{\pm}\rangle_{9} =|01001\rangle|11\rangle \pm |10110\rangle|00\rangle, \\
&|\varphi_{\pm}\rangle_{10}=|01010\rangle|11\rangle \pm |10101\rangle|00\rangle, \\
&|\varphi_{\pm}\rangle_{11}=|01011\rangle|11\rangle \pm |10100\rangle|00\rangle, \\
&|\varphi_{\pm}\rangle_{12}=|01100\rangle|11\rangle \pm |10011\rangle|00\rangle, \\
&|\varphi_{\pm}\rangle_{13}=|01101\rangle|11\rangle \pm |10010\rangle|00\rangle, \\
&|\varphi_{\pm}\rangle_{14}=|01110\rangle|11\rangle \pm |10001\rangle|00\rangle, \\
&|\varphi_{\pm}\rangle_{15}=|01111\rangle|11\rangle \pm |10000\rangle|00\rangle,
\end{aligned}
\]
where $|00\rangle$ and $|11\rangle$ mean $|00\rangle_{ab}$ and $|11\rangle_{ab}$, respectively.

\medskip\noindent\textbf{Step 2.}
Alice makes a projective measurement
\[
\begin{aligned}
\mathcal{M}_2 \equiv \bigl\{\, &
M_{2,1} = P[|0\rangle_A; |0\rangle_a] + P[|1\rangle_A; |1\rangle_a], \\
& M_{2,2} = P[|0\rangle_A; |1\rangle_a] + P[|1\rangle_A; |0\rangle_a]
\bigr\}.
\end{aligned}
\]
If outcome $M_{2,1}$ occurs, the given state is one of
$\{|\varphi_{\pm}\rangle_{0},|\varphi_{\pm}\rangle_{1},|\varphi_{\pm}\rangle_{2},|\varphi_{\pm}\rangle_{3},
|\varphi_{\pm}\rangle_{4},|\varphi_{\pm}\rangle_{5},|\varphi_{\pm}\rangle_{6},|\varphi_{\pm}\rangle_{7}\}$;
if outcome $M_{2,2}$ occurs, the state is one of
$\{|\varphi_{\pm}\rangle_{8},|\varphi_{\pm}\rangle_{9},|\varphi_{\pm}\rangle_{10},|\varphi_{\pm}\rangle_{11},
|\varphi_{\pm}\rangle_{12},|\varphi_{\pm}\rangle_{13},|\varphi_{\pm}\rangle_{14},|\varphi_{\pm}\rangle_{15}\}$.
Denote the two cases by $P_1$ and $P_2$, respectively:
\[
\begin{aligned}
P_1 &= \{|\varphi_{\pm}\rangle_{0},|\varphi_{\pm}\rangle_{1},|\varphi_{\pm}\rangle_{2},|\varphi_{\pm}\rangle_{3},
|\varphi_{\pm}\rangle_{4},|\varphi_{\pm}\rangle_{5},|\varphi_{\pm}\rangle_{6},|\varphi_{\pm}\rangle_{7}\},\\
P_2 &= \{|\varphi_{\pm}\rangle_{8},|\varphi_{\pm}\rangle_{9},|\varphi_{\pm}\rangle_{10},|\varphi_{\pm}\rangle_{11},
|\varphi_{\pm}\rangle_{12},|\varphi_{\pm}\rangle_{13},|\varphi_{\pm}\rangle_{14},|\varphi_{\pm}\rangle_{15}\}.
\end{aligned}
\]

\medskip\noindent\textbf{Step 3.}
Following the idea in Theorem~\ref{thm:4qubit-ghz}, we apply a CNOT gate to further separate the states. Let the quantum states pass through a CNOT gate with Eve as the control qubit and Alice as the target qubit. Then the states become
\[
\begin{aligned}
P_1:\;
\begin{cases}
|\varphi_{\pm}\rangle_{0}=|00000\rangle|00\rangle \pm |01111\rangle|11\rangle,\\
|\varphi_{\pm}\rangle_{1}=|10001\rangle|00\rangle \pm |11110\rangle|11\rangle,\\
|\varphi_{\pm}\rangle_{2}=|00010\rangle|00\rangle \pm |01101\rangle|11\rangle,\\
|\varphi_{\pm}\rangle_{3}=|10011\rangle|00\rangle \pm |11100\rangle|11\rangle,\\
|\varphi_{\pm}\rangle_{4}=|00100\rangle|00\rangle \pm |01011\rangle|11\rangle,\\
|\varphi_{\pm}\rangle_{5}=|10101\rangle|00\rangle \pm |11010\rangle|11\rangle,\\
|\varphi_{\pm}\rangle_{6}=|00110\rangle|00\rangle \pm |01001\rangle|11\rangle,\\
|\varphi_{\pm}\rangle_{7}=|10111\rangle|00\rangle \pm |11000\rangle|11\rangle,
\end{cases}
\end{aligned}
\]
and
\[
\begin{aligned}
P_2:\;
\begin{cases}
|\varphi_{\pm}\rangle_{8}=|01000\rangle|11\rangle \pm |00111\rangle|00\rangle,\\
|\varphi_{\pm}\rangle_{9}=|11001\rangle|11\rangle \pm |10110\rangle|00\rangle,\\
|\varphi_{\pm}\rangle_{10}=|01010\rangle|11\rangle \pm |00101\rangle|00\rangle,\\
|\varphi_{\pm}\rangle_{11}=|11011\rangle|11\rangle \pm |10100\rangle|00\rangle,\\
|\varphi_{\pm}\rangle_{12}=|01100\rangle|11\rangle \pm |00011\rangle|00\rangle,\\
|\varphi_{\pm}\rangle_{13}=|11101\rangle|11\rangle \pm |10010\rangle|00\rangle,\\
|\varphi_{\pm}\rangle_{14}=|01110\rangle|11\rangle \pm |00001\rangle|00\rangle,\\
|\varphi_{\pm}\rangle_{15}=|11111\rangle|11\rangle \pm |10000\rangle|00\rangle.
\end{cases}
\end{aligned}
\]

\medskip\noindent\textbf{Step 4.}
Alice performs the measurement
\[
\mathcal{M}_3 \equiv \bigl\{\, M_{3,1} = |0\rangle_A\langle 0|,\; M_{3,2} = I - M_{3,1} \,\bigr\}.
\]
For the states in $P_1$ and $P_2$ given by Step~3, the outcomes correspond to
\[
P_1:\;
\begin{cases}
M_{3,1} \Rightarrow \{|\varphi_{\pm}\rangle_{0},|\varphi_{\pm}\rangle_{2},|\varphi_{\pm}\rangle_{4},|\varphi_{\pm}\rangle_{6}\},\\
M_{3,2} \Rightarrow \{|\varphi_{\pm}\rangle_{1},|\varphi_{\pm}\rangle_{3},|\varphi_{\pm}\rangle_{5},|\varphi_{\pm}\rangle_{7}\},
\end{cases}
\]
and
\[
P_2:\;
\begin{cases}
M_{3,1} \Rightarrow \{|\varphi_{\pm}\rangle_{8},|\varphi_{\pm}\rangle_{10},|\varphi_{\pm}\rangle_{12},|\varphi_{\pm}\rangle_{14}\},\\
M_{3,2} \Rightarrow \{|\varphi_{\pm}\rangle_{9},|\varphi_{\pm}\rangle_{11},|\varphi_{\pm}\rangle_{13},|\varphi_{\pm}\rangle_{15}\}.
\end{cases}
\]
Thus we obtain
\[
\begin{aligned}
P_{11} &= \{|\varphi_{\pm}\rangle_{0},|\varphi_{\pm}\rangle_{2},|\varphi_{\pm}\rangle_{4},|\varphi_{\pm}\rangle_{6}\},&
P_{12} &= \{|\varphi_{\pm}\rangle_{1},|\varphi_{\pm}\rangle_{3},|\varphi_{\pm}\rangle_{5},|\varphi_{\pm}\rangle_{7}\},\\
P_{21} &= \{|\varphi_{\pm}\rangle_{8},|\varphi_{\pm}\rangle_{10},|\varphi_{\pm}\rangle_{12},|\varphi_{\pm}\rangle_{14}\},&
P_{22} &= \{|\varphi_{\pm}\rangle_{9},|\varphi_{\pm}\rangle_{11},|\varphi_{\pm}\rangle_{13},|\varphi_{\pm}\rangle_{15}\}.
\end{aligned}
\]

\medskip\noindent\textbf{Step 5.}
Apply a CNOT gate again with Dave as the control qubit and Alice as the target qubit. Then the states become
\[
\begin{aligned}
P_{11}:\;
\begin{cases}
|\varphi_{\pm}\rangle_{0}=|00000\rangle|00\rangle \pm |11111\rangle|11\rangle,\\
|\varphi_{\pm}\rangle_{2}=|10010\rangle|00\rangle \pm |01101\rangle|11\rangle,\\
|\varphi_{\pm}\rangle_{4}=|00100\rangle|00\rangle \pm |11011\rangle|11\rangle,\\
|\varphi_{\pm}\rangle_{6}=|10110\rangle|00\rangle \pm |01001\rangle|11\rangle,
\end{cases}
\qquad
P_{12}:\;
\begin{cases}
|\varphi_{\pm}\rangle_{1}=|10001\rangle|00\rangle \pm |01110\rangle|11\rangle,\\
|\varphi_{\pm}\rangle_{3}=|00011\rangle|00\rangle \pm |11100\rangle|11\rangle,\\
|\varphi_{\pm}\rangle_{5}=|10101\rangle|00\rangle \pm |01010\rangle|11\rangle,\\
|\varphi_{\pm}\rangle_{7}=|00111\rangle|00\rangle \pm |11000\rangle|11\rangle,
\end{cases}
\end{aligned}
\]
and
\[
\begin{aligned}
P_{21}:\;
\begin{cases}
|\varphi_{\pm}\rangle_{8}=|01000\rangle|11\rangle \pm |10111\rangle|00\rangle,\\
|\varphi_{\pm}\rangle_{10}=|11010\rangle|11\rangle \pm |00101\rangle|00\rangle,\\
|\varphi_{\pm}\rangle_{12}=|01100\rangle|11\rangle \pm |10011\rangle|00\rangle,\\
|\varphi_{\pm}\rangle_{14}=|11110\rangle|11\rangle \pm |00001\rangle|00\rangle,
\end{cases}
\qquad
P_{22}:\;
\begin{cases}
|\varphi_{\pm}\rangle_{9}=|11001\rangle|11\rangle \pm |00110\rangle|00\rangle,\\
|\varphi_{\pm}\rangle_{11}=|01011\rangle|11\rangle \pm |10100\rangle|00\rangle,\\
|\varphi_{\pm}\rangle_{13}=|11101\rangle|11\rangle \pm |00010\rangle|00\rangle,\\
|\varphi_{\pm}\rangle_{15}=|01111\rangle|11\rangle \pm |10000\rangle|00\rangle.
\end{cases}
\end{aligned}
\]

\medskip\noindent\textbf{Step 6.}
Alice performs the measurement
\[
\mathcal{M}_4 \equiv \bigl\{
M_{4,1} = P[|0\rangle_A; |0\rangle_a] + P[|1\rangle_A; |1\rangle_a],\,
M_{4,2} = I - M_{4,1}
\bigr\}.
\]
For the parts $P_{11}$, $P_{12}$, $P_{21}$, and $P_{22}$ in Step~5, the outcomes correspond to
\[
\begin{aligned}
P_{11}:~
\begin{cases}
M_{4,1} \Rightarrow \{|\varphi_{\pm}\rangle_{0},|\varphi_{\pm}\rangle_{4}\},\\
M_{4,2} \Rightarrow \{|\varphi_{\pm}\rangle_{2},|\varphi_{\pm}\rangle_{6}\},
\end{cases}\qquad
P_{12}:~
\begin{cases}
M_{4,1} \Rightarrow \{|\varphi_{\pm}\rangle_{3},|\varphi_{\pm}\rangle_{7}\},\\
M_{4,2} \Rightarrow \{|\varphi_{\pm}\rangle_{1},|\varphi_{\pm}\rangle_{5}\},
\end{cases}
\end{aligned}
\]
and
\[
\begin{aligned}
P_{21}:~
\begin{cases}
M_{4,1} \Rightarrow \{|\varphi_{\pm}\rangle_{10},|\varphi_{\pm}\rangle_{14}\},\\
M_{4,2} \Rightarrow \{|\varphi_{\pm}\rangle_{8},|\varphi_{\pm}\rangle_{12}\},
\end{cases}\qquad
P_{22}:~
\begin{cases}
M_{4,1} \Rightarrow \{|\varphi_{\pm}\rangle_{9},|\varphi_{\pm}\rangle_{13}\},\\
M_{4,2} \Rightarrow \{|\varphi_{\pm}\rangle_{11},|\varphi_{\pm}\rangle_{15}\}.
\end{cases}
\end{aligned}
\]

As a result, the states in the set~(\ref{set3.2}) can be divided into eight parts:
\[
\setlength{\jot}{1pt}
\begin{aligned}
P_1 &=
\begin{cases}
P_{11} =
\begin{cases}
P_{111}=\{|\varphi_{\pm}\rangle_{0},|\varphi_{\pm}\rangle_{4}\},\\
P_{112}=\{|\varphi_{\pm}\rangle_{2},|\varphi_{\pm}\rangle_{6}\},
\end{cases}\\
P_{12} =
\begin{cases}
P_{121}=\{|\varphi_{\pm}\rangle_{3},|\varphi_{\pm}\rangle_{7}\},\\
P_{122}=\{|\varphi_{\pm}\rangle_{1},|\varphi_{\pm}\rangle_{5}\},
\end{cases}
\end{cases}
\\[2pt]
P_2 &=
\begin{cases}
P_{21}=
\begin{cases}
P_{211}=\{|\varphi_{\pm}\rangle_{10},|\varphi_{\pm}\rangle_{14}\},\\
P_{212}=\{|\varphi_{\pm}\rangle_{8},|\varphi_{\pm}\rangle_{12}\},
\end{cases}\\
P_{22}=
\begin{cases}
P_{221}=\{|\varphi_{\pm}\rangle_{9},|\varphi_{\pm}\rangle_{13}\},\\
P_{222}=\{|\varphi_{\pm}\rangle_{11},|\varphi_{\pm}\rangle_{15}\}.
\end{cases}
\end{cases}.
\end{aligned}
\]

\medskip\noindent\textbf{Step 7.}
Apply a CNOT gate again with Charlie as the control qubit and Alice as the target qubit. Then the states become
\[
\begin{aligned}
P_{111}:\;
\begin{cases}
|\varphi_{\pm}\rangle_{0}=|00000\rangle|00\rangle \pm |01111\rangle|11\rangle,\\
|\varphi_{\pm}\rangle_{4}=|10100\rangle|00\rangle \pm |11011\rangle|11\rangle,
\end{cases}\qquad
P_{112}:\;
\begin{cases}
|\varphi_{\pm}\rangle_{2}=|10010\rangle|00\rangle \pm |11101\rangle|11\rangle,\\
|\varphi_{\pm}\rangle_{6}=|00110\rangle|00\rangle \pm |01001\rangle|11\rangle,
\end{cases}\\
P_{121}:\;
\begin{cases}
|\varphi_{\pm}\rangle_{3}=|00011\rangle|00\rangle \pm |01100\rangle|11\rangle,\\
|\varphi_{\pm}\rangle_{7}=|10111\rangle|00\rangle \pm |11000\rangle|11\rangle,
\end{cases}\qquad
P_{122}:\;
\begin{cases}
|\varphi_{\pm}\rangle_{1}=|10001\rangle|00\rangle \pm |11110\rangle|11\rangle,\\
|\varphi_{\pm}\rangle_{5}=|00101\rangle|00\rangle \pm |01010\rangle|11\rangle,
\end{cases}\\
P_{211}:\;
\begin{cases}
|\varphi_{\pm}\rangle_{10}=|11010\rangle|11\rangle \pm |10101\rangle|00\rangle,\\
|\varphi_{\pm}\rangle_{14}=|01110\rangle|11\rangle \pm |00001\rangle|00\rangle,
\end{cases}\qquad
P_{212}:\;
\begin{cases}
|\varphi_{\pm}\rangle_{8}=|01000\rangle|11\rangle \pm |00111\rangle|00\rangle,\\
|\varphi_{\pm}\rangle_{12}=|11100\rangle|11\rangle \pm |10011\rangle|00\rangle,
\end{cases}\\
P_{221}:\;
\begin{cases}
|\varphi_{\pm}\rangle_{9}=|11001\rangle|11\rangle \pm |10110\rangle|00\rangle,\\
|\varphi_{\pm}\rangle_{13}=|01101\rangle|11\rangle \pm |00010\rangle|00\rangle,
\end{cases}\qquad
P_{222}:\;
\begin{cases}
|\varphi_{\pm}\rangle_{11}=|01011\rangle|11\rangle \pm |00100\rangle|00\rangle,\\
|\varphi_{\pm}\rangle_{15}=|11111\rangle|11\rangle \pm |10000\rangle|00\rangle.
\end{cases}
\end{aligned}
\]

\medskip\noindent\textbf{Step 8.}
Alice performs the measurement
\[
\mathcal{M}_5 \equiv \bigl\{\, M_{5,1} = |0\rangle_A\langle 0|,\; M_{5,2} = I - M_{5,1} \,\bigr\}.
\]
The outcomes identify the index as
\[
\setlength{\jot}{2pt}
\begin{aligned}
P_{111}:~
&\begin{cases}
M_{5,1}\Rightarrow |\varphi_{\pm}\rangle_{0},\\
M_{5,2}\Rightarrow |\varphi_{\pm}\rangle_{4},
\end{cases}\qquad
P_{112}:~
\begin{cases}
M_{5,1}\Rightarrow |\varphi_{\pm}\rangle_{6},\\
M_{5,2}\Rightarrow |\varphi_{\pm}\rangle_{2},
\end{cases}\\
P_{121}:~
&\begin{cases}
M_{5,1}\Rightarrow |\varphi_{\pm}\rangle_{3},\\
M_{5,2}\Rightarrow |\varphi_{\pm}\rangle_{7},
\end{cases}\qquad
P_{122}:~
\begin{cases}
M_{5,1}\Rightarrow |\varphi_{\pm}\rangle_{5},\\
M_{5,2}\Rightarrow |\varphi_{\pm}\rangle_{1},
\end{cases}\\
P_{211}:~
&\begin{cases}
M_{5,1}\Rightarrow |\varphi_{\pm}\rangle_{14},\\
M_{5,2}\Rightarrow |\varphi_{\pm}\rangle_{10},
\end{cases}\qquad
P_{212}:~
\begin{cases}
M_{5,1}\Rightarrow |\varphi_{\pm}\rangle_{8},\\
M_{5,2}\Rightarrow |\varphi_{\pm}\rangle_{12},
\end{cases}\\
P_{221}:~
&\begin{cases}
M_{5,1}\Rightarrow |\varphi_{\pm}\rangle_{13},\\
M_{5,2}\Rightarrow |\varphi_{\pm}\rangle_{9},
\end{cases}\qquad
P_{222}:~
\begin{cases}
M_{5,1}\Rightarrow |\varphi_{\pm}\rangle_{11},\\
M_{5,2}\Rightarrow |\varphi_{\pm}\rangle_{15}.
\end{cases}
\end{aligned}
\]

Obviously, each remaining subset consists of only the pair $\{|\varphi_{+}\rangle_i,|\varphi_{-}\rangle_i\}$, which is locally distinguishable. If outcome $M_{1,2}$ occurs in Step~1, a similar method achieves local discrimination. Therefore, the GHZ basis~(\ref{set3.2}) in the five-qubit system can be perfectly distinguished.

In this scenario, the original states have been transformed due to the three CNOT gates. To recover the original states, it suffices to apply the same three CNOT gates again. As shown in Fig.~\ref{f3.2}, let Alice be the target qubit while Charlie, Dave, and Eve serve successively as the control qubits.
\end{proof}

In this section, we have proposed a new method incorporating CNOT gates to achieve the local discrimination of multipartite orthogonal entangled sets exhibiting nonlocality, such as the sets~(\ref{set3.1}) and~(\ref{set3.2}). For both the four- and five-qubit GHZ bases, the discrimination requires only $1$ ebit of total entanglement, which is significantly less than the $2$--$3$ ebits required in existing protocols, thereby improving resource utilization. Note that the control qubit may vary from one CNOT application to another; the original state can be recovered by reapplying the corresponding CNOT gates.

\section{Discrimination of Strongly Nonlocal OPSs in Four-Partite Systems}\label{Q4}

In the subsequent investigation, we focus on the orthogonal product sets (OPSs) constructed in Ref.~\cite{Ref36}. That reference presents two structures of OPSs with strong quantum nonlocality in four-partite systems: asymmetric and symmetric. In this section, using different types of entangled states, we propose several local discrimination protocols for the asymmetric and symmetric OPSs, respectively. Hereafter, unless otherwise specified, we use ``$A$'' to denote the first subsystem, ``$B$'' the second subsystem, ``$C$'' the third subsystem, and so on.
\subsection{Asymmetric Case}
In the system $\mathcal{C}^{d_1} \otimes \mathcal{C}^{d_2} \otimes \mathcal{C}^{d_3} \otimes \mathcal{C}^{d_4}$ ($d_k\geq 3$ for $k=1,2,3,4$), the strongly nonlocal OPS with asymmetric structure \cite{Ref36} is
\begin{equation}\label{set4.1}
\begin{aligned}
&H_1 = \{|0\rangle_1|0\rangle_2|\beta_j\rangle_3|\alpha_i\rangle_4\}_{i,j}, \\
&H_2 = \{|0\rangle_1|0\rangle_2|d_3'\rangle_3|\beta_j\rangle_4\}_j, \\
&H_3 = \{|\beta_j\rangle_1|\alpha_i\rangle_2|d_3'\rangle_3|d_4'\rangle_4\}_{i,j}, \\
&H_4 = \{|d_1'\rangle_1|\beta_j\rangle_2|d_3'\rangle_3|d_4'\rangle_4\}_j, \\
&H_5 = \{|d_1'\rangle_1|d_2'\rangle_2|\gamma_m\rangle_3|\alpha_i\rangle_4\}_{i,m}, \\
&H_6 = \{|d_1'\rangle_1|d_2'\rangle_2|0\rangle_3|\gamma_m\rangle_4\}_m, \\
&H_7 = \{|\gamma_m\rangle_1|\alpha_i\rangle_2|0\rangle_3|0\rangle_4\}_{i,m}, \\
&H_8 = \{|0\rangle_1|\gamma_m\rangle_2|0\rangle_3|0\rangle_4\}_m, \\
&H_9 = \{|0\rangle_1|\kappa_{I}\rangle_2|0\rangle_3|\gamma_m\rangle_4\}_{I,m}, \\
&H_{10} = \{|d_1'\rangle_1|\kappa_{I}\rangle_2|d_3'\rangle_3|\beta_j\rangle_4\}_{j,I}, \\
&H_{11} = \{|\kappa_{I}\rangle_1|d_2'\rangle_2|\beta_j\rangle_3|d_4'\rangle_4\}_{j,I}, \\
&H_{12} = \{|\kappa_{I}\rangle_1|0\rangle_2|\gamma_m\rangle_3|0\rangle_4\}_{I,m}, \\
&H_{13} = \{|0\rangle_1|\gamma_m\rangle_2|\kappa_{I}\rangle_3|d_4'\rangle_4\}_{I,m}, \\
&H_{14} = \{|d_1'\rangle_1|\beta_j\rangle_2|\kappa_{I}\rangle_3|0\rangle_4\}_{j,I}, \\
&H_{15} = \{|\kappa_{I}\rangle_1|d_2'\rangle_2|d_3'\rangle_3|\beta_j\rangle_4\}_{j,I}, \\
&H_{16} = \{|\kappa_{I}\rangle_1|0\rangle_2|0\rangle_3|\gamma_m\rangle_4\}_{I,m}, \\
&H_{17} = \{|0\rangle_1|d_2'\rangle_2|0 \pm d_3'\rangle_3|\kappa_{I}\rangle_4\}_I, \\
&H_{18} = \{|d_1'\rangle_1|0\rangle_2|0 \pm d_3'\rangle_3|\kappa_{I}\rangle_4\}_I.
\end{aligned}
\end{equation}
Here
\[
|\alpha_i\rangle_k=\sum_{u=0}^{d_k-1}\omega_{d_k}^{iu}|u\rangle,\qquad
|\beta_j\rangle_k=\sum_{u=0}^{d_k-2}\omega_{d_k-1}^{ju}|u\rangle,
\]
\[
|\gamma_m\rangle_k=\sum_{u=0}^{d_k-2}\omega_{d_k-1}^{mu}|u+1\rangle,\qquad
|\kappa_I\rangle_k=\sum_{u=0}^{d_k-3}\omega_{d_k-2}^{Iu}|u+1\rangle,
\]
and \(d_k'=d_k-1\), where \(i\in\mathbb{Z}_{d_k}\), \(j,m\in\mathbb{Z}_{d_k-1}\), \(I\in\mathbb{Z}_{d_k-2}\), and \(k\in\{1,2,3,4\}\).

\begin{theorem}\label{thm:asym-ops-1}
The strongly nonlocal set $\bigcup_{i=1}^{18}H_{i}$ in Eq.~(\ref{set4.1}) can be locally distinguished using the entanglement resource configuration
$\{(1,|\phi^{+}(2)\rangle_{AC});(1,|\phi^{+}(d_{4})\rangle_{CD})\}$.
\end{theorem}

\begin{theorem}\label{thm:asym-ops-2}
The set $\bigcup_{i=1}^{18}H_{i}$ in Eq.~(\ref{set4.1}) can be locally distinguished using the resource configuration
$\{(1,|\phi^{+}(2)\rangle_{AD});(1,|\phi^{+}(3)\rangle_{CD})\}$.
\end{theorem}

\begin{theorem}\label{thm:asym-ops-3}
The entanglement resource configuration
$\{(1,|\phi^{+}(2)\rangle_{AD});(1,|\phi^{+}(2)\rangle_{CD});(\frac{r}{s},|\phi^{+}(2)\rangle_{BD})\}$
is sufficient for local discrimination of the strongly nonlocal set~(\ref{set4.1}), where
$s=2\big(\sum_{1\leq i<j\leq 4}d_id_j-2\sum_{i=1}^{4}d_i-d_2-d_3+2\big)$
and
$r=\sum_{1\leq i<j\leq 4}d_id_j-4\sum_{i=1}^{4}d_i+d_1(d_3+d_4)+d_2-d_3+6$.
\end{theorem}

The proofs of Theorems~\ref{thm:asym-ops-1}, \ref{thm:asym-ops-2}, and~\ref{thm:asym-ops-3} are given in Appendices~\ref{T3}, \ref{T4}, and~\ref{T5}, respectively. In the protocol of Theorem~\ref{thm:asym-ops-1}, the total entanglement consumption is \(1+\log_2 d_4\) ebits. In the protocols of Theorems~\ref{thm:asym-ops-2} and~\ref{thm:asym-ops-3}, the total consumption is \(1+\log_2 3\) ebits and \(2+\frac{r}{s}\) ebits, respectively.

For the asymmetric OPS~(\ref{set4.1}), we have proposed three local discrimination protocols. The protocol in Theorem~\ref{thm:asym-ops-1} relies on quantum teleportation to transfer subsystem \(D\) to party \(C\), after which entanglement is established between parties \(A\) and \(C\) to complete the discrimination. Consequently, most of the entanglement is consumed by the teleportation step, and the total consumption increases with the dimension \(d_4\). To reduce resource consumption, the protocols in Theorems~\ref{thm:asym-ops-2} and~\ref{thm:asym-ops-3} avoid teleportation and instead attempt to complete the discrimination by establishing entanglement only between suitable pairs of parties. When \(d_4>3\), both Theorem~\ref{thm:asym-ops-2} and Theorem~\ref{thm:asym-ops-3} consume strictly less entanglement than Theorem~\ref{thm:asym-ops-1}, since \(\frac{r}{s}<1\).

Comparing Theorems~\ref{thm:asym-ops-2} and~\ref{thm:asym-ops-3}, we find that the difference in entanglement consumption is relatively small. The entanglement required in Theorem~\ref{thm:asym-ops-3} depends explicitly on the dimensions \(d_1,d_2,d_3,d_4\). As shown in Fig.~\ref{f4.1}, for \(d_1=d_2=4\), the total entanglement consumed in Theorem~\ref{thm:asym-ops-3} is larger than that of Theorem~\ref{thm:asym-ops-2} when \(d_3\) and \(d_4\) are relatively small, whereas it becomes smaller when these dimensions are sufficiently large.

\begin{figure}[h]
\centering
\includegraphics[width=0.44\textwidth]{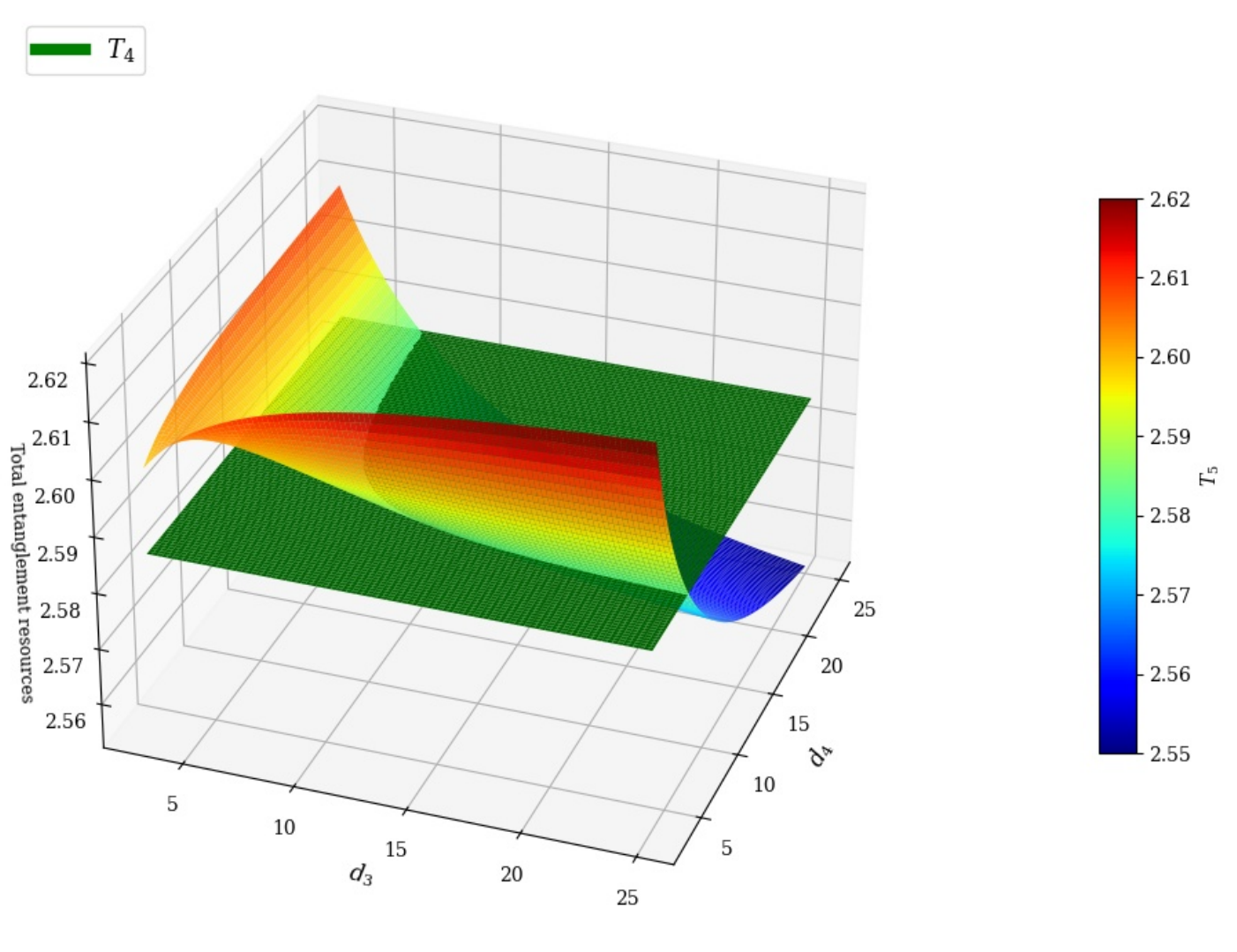}
\caption{Comparison of total entanglement consumed by Theorem~\ref{thm:asym-ops-2} and Theorem~\ref{thm:asym-ops-3} for $d_1=4$ and $d_2=4$. The green plane $T_4$ and the gradient surface $T_5$ correspond to Theorem~\ref{thm:asym-ops-2} and Theorem~\ref{thm:asym-ops-3}, respectively.\label{f4.1}}
\end{figure}

\subsection{Symmetric Case}

In this subsection, we consider a class of strongly nonlocal OPSs with symmetric structure. In \(\mathbb{C}^{d_1}\otimes\mathbb{C}^{d_2}\otimes\mathbb{C}^{d_3}\otimes\mathbb{C}^{d_4}\) (\(d_k\ge 3\) for \(k=1,2,3,4\)), the OPS constructed in Ref.~\cite{Ref36} is given by
\begin{equation}\label{set4.2}
\begin{aligned}
&H_{1,1} = \{|\kappa_{I}\rangle_1 |\kappa_{I}\rangle_2 |0\rangle_3 |\gamma_m\rangle_4\}_{m,\,I|_1,\,I|_2}, \\
&H_{1,2} = \{|\kappa_{I}\rangle_1 |0\rangle_2 |\gamma_m\rangle_3 |\kappa_{I}\rangle_4\}_{m,\,I|_1,\,I|_4}, \\
&H_{1,3} = \{|0\rangle_1|\gamma_m\rangle_2 |\kappa_{I}\rangle_3 |\kappa_{I}\rangle_4\}_{m,\,I|_3,\,I|_4}, \\
&H_{1,4} = \{|\gamma_m\rangle_1 |\kappa_{I}\rangle_2 |\kappa_{I}\rangle_3 |0\rangle_4\}_{m,\,I|_2,\,I|_3}, \\
&H_{2,1} = \{|d_1'\rangle_1 |0\rangle_2 |0\rangle_3 |\beta_j\rangle_4\}_j, \\
&H_{2,2} = \{|0\rangle_1 |0\rangle_2 |\beta_j\rangle_3 |d_4'\rangle_4\}_j, \\
&H_{2,3} = \{|0\rangle_1 |\beta_j\rangle_2 |d_3'\rangle_3 |0\rangle_4\}_j, \\
&H_{2,4} = \{|\beta_j\rangle_1 |d_2'\rangle_2 |0\rangle_3 |0\rangle_4\}_j, \\
&H_{3,1} = \{|0\rangle_1 |0\rangle_2 |d_3'\rangle_3 |\gamma_m\rangle_4\}_m, \\
&H_{3,2} = \{|0\rangle_1 |d_2'\rangle_2 |\gamma_m\rangle_3 |0\rangle_4\}_m, \\
&H_{3,3} = \{|d_1'\rangle_1 |\gamma_m\rangle_2 |0\rangle_3 |0\rangle_4\}_m, \\
&H_{3,4} = \{|\gamma_m\rangle_1 |0\rangle_2 |0\rangle_3 |d_4'\rangle_4\}_m, \\
&H_{4,1} = \{|\kappa_{I}\rangle_1 |d_2'\rangle_2 |d_3'\rangle_3 |0\pm d_4'\rangle_4\}_I, \\
&H_{4,2} = \{|d_1'\rangle_1 |d_2'\rangle_2 |0\pm d_3'\rangle_3 |\kappa_{I}\rangle_4\}_I, \\
&H_{4,3} = \{|d_1'\rangle_1 |0\pm d_2'\rangle_2 |\kappa_{I}\rangle_3 |d_4'\rangle_4\}_I, \\
&H_{4,4} = \{|0\pm d_1'\rangle_1 |\kappa_{I}\rangle_2 |d_3'\rangle_3 |d_4'\rangle_4\}_I, \\
&H_{5,1} = \{|d_1'\rangle_1 |d_2'\rangle_2 |\kappa_{I}\rangle_3 |\beta_j\rangle_4\}_{j,I}, \\
&H_{5,2} = \{|d_1'\rangle_1 |\kappa_{I}\rangle_2 |\beta_j\rangle_3 |d_4'\rangle_4\}_{j,I}, \\
&H_{5,3} = \{|\kappa_{I}\rangle_1 |\beta_j\rangle_2 |d_3'\rangle_3 |d_4'\rangle_4\}_{j,I}, \\
&H_{5,4} = \{|\beta_j\rangle_1 |d_2'\rangle_2 |d_3'\rangle_3 |\kappa_{I}\rangle_4\}_{j,I}, \\
&H_{6,1} = \{|0\rangle_1 |d_2'\rangle_2 |d_3'\rangle_3 |d_4'\rangle_4\}, \\
&H_{6,2} = \{|d_1'\rangle_1 |d_2'\rangle_2 |d_3'\rangle_3 |0\rangle_4\}, \\
&H_{6,3} = \{|d_1'\rangle_1 |d_2'\rangle_2 |0\rangle_3 |d_4'\rangle_4\}, \\
&H_{6,4} = \{|d_1'\rangle_1 |0\rangle_2 |d_3'\rangle_3 |d_4'\rangle_4\}, \\
&H_{7,1} = \{|0\rangle_1 |\kappa_{I}\rangle_2 |d_3'\rangle_3 |\kappa_{I}\rangle_4\}_{I|_2,\,I|_4}, \\
&H_{7,2} = \{|\kappa_{I}\rangle_1 |d_2'\rangle_2 |\kappa_{I}\rangle_3 |0\rangle_4\}_{I|_1,\,I|_3}, \\
&H_{7,3} = \{|d_1'\rangle_1 |\kappa_{I}\rangle_2 |0\rangle_3 |\kappa_{I}\rangle_4\}_{I|_2,\,I|_4}, \\
&H_{7,4} = \{|\kappa_{I}\rangle_1 |0\rangle_2 |\kappa_{I}\rangle_3 |d_4'\rangle_4\}_{I|_1,\,I|_3}, \\
&H_{8,1} = \{|d_1'\rangle_1 |\kappa_{I}\rangle_2 |d_3'\rangle_3 |\kappa_{I}\rangle_4\}_{I|_2,\,I|_4}, \\
&H_{8,2} = \{|\kappa_{I}\rangle_1 |d_2'\rangle_2 |\kappa_{I}\rangle_3 |d_4'\rangle_4\}_{I|_1,\,I|_3}.
\end{aligned}
\end{equation}
Here $|\beta_j\rangle_k = \sum_{u=0}^{d_k-2} \omega_{d_k-1}^{ju} |u\rangle$,
$|\gamma_m\rangle_k = \sum_{u=0}^{d_k-2} \omega_{d_k-1}^{mu} |u+1\rangle$,
$|\kappa_I\rangle_k = \sum_{u=0}^{d_k-3} \omega_{d_k-2}^{Iu} |u+1\rangle$,
and $d_k' = d_k - 1$ for $j,m \in \mathbb{Z}_{d_k-1}$, $I \in \mathbb{Z}_{d_k-2}$, and $k\in\{1,2,3,4\}$.

\begin{theorem}\label{thm:sym-ops-1}
The strongly nonlocal set given by Eq.~(\ref{set4.2}) can be locally distinguished using the entanglement resource configuration
\[
\{(1,|\phi^{+}(3)\rangle_{AC});(1,|\phi^{+}(d_{4})\rangle_{CD})\}.
\]
\end{theorem}

\begin{theorem}\label{thm:sym-ops-2}
The entanglement resource configuration
\[
\{(1,|\phi^{+}(2)\rangle_{AC});(1,|\phi^{+}(3)\rangle_{AD});(\tfrac{r'}{s'},|\phi^{+}(2)\rangle_{BD})\}
\]
is sufficient for the local discrimination of the strongly nonlocal set~(\ref{set4.2}), where
\[
s'=-2\sum_{1\leq i<j\leq 4}d_id_j+3\sum_{i=1}^{4}d_i+d_1d_3(1+d_2+d_4)+d_2d_4(1+d_1+d_3)-4
\]
and
\[
r'=s'-(d_1+d_2+d_3)+3.
\]
\end{theorem}

The detailed proofs of Theorems~\ref{thm:sym-ops-1} and~\ref{thm:sym-ops-2} are given in Appendices~\ref{T6} and~\ref{T7}, respectively. In the protocol of Theorem~\ref{thm:sym-ops-1}, one use of quantum teleportation is required, leading to a total entanglement consumption of \(\log_2 3+\log_2 d_4\) ebits. By contrast, the protocol of Theorem~\ref{thm:sym-ops-2} avoids teleportation and consumes \(1+\frac{r'}{s'}+\log_2 3\) ebits. Since \(\frac{r'}{s'}<1\), we have \(1+\frac{r'}{s'}<\log_2 d_4\) whenever \(d_4>4\). Hence, for \(d_4>4\), Theorem~\ref{thm:sym-ops-2} is more entanglement-efficient than Theorem~\ref{thm:sym-ops-1}. Moreover, compared with Theorem~\ref{thm:sym-ops-1}, Theorem~\ref{thm:sym-ops-2} only requires maximally entangled states of lower local dimension, which are generally easier to realize experimentally. Therefore, the scheme in Theorem~\ref{thm:sym-ops-2} is preferable both in terms of entanglement cost and practical implementability.

The entanglement-assisted local discrimination protocols proposed in Theorems~\ref{thm:sym-ops-1} and~\ref{thm:sym-ops-2} follow the same general construction principles as those in Theorems~\ref{thm:asym-ops-1} and~\ref{thm:asym-ops-2}, respectively. Comparing the two OPSs~(\ref{set4.1}) and~(\ref{set4.2}), we also find that the set with symmetric structure appears to be harder to distinguish perfectly by LOCC, owing to its invariance under cyclic permutations of the parties.
\section{Discrimination of Symmetric Strongly Nonlocal OPSs in Five-Partite Systems}\label{Q5}

Unlike the previous sections, in this section we explore local discrimination protocols by increasing the number of auxiliary subsystems involved in a shared maximally entangled resource. In Ref.~\cite{Ref36}, strongly nonlocal OPSs with symmetric structure were further studied in \(n\)-partite systems. Taking the five-partite case as an example, we propose several local discrimination protocols using only EPR states, GHZ states, and \(|F\rangle\) states, respectively. In the quantum system
\[
\mathbb{C}^{d_1}\otimes\mathbb{C}^{d_2}\otimes\mathbb{C}^{d_3}\otimes\mathbb{C}^{d_4}\otimes\mathbb{C}^{d_5},
\qquad d_k\ge 3\ \text{for}\ k=1,2,\ldots,5,
\]
the strongly nonlocal OPS is given by
\begin{equation}\label{set5}
\begin{aligned}
&H_1 = \{|\alpha_i\rangle_1 |\alpha_i\rangle_2 |p\rangle_3 |0\rangle_4 |0\rangle_5\}, \\
&H_2 = \{|\alpha_i\rangle_1 |p\rangle_2 |0\rangle_3 |0\rangle_4 |\alpha_i\rangle_5\}, \\
&H_3 = \{|p\rangle_1 |0\rangle_2 |0\rangle_3 |\alpha_i\rangle_4 |\alpha_i\rangle_5\}, \\
&H_4 = \{|0\rangle_1 |0\rangle_2 |\alpha_i\rangle_3 |\alpha_i\rangle_4 |p\rangle_5\}, \\
&H_5 = \{|0\rangle_1 |\alpha_i\rangle_2 |\alpha_i\rangle_3 |p\rangle_4 |0\rangle_5\}, \\
&H_6 = \{|\gamma_{m}\rangle_1 |0\rangle_2 |p\rangle_3 |0\rangle_4 |1\rangle_5\}, \\
&H_7 = \{|0\rangle_1 |p\rangle_2 |0\rangle_3 |1\rangle_4 |\gamma_{m}\rangle_5\}, \\
&H_8 = \{|p\rangle_1 |0\rangle_2 |1\rangle_3 |\gamma_{m}\rangle_4 |0\rangle_5\}, \\
&H_9 = \{|0\rangle_1 |1\rangle_2 |\gamma_{m}\rangle_3 |0\rangle_4 |p\rangle_5\}, \\
&H_{10} = \{|1\rangle_1 |\gamma_{m}\rangle_2 |0\rangle_3 |p\rangle_4 |0\rangle_5\}.
\end{aligned}
\end{equation}
Here
\[
|\alpha_i\rangle_k=\sum_{u=0}^{d_k-1}\omega_{d_k}^{iu}|u\rangle,\qquad
|\gamma_m\rangle_k=\sum_{u=0}^{d_k-2}\omega_{d_k-1}^{mu}|u+1\rangle,
\]
and
\[
|p\rangle_k\in\{|1\rangle_k,|2\rangle_k,\ldots,|d_k-1\rangle_k\},
\]
where \(i\in\mathbb{Z}_{d_k}\), \(m\in\mathbb{Z}_{d_k-1}\), and \(k\in\{1,2,\ldots,5\}\).

\begin{theorem}\label{thm:fivepartite-1}
The strongly nonlocal set \(\bigcup_{i=1}^{10}H_i\) given by Eq.~(\ref{set5}) can be locally distinguished using the entanglement resource configuration
\[
\{(1,|\phi^{+}(2)\rangle_{AE});(1,|\phi^{+}(2)\rangle_{BE});(1,|\phi^{+}(2)\rangle_{CE});(1,|\phi^{+}(2)\rangle_{DE})\}.
\]
\end{theorem}

\begin{theorem}\label{thm:fivepartite-2}
The set \(\bigcup_{i=1}^{10}H_i\) given by Eq.~(\ref{set5}) can be locally distinguished using the entanglement resource configuration
\[
\{(1,|G\rangle_{ADE});(1,|G\rangle_{BDE});(1,|G\rangle_{CDE})\}.
\]
\end{theorem}

\begin{theorem}\label{thm:fivepartite-3}
The entanglement resource configuration
\[
\{(1,|F\rangle_{ACDE});(1,|F\rangle_{BCDE})\}
\]
is sufficient for the local discrimination of the strongly nonlocal set~(\ref{set5}).
\end{theorem}

When subsystems \(A\), \(B\), \(C\), and \(D\) each share one EPR pair with subsystem \(E\), the discrimination task can be completed in two steps. In the second step, the measurement on subsystem \(E\) simultaneously imposes five filtering conditions on subsystem \(E\) and the auxiliary subsystems, namely
\[
\{|i\rangle_E;\ |i\rangle_{e_1};\ |i\rangle_{e_2};\ |i\rangle_{e_3};\ |i\rangle_{e_4}\}.
\]
Thus, a single orthogonality-preserving projective measurement on party \(E\) suffices to identify all subsets in Eq.~(\ref{set5}). See Appendix~\ref{T8} for details.

When only GHZ states are shared, perfect discrimination can be achieved provided that subsystems \(A\), \(B\), and \(C\) each share one GHZ state with subsystems \(D\) and \(E\). In this case, the measurements on parties \(D\) and \(E\) each impose four filtering conditions simultaneously:
\[
\{|i\rangle_X;\ |i\rangle_{x_1};\ |i\rangle_{x_2};\ |i\rangle_{x_3}\},
\]
where \((X,x)\) denotes either \((D,v)\) or \((E,e)\). When only \(|F\rangle\) states are shared, subsystems \(A\) and \(B\) each share one \(|F\rangle\) state with subsystems \(C\), \(D\), and \(E\), respectively. In that case, three filtering conditions can be imposed through measurements on parties \(C\), \(D\), and \(E\). The detailed proofs of Theorems~\ref{thm:fivepartite-2} and~\ref{thm:fivepartite-3} are given in Appendices~\ref{T9} and~\ref{T10}, respectively.

As the number of auxiliary subsystems increases, the filtering power of local measurements can be enhanced. Compared with Theorem~\ref{thm:fivepartite-1}, although the measurements on party \(E\) in Theorems~\ref{thm:fivepartite-2} and~\ref{thm:fivepartite-3} are individually weaker, the measurements on the other parties become stronger. Therefore, for the strongly nonlocal OPS~(\ref{set5}), the discrimination task can be completed under LOCC using four EPR pairs, three GHZ states, or two \(|F\rangle\) states, respectively. The choice among Theorems~\ref{thm:fivepartite-1}--\ref{thm:fivepartite-3} may be guided by implementation constraints.

\section{Conclusion}\label{Q6}
In this paper, we have studied efficient entanglement-assisted discrimination protocols for nonlocal sets of orthogonal states. In particular, for genuinely nonlocal GHZ bases in four- and five-qubit systems, we showed that by incorporating CNOT gates into the discrimination process, only one shared EPR pair is sufficient to complete the task. In \(\mathbb{C}^{d_1}\otimes\mathbb{C}^{d_2}\otimes\mathbb{C}^{d_3}\otimes\mathbb{C}^{d_4}\) (\(d_k\ge 3\), \(k=1,2,3,4\)), for asymmetric and symmetric OPSs with strong quantum nonlocality, we provided several efficient protocols based on different entanglement resources and methods, including quantum teleportation and the use of shared EPR pairs or \(3\otimes 3\) maximally entangled states. For the protocols considered here, those without quantum teleportation tend to be more efficient in terms of total entanglement consumption. In addition, for strongly nonlocal OPSs in five-partite systems, we presented three protocols based on four EPR pairs, three GHZ states, and two \(|F\rangle\) states, respectively. Our analysis shows that the filtering power of measurements on a given party generally increases as the number of auxiliary subsystems increases. How to exploit auxiliary subsystems more effectively remains an important problem in the design of improved entanglement-assisted discrimination protocols.
\begin{appendix}

\section{The proof of Theorem 3}\label{T3}

Suppose that the whole quantum system is shared among Alice, Bob, Charlie, and Dave. First, subsystem \(D\) is teleported to Charlie by using the entanglement resource \(\left|\phi^{+}(d_4)\right\rangle\), and the resulting joint subsystem is denoted by \(\widetilde{C}\). To distinguish the states locally, Alice and Charlie further share a maximally entangled state \(\vert \phi^{+}(2)\rangle_{a\tilde{c}}\). The initial state is
\[
\begin{aligned}
\left|\psi\right\rangle_{AB\widetilde{C}} \otimes \left|\phi^{+}(2)\right\rangle_{a\tilde{c}}.
\end{aligned}
\]
In fact, it is sufficient to distinguish the subsets \(H_r\) \((r=1,2,\ldots,18)\), because each subset is LOCC distinguishable. The protocol proceeds as follows.

\textbf{Step 1.} Alice performs the measurement
\[
\begin{aligned}
\mathcal{M}_1 \equiv \bigl\{&
M_{1,1}=P\bigl[|0\rangle_A;|0\rangle_a\bigr]
+P\bigl[(|1\rangle,\ldots,|d_1'\rangle)_A;|1\rangle_a\bigr],\\
& M_{1,2}=I-M_{1,1}\bigr\}.
\end{aligned}
\]
Suppose that the outcome corresponding to \(M_{1,1}\) occurs. Then
\[
\begin{aligned}
&H_1 \rightarrow |0\rangle|0\rangle|\beta_j\circ\alpha_i\rangle|00\rangle,\\
&H_2 \rightarrow |0\rangle|0\rangle|d_3'\circ\beta_j\rangle|00\rangle,\\
&H_3 \rightarrow |0\rangle|\alpha_i\rangle|d_3'\circ d_4'\rangle|00\rangle
+|\beta_j'\rangle|\alpha_i\rangle|d_3'\circ d_4'\rangle|11\rangle,\\
&H_4 \rightarrow |d_1'\rangle|\beta_j\rangle|d_3'\circ d_4'\rangle|11\rangle,\\
&H_5 \rightarrow |d_1'\rangle|d_2'\rangle|\gamma_m\circ\alpha_i\rangle|11\rangle,\\
&H_6 \rightarrow |d_1'\rangle|d_2'\rangle|0\circ\gamma_m\rangle|11\rangle,\\
&H_7 \rightarrow |\gamma_m\rangle|\alpha_i\rangle|0\circ0\rangle|11\rangle,\\
&H_8 \rightarrow |0\rangle|\gamma_m\rangle|0\circ0\rangle|00\rangle,\\
&H_9 \rightarrow |0\rangle|\kappa_I\rangle|0\circ\gamma_m\rangle|00\rangle,\\
&H_{10} \rightarrow |d_1'\rangle|\kappa_I\rangle|d_3'\circ\beta_j\rangle|11\rangle,\\
&H_{11} \rightarrow |\kappa_I\rangle|d_2'\rangle|\beta_j\circ d_4'\rangle|11\rangle,\\
&H_{12} \rightarrow |\kappa_I\rangle|0\rangle|\gamma_m\circ0\rangle|11\rangle,\\
&H_{13} \rightarrow |0\rangle|\gamma_m\rangle|\kappa_I\circ d_4'\rangle|00\rangle,\\
&H_{14} \rightarrow |d_1'\rangle|\beta_j\rangle|\kappa_I\circ0\rangle|11\rangle,\\
&H_{15} \rightarrow |\kappa_I\rangle|d_2'\rangle|d_3'\circ\beta_j\rangle|11\rangle,\\
&H_{16} \rightarrow |\kappa_I\rangle|0\rangle|0\circ\gamma_m\rangle|11\rangle,\\
&H_{17} \rightarrow |0\rangle|d_2'\rangle|(0\pm d_3')\circ\kappa_I\rangle|00\rangle,\\
&H_{18} \rightarrow |d_1'\rangle|0\rangle|(0\pm d_3')\circ\kappa_I\rangle|11\rangle.
\end{aligned}
\]
Here
\[
|\beta_j'\rangle=\sum_{u=1}^{d_1-2}\omega_{d_1-1}^{ju}|u\rangle.
\]
Moreover, \(|00\rangle\) and \(|11\rangle\) stand for \(|00\rangle_{a\tilde{c}}\) and \(|11\rangle_{a\tilde{c}}\), respectively. The symbol ``\(\circ\)'' denotes the union of parties. For example,
\[
|\psi_1\circ\psi_2\rangle_{\widetilde{C}}=|\psi_1\rangle_C|\psi_2\rangle_D
\]
for any two states \(|\psi_1\rangle_C\) and \(|\psi_2\rangle_D\). Likewise,
\[
|(0,\ldots,d_3')\circ(0,\ldots,d_4')\rangle_{\widetilde{C}}
\]
denotes the ordered list
\[
(|00\rangle,\ldots,|0d_4'\rangle,|10\rangle,\ldots,|d_3'd_4'\rangle)_{\widetilde{C}}.
\]

\textbf{Step 2.} \(\widetilde{C}\) performs the measurement
\[
\begin{aligned}
\mathcal{M}_2\equiv \bigl\{
M_{2,1}=P[|00\rangle_{\widetilde{C}};|1\rangle_{\tilde{c}}],
\;
M_{2,2}=I-M_{2,1}
\bigr\}.
\end{aligned}
\]
If outcome \(M_{2,1}\) occurs, then the subset is \(H_7\). Otherwise, the state belongs to one of the remaining subsets.

\textbf{Step 3.} Alice performs the measurement
\[
\begin{aligned}
\mathcal{M}_3\equiv \bigl\{
M_{3,1}=P[|d_1'\rangle_A;|1\rangle_a],
\;
M_{3,2}=I-M_{3,1}
\bigr\}.
\end{aligned}
\]

If outcome \(M_{3,1}\) occurs, then the state belongs to one of
\[
\{H_4,H_5,H_6,H_{10},H_{14},H_{18}\}.
\]
Bob then performs the measurement
\[
\begin{aligned}
\mathcal{M}_4\equiv \bigl\{
M_{4,1}=|d_2'\rangle_B\langle d_2'|,
\;
M_{4,2}=I-M_{4,1}
\bigr\}.
\end{aligned}
\]
If outcome \(M_{4,1}\) occurs, the remaining possibilities are \(\{H_5,H_6\}\), which are LOCC distinguishable. For outcome \(M_{4,2}\), the remaining subsets are \(H_4,H_{10},H_{14},H_{18}\). Next, \(\widetilde{C}\) performs
\[
\mathcal{M}_4' \equiv \bigl\{
M_{4,1}'=P[|d_3'd_4'\rangle_{\widetilde{C}};|1\rangle_{\tilde{c}}],
\;
M_{4,2}'=P[|(1,\ldots,d_3'-1)\circ 0\rangle_{\widetilde{C}};|1\rangle_{\tilde{c}}],
\;
M_{4,3}'=I-M_{4,1}'-M_{4,2}'
\bigr\}.
\]
The corresponding results for \(M_{4,1}'\), \(M_{4,2}'\), and \(M_{4,3}'\) are \(H_4, H_{14}\) and \(\{H_{10},H_{18}\}\), respectively. These sets are all LOCC distinguishable.

If outcome \(M_{3,2}\) occurs, then the state belongs to one of
\[
\{H_1,H_2,H_3,H_8,H_9,H_{11},H_{12},H_{13},H_{15},H_{16},H_{17}\}.
\]
We proceed to the next step.

\textbf{Step 4.} \(\widetilde{C}\) performs the measurement
\[
\begin{aligned}
\mathcal{M}_5\equiv \bigl\{&
M_{5,1}=P[|(1,\ldots,d_3')\circ(0,\ldots,d_4'-1)\rangle_{\widetilde{C}};|1\rangle_{\tilde{c}}]\\
&\qquad\qquad
+P[|(0,\ldots,d_3'-1)\circ (1,\ldots,d_4')\rangle_{\widetilde{C}};|1\rangle_{\tilde{c}}],\\
&M_{5,2}=P[|d_3'd_4'\rangle_{\widetilde{C}};I_{\tilde{c}}],\\
&M_{5,3}=I-M_{5,1}-M_{5,2}
\bigr\}.
\end{aligned}
\]
If outcome \(M_{5,1}\) occurs, then the remaining subsets are
\[
\{H_{11},H_{12},H_{15},H_{16}\}.
\]
Bob then performs
\[
\mathcal{M}_5'\equiv \bigl\{
M_{5,1}'=|0\rangle_B\langle 0|,
\;
M_{5,2}'=I-M_{5,1}'
\bigr\}.
\]
The corresponding results for \(M_{5,1}'\) and \(M_{5,2}'\) are \(\{H_{12},H_{16}\}\) and \(\{H_{11},H_{15}\}\), respectively. These sets are LOCC distinguishable.

If outcome \(M_{5,2}\) occurs, then the subset is \(H_3\). If outcome \(M_{5,3}\) occurs, then the remaining possibilities are
\[
\{H_1,H_2,H_8,H_9,H_{13},H_{17}\}.
\]

\textbf{Step 5.} Bob performs the measurement
\[
\begin{aligned}
\mathcal{M}_6\equiv \bigl\{
M_{6,1}=|0\rangle_B\langle 0|,
\;
M_{6,2}=I-M_{6,1}
\bigr\}.
\end{aligned}
\]
If outcome \(M_{6,1}\) occurs, then the remaining possibilities are \(\{H_1,H_2\}\), which are LOCC distinguishable. If outcome \(M_{6,2}\) occurs, then the remaining subsets are
\[
\{H_8,H_9,H_{13},H_{17}\}.
\]

\textbf{Step 6.} \(\widetilde{C}\) performs the measurement
\[
\begin{aligned}
\mathcal{M}_7\equiv \bigl\{&
M_{7,1}=P[|00\rangle_{\widetilde{C}};|0\rangle_{\tilde{c}}],\\
&M_{7,2}=P[|(1,\ldots,d_3'-1)\circ d_4'\rangle_{\widetilde{C}};|0\rangle_{\tilde{c}}],\\
&M_{7,3}=I-M_{7,1}-M_{7,2}
\bigr\}.
\end{aligned}
\]
If outcome \(M_{7,1}\) occurs, then the subset is \(H_8\). If outcome \(M_{7,2}\) occurs, then the subset is \(H_{13}\). If outcome \(M_{7,3}\) occurs, then the remaining possibilities are \(\{H_9,H_{17}\}\), which are LOCC distinguishable.

If outcome \(M_{1,2}\) occurs in Step 1, a similar argument yields a protocol that perfectly distinguishes the states in set~(\ref{set4.1}) by LOCC.

\section{The proof of Theorem 4}\label{T4}

Similarly, it is sufficient to distinguish the subsets \(H_r\) \((r=1,2,\ldots,18)\) in set~(\ref{set4.1}). Without using teleportation, let Alice and Dave share an EPR state \(\vert \phi^{+}(2)\rangle\), and let Charlie and Dave share a maximally entangled state \(\vert \phi^{+}(3)\rangle\). The initial state is
\[
\begin{aligned}
\left|\psi\right\rangle_{ABCD}\otimes\left|\phi^{+}(2)\right\rangle_{av_1}\otimes\left|\phi^{+}(3)\right\rangle_{cv_2},
\end{aligned}
\]
where \(a\) is Alice's ancillary system, \(c\) is Charlie's ancillary system, and \(v_1,v_2\) are Dave's ancillary systems.

\textbf{Step 1.} Alice and Charlie perform the measurements
\[
\begin{aligned}
\mathcal{M}_1\equiv \bigl\{&
M_{1,1}=P\bigl[|0\rangle_A;|0\rangle_a\bigr]
+P\bigl[(|1\rangle,\ldots,|d_1'\rangle)_A;|1\rangle_a\bigr],\\
&M_{1,2}=I-M_{1,1}
\bigr\},
\end{aligned}
\]
and
\[
\begin{aligned}
\mathcal{M}_2\equiv \bigl\{&
M_{2,1}=P\bigl[|0\rangle_C;|0\rangle_c\bigr]
+P\bigl[(|1\rangle,\ldots,|d_3'-1\rangle)_C;|1\rangle_c\bigr]
+P\bigl[|d_3'\rangle_C;|2\rangle_c\bigr],\\
&M_{2,2}=P\bigl[|0\rangle_C;|1\rangle_c\bigr]
+P\bigl[(|1\rangle,\ldots,|d_3'-1\rangle)_C;|2\rangle_c\bigr]
+P\bigl[|d_3'\rangle_C;|0\rangle_c\bigr],\\
&M_{2,3}=I-M_{2,1}-M_{2,2}
\bigr\},
\end{aligned}
\]
respectively. Suppose that the outcomes corresponding to \(M_{1,1}\) and \(M_{2,1}\) occur. Then
\[
\begin{aligned}
&H_1 \rightarrow |0\rangle_1|0\rangle_2|\beta_j'\rangle_3|\alpha_i\rangle_4|00\rangle_{av_1}|11\rangle_{cv_2}\\
&\qquad\qquad+|0\rangle_1|0\rangle_2|0\rangle_3|\alpha_i\rangle_4|00\rangle_{av_1}|00\rangle_{cv_2},\\
&H_2\rightarrow |0\rangle_1|0\rangle_2|d_3'\rangle_3|\beta_j\rangle_4|00\rangle_{av_1}|22\rangle_{cv_2},\\
&H_3 \rightarrow |0\rangle_1|\alpha_i\rangle_2|d_3'\rangle_3|d_4'\rangle_4|00\rangle_{av_1}|22\rangle_{cv_2}\\
&\qquad\qquad+|\beta_j'\rangle_1|\alpha_i\rangle_2|d_3'\rangle_3|d_4'\rangle_4|11\rangle_{av_1}|22\rangle_{cv_2},\\
&H_4 \rightarrow |d_1'\rangle_1|\beta_j\rangle_2|d_3'\rangle_3|d_4'\rangle_4|11\rangle_{av_1}|22\rangle_{cv_2},\\
&H_5 \rightarrow |d_1'\rangle_1|d_2'\rangle_2|\gamma_m'\rangle_3|\alpha_i\rangle_4|11\rangle_{av_1}|11\rangle_{cv_2}\\
&\qquad\qquad+|d_1'\rangle_1|d_2'\rangle_2|d_3'\rangle_3|\alpha_i\rangle_4|11\rangle_{av_1}|22\rangle_{cv_2},\\
&H_6 \rightarrow |d_1'\rangle_1|d_2'\rangle_2|0\rangle_3|\gamma_m\rangle_4|11\rangle_{av_1}|00\rangle_{cv_2},\\
&H_7 \rightarrow |\gamma_m\rangle_1|\alpha_i\rangle_2|0\rangle_3|0\rangle_4|11\rangle_{av_1}|00\rangle_{cv_2},\\
&H_8 \rightarrow |0\rangle_1|\gamma_m\rangle_2|0\rangle_3|0\rangle_4|00\rangle_{av_1}|00\rangle_{cv_2},\\
&H_9 \rightarrow |0\rangle_1|\kappa_I\rangle_2|0\rangle_3|\gamma_m\rangle_4|00\rangle_{av_1}|00\rangle_{cv_2},\\
&H_{10} \rightarrow |d_1'\rangle_1|\kappa_I\rangle_2|d_3'\rangle_3|\beta_j\rangle_4|11\rangle_{av_1}|22\rangle_{cv_2},\\
&H_{11} \rightarrow |\kappa_I\rangle_1|d_2'\rangle_2|\beta_j'\rangle_3|d_4'\rangle_4|11\rangle_{av_1}|11\rangle_{cv_2}\\
&\qquad\qquad+|\kappa_I\rangle_1|d_2'\rangle_2|0\rangle_3|d_4'\rangle_4|11\rangle_{av_1}|00\rangle_{cv_2},\\
&H_{12} \rightarrow |\kappa_I\rangle_1|0\rangle_2|\gamma_m'\rangle_3|0\rangle_4|11\rangle_{av_1}|11\rangle_{cv_2}\\
&\qquad\qquad+|\kappa_I\rangle_1|0\rangle_2|d_3'\rangle_3|0\rangle_4|11\rangle_{av_1}|22\rangle_{cv_2},\\
&H_{13} \rightarrow |0\rangle_1|\gamma_m\rangle_2|\kappa_I\rangle_3|d_4'\rangle_4|00\rangle_{av_1}|11\rangle_{cv_2},\\
&H_{14} \rightarrow |d_1'\rangle_1|\beta_j\rangle_2|\kappa_I\rangle_3|0\rangle_4|11\rangle_{av_1}|11\rangle_{cv_2},\\
&H_{15} \rightarrow |\kappa_I\rangle_1|d_2'\rangle_2|d_3'\rangle_3|\beta_j\rangle_4|11\rangle_{av_1}|22\rangle_{cv_2},\\
&H_{16} \rightarrow |\kappa_I\rangle_1|0\rangle_2|0\rangle_3|\gamma_m\rangle_4|11\rangle_{av_1}|00\rangle_{cv_2},\\
&H_{17} \rightarrow |0\rangle_1|d_2'\rangle_2|0\rangle_3|\kappa_I\rangle_4|00\rangle_{av_1}|00\rangle_{cv_2}\\
&\qquad\qquad\pm|0\rangle_1|d_2'\rangle_2|d_3'\rangle_3|\kappa_I\rangle_4|00\rangle_{av_1}|22\rangle_{cv_2},\\
&H_{18} \rightarrow |d_1'\rangle_1|0\rangle_2|0\rangle_3|\kappa_I\rangle_4|11\rangle_{av_1}|00\rangle_{cv_2}\\
&\qquad\qquad\pm|d_1'\rangle_1|0\rangle_2|d_3'\rangle_3|\kappa_I\rangle_4|11\rangle_{av_1}|22\rangle_{cv_2},
\end{aligned}
\]
where
\[
|\beta_j'\rangle_k=\sum_{u=1}^{d_k-2}\omega_{d_k-1}^{ju}|u\rangle,
\qquad
|\gamma_m'\rangle_k=\sum_{u=0}^{d_k-3}\omega_{d_k-1}^{mu}|u+1\rangle
\]
for \(k=1,2,3,4\).

\textbf{Step 2.} Dave performs the measurement
\[
\begin{aligned}
\mathcal{M}_3\equiv \bigl\{
M_{3,1}=P[|0\rangle_D;|1\rangle_{v_1};|0\rangle_{v_2}],
\;
M_{3,2}=I-M_{3,1}
\bigr\}.
\end{aligned}
\]
If outcome \(M_{3,1}\) occurs, then the subset is \(H_7\). Otherwise, the state belongs to one of the remaining subsets.

\textbf{Step 3.} Alice performs the measurement
\[
\begin{aligned}
\mathcal{M}_4\equiv \bigl\{
M_{4,1}=|d_1'\rangle_A\langle d_1'|,
\;
M_{4,2}=I-M_{4,1}
\bigr\}.
\end{aligned}
\]
The corresponding results are
\[
\begin{aligned}
&M_{4,1}\Rightarrow
\begin{cases}
H_4,H_5,H_6,\\
H_{10},H_{14},H_{18},
\end{cases}\\
&M_{4,2}\Rightarrow
\begin{cases}
H_1,H_2,H_3,H_8,H_9,H_{11},\\
H_{12},H_{13},H_{15},H_{16},H_{17}.
\end{cases}
\end{aligned}
\]

\textbf{Step 4.} For the branch corresponding to \(M_{4,1}\), Bob performs
\[
\mathcal{M}_5\equiv \bigl\{
M_{5,1}=|d_2'\rangle_B\langle d_2'|,
\;
M_{5,2}=I-M_{5,1}
\bigr\}.
\]
If outcome \(M_{5,1}\) occurs, the remaining possibilities are \(\{H_5,H_6\}\), which can be perfectly distinguished by LOCC. Otherwise, we proceed to the next step.

\textbf{Step 5.} Dave performs the measurement
\[
\begin{aligned}
\mathcal{M}_6\equiv \bigl\{&
M_{6,1}=P[|d_4'\rangle_D;|1\rangle_{v_1};|2\rangle_{v_2}],\\
&M_{6,2}=P[|0\rangle_D;|1\rangle_{v_1};|1\rangle_{v_2}],\\
&M_{6,3}=I-M_{6,1}-M_{6,2}
\bigr\}.
\end{aligned}
\]
If outcome \(M_{6,1}\) occurs, then the subset is \(H_4\); if outcome \(M_{6,2}\) occurs, then the subset is \(H_{14}\); if outcome \(M_{6,3}\) occurs, then the remaining possibilities are \(\{H_{10},H_{18}\}\), which are LOCC distinguishable.

\textbf{Step 4$'$.} For the branch corresponding to \(M_{4,2}\) in Step 3, Dave performs
\[
\begin{aligned}
\mathcal{M}_7\equiv \bigl\{&
M_{7,1}=P[|d_4'\rangle_D;(|0\rangle,|1\rangle)_{v_1};|2\rangle_{v_2}],\\
&M_{7,2}=I-M_{7,1}
\bigr\}.
\end{aligned}
\]
If outcome \(M_{7,1}\) occurs, then the subset is \(H_3\). Otherwise, we continue.

\textbf{Step 5$'$.} Bob performs the measurement
\[
\begin{aligned}
\mathcal{M}_8\equiv \bigl\{
M_{8,1}=|0\rangle_B\langle 0|,
\;
M_{8,2}=I-M_{8,1}
\bigr\}.
\end{aligned}
\]
The corresponding results are
\[
\begin{aligned}
&M_{8,1}\Rightarrow H_1,H_2,H_{12},H_{16},\\
&M_{8,2}\Rightarrow H_8,H_9,H_{11},H_{13},H_{15},H_{17}.
\end{aligned}
\]

For the branch corresponding to \(M_{8,1}\), Alice performs
\[
\mathcal{M}_9\equiv \bigl\{
M_{9,1}=P[|0\rangle_A;|0\rangle_a],
\;
M_{9,2}=I-M_{9,1}
\bigr\}.
\]
If outcome \(M_{9,1}\) occurs, then the remaining possibilities are \(\{H_1,H_2\}\), which can be perfectly distinguished by measuring subsystem \(C\). Otherwise, the remaining possibilities are \(\{H_{12},H_{16}\}\), which can be perfectly distinguished by measuring subsystem \(D\).

For the branch corresponding to \(M_{8,2}\), Alice again performs
\[
\mathcal{M}_9\equiv \bigl\{
M_{9,1}=P[|0\rangle_A;|0\rangle_a],
\;
M_{9,2}=I-M_{9,1}
\bigr\}.
\]
If outcome \(M_{9,1}\) occurs, then the remaining possibilities are
\[
\{H_8,H_9,H_{13},H_{17}\}.
\]
Dave then performs
\[
\mathcal{M}_{10}\equiv \bigl\{
M_{10,1}=P[|0\rangle_D;|0\rangle_{v_1};|0\rangle_{v_2}],
\;
M_{10,2}=P[|d_4'\rangle_D;|0\rangle_{v_1};|1\rangle_{v_2}],
\;
M_{10,3}=I-M_{10,1}-M_{10,2}
\bigr\}.
\]
If outcome \(M_{10,1}\) occurs, then the subset is \(H_8\); if outcome \(M_{10,2}\) occurs, then the subset is \(H_{13}\); if outcome \(M_{10,3}\) occurs, then the remaining possibilities are \(\{H_9,H_{17}\}\), which are LOCC distinguishable.

If outcome \(M_{9,2}\) occurs, then the remaining possibilities are \(\{H_{11},H_{15}\}\), which are LOCC distinguishable.

If any other outcomes occur in Step 1, one can construct a similar protocol to distinguish the states in set~(\ref{set4.1}) locally.

\section{The proof of Theorem 5}\label{T5}

It is known that each subset \(H_r\) \((r=1,2,\ldots,18)\) in set~(\ref{set4.1}) is locally distinguishable. Thus, we only need to identify these subsets by LOCC. Let one EPR state \(\vert \phi^{+}(2)\rangle\) be shared between Alice and Dave, and another EPR state \(\vert \phi^{+}(2)\rangle\) be shared between Charlie and Dave. The initial state is
\[
\begin{aligned}
\left|\psi\right\rangle_{ABCD}\otimes\left|\phi^{+}(2)\right\rangle_{av_1}\otimes\left|\phi^{+}(2)\right\rangle_{cv_2},
\end{aligned}
\]
where \(a\) is Alice's ancillary system, \(c\) is Charlie's ancillary system, and \(v_1,v_2\) are Dave's ancillary systems.

\textbf{Step 1.} Alice and Charlie perform the measurements
\[
\begin{aligned}
\mathcal{M}_1\equiv \bigl\{&
M_{1,1}=P\bigl[|0\rangle_A;|0\rangle_a\bigr]
+P\bigl[(|1\rangle,\ldots,|d_1'\rangle)_A;|1\rangle_a\bigr],\\
&M_{1,2}=I-M_{1,1}
\bigr\},
\end{aligned}
\]
and
\[
\begin{aligned}
\mathcal{M}_2\equiv \bigl\{&
M_{2,1}=P\bigl[|0\rangle_C;|0\rangle_c\bigr]
+P\bigl[(|1\rangle,\ldots,|d_3'\rangle)_C;|1\rangle_c\bigr],\\
&M_{2,2}=I-M_{2,1}
\bigr\},
\end{aligned}
\]
respectively. Suppose that the outcomes corresponding to \(M_{1,1}\) and \(M_{2,1}\) occur. Then
\[
\begin{aligned}
&H_1\rightarrow |0\rangle_1|0\rangle_2|\beta_j'\rangle_3|\alpha_i\rangle_4|00\rangle_{av_1}|11\rangle_{cv_2}\\
&\qquad\qquad+|0\rangle_1|0\rangle_2|0\rangle_3|\alpha_i\rangle_4|00\rangle_{av_1}|00\rangle_{cv_2},\\
&H_2\rightarrow |0\rangle_1|0\rangle_2|d_3'\rangle_3|\beta_j\rangle_4|00\rangle_{av_1}|11\rangle_{cv_2},\\
&H_3\rightarrow |0\rangle_1|\alpha_i\rangle_2|d_3'\rangle_3|d_4'\rangle_4|00\rangle_{av_1}|11\rangle_{cv_2}\\
&\qquad\qquad+|\beta_j'\rangle_1|\alpha_i\rangle_2|d_3'\rangle_3|d_4'\rangle_4|11\rangle_{av_1}|11\rangle_{cv_2},\\
&H_4\rightarrow |d_1'\rangle_1|\beta_j\rangle_2|d_3'\rangle_3|d_4'\rangle_4|11\rangle_{av_1}|11\rangle_{cv_2},\\
&H_5\rightarrow |d_1'\rangle_1|d_2'\rangle_2|\gamma_m\rangle_3|\alpha_i\rangle_4|11\rangle_{av_1}|11\rangle_{cv_2},\\
&H_6\rightarrow |d_1'\rangle_1|d_2'\rangle_2|0\rangle_3|\gamma_m\rangle_4|11\rangle_{av_1}|00\rangle_{cv_2},\\
&H_7\rightarrow |\gamma_m\rangle_1|\alpha_i\rangle_2|0\rangle_3|0\rangle_4|11\rangle_{av_1}|00\rangle_{cv_2},\\
&H_8\rightarrow |0\rangle_1|\gamma_m\rangle_2|0\rangle_3|0\rangle_4|00\rangle_{av_1}|00\rangle_{cv_2},\\
&H_9\rightarrow |0\rangle_1|\kappa_I\rangle_2|0\rangle_3|\gamma_m\rangle_4|00\rangle_{av_1}|00\rangle_{cv_2},\\
&H_{10}\rightarrow |d_1'\rangle_1|\kappa_I\rangle_2|d_3'\rangle_3|\beta_j\rangle_4|11\rangle_{av_1}|11\rangle_{cv_2},\\
&H_{11}\rightarrow |\kappa_I\rangle_1|d_2'\rangle_2|\beta_j'\rangle_3|d_4'\rangle_4|11\rangle_{av_1}|11\rangle_{cv_2}\\
&\qquad\qquad+|\kappa_I\rangle_1|d_2'\rangle_2|0\rangle_3|d_4'\rangle_4|11\rangle_{av_1}|00\rangle_{cv_2},\\
&H_{12}\rightarrow |\kappa_I\rangle_1|0\rangle_2|\gamma_m\rangle_3|0\rangle_4|11\rangle_{av_1}|11\rangle_{cv_2},\\
&H_{13}\rightarrow |0\rangle_1|\gamma_m\rangle_2|\kappa_I\rangle_3|d_4'\rangle_4|00\rangle_{av_1}|11\rangle_{cv_2},\\
&H_{14}\rightarrow |d_1'\rangle_1|\beta_j\rangle_2|\kappa_I\rangle_3|0\rangle_4|11\rangle_{av_1}|11\rangle_{cv_2},\\
&H_{15}\rightarrow |\kappa_I\rangle_1|d_2'\rangle_2|d_3'\rangle_3|\beta_j\rangle_4|11\rangle_{av_1}|11\rangle_{cv_2},\\
&H_{16}\rightarrow |\kappa_I\rangle_1|0\rangle_2|0\rangle_3|\gamma_m\rangle_4|11\rangle_{av_1}|00\rangle_{cv_2},\\
&H_{17}\rightarrow |0\rangle_1|d_2'\rangle_2|0\rangle_3|\kappa_I\rangle_4|00\rangle_{av_1}|00\rangle_{cv_2}\\
&\qquad\qquad\pm|0\rangle_1|d_2'\rangle_2|d_3'\rangle_3|\kappa_I\rangle_4|00\rangle_{av_1}|11\rangle_{cv_2},\\
&H_{18}\rightarrow |d_1'\rangle_1|0\rangle_2|0\rangle_3|\kappa_I\rangle_4|11\rangle_{av_1}|00\rangle_{cv_2}\\
&\qquad\qquad\pm|d_1'\rangle_1|0\rangle_2|d_3'\rangle_3|\kappa_I\rangle_4|11\rangle_{av_1}|11\rangle_{cv_2}.
\end{aligned}
\]
Here
\[
|\beta_j'\rangle_k=\sum_{u=1}^{d_k-2}\omega_{d_k-1}^{ju}|u\rangle
\]
for \(k=1,2,3,4\).

\textbf{Step 2.} Dave performs
\[
\begin{aligned}
\mathcal{M}_3\equiv \bigl\{
M_{3,1}=P[|0\rangle_D;|1\rangle_{v_1};|0\rangle_{v_2}],
\;
M_{3,2}=I-M_{3,1}
\bigr\}.
\end{aligned}
\]
If outcome \(M_{3,1}\) occurs, then the subset is \(H_7\). Otherwise, the state belongs to one of the remaining subsets.

\textbf{Step 3.} Alice performs
\[
\begin{aligned}
\mathcal{M}_4\equiv \bigl\{
M_{4,1}=|d_1'\rangle_A\langle d_1'|,
\;
M_{4,2}=I-M_{4,1}
\bigr\}.
\end{aligned}
\]
The corresponding results are
\[
\begin{aligned}
&M_{4,1}\Rightarrow
\begin{cases}
H_4,H_5,H_6,\\
H_{10},H_{14},H_{18},
\end{cases}\\
&M_{4,2}\Rightarrow
\begin{cases}
H_1,H_2,H_3,H_8,H_9,H_{11},\\
H_{12},H_{13},H_{15},H_{16},H_{17}.
\end{cases}
\end{aligned}
\]

\textbf{Step 4.} If outcome \(M_{4,1}\) occurs, Bob performs
\[
\mathcal{M}_5\equiv \bigl\{
M_{5,1}=|d_2'\rangle_B\langle d_2'|,
\;
M_{5,2}=I-M_{5,1}
\bigr\}.
\]
If outcome \(M_{5,1}\) occurs, then the remaining possibilities are \(\{H_5,H_6\}\), which can be perfectly distinguished by measuring subsystem \(C\). Otherwise, Dave performs
\[
\mathcal{M}_6\equiv \bigl\{
M_{6,1}=P[|d_4'\rangle_D;|1\rangle_{v_1};|1\rangle_{v_2}],
\;
M_{6,2}=I-M_{6,1}
\bigr\}.
\]
If outcome \(M_{6,1}\) occurs, then the subset is \(H_4\). Otherwise, Charlie performs
\[
\mathcal{M}_7\equiv \bigl\{
M_{7,1}=|1\rangle_C\langle 1|+\cdots+|d_3'-1\rangle_C\langle d_3'-1|,
\;
M_{7,2}=I-M_{7,1}
\bigr\}.
\]
The corresponding results are \(H_{14}\) and \(\{H_{10},H_{18}\}\), respectively. These sets are LOCC distinguishable.

If outcome \(M_{4,2}\) occurs, let Bob and Dave additionally share a maximally entangled state \(\left|\phi^{+}(2)\right\rangle\). Then the relevant initial state becomes
\[
\begin{aligned}
\left|\psi\right\rangle_{ABCD}
\otimes\left|\phi^{+}(2)\right\rangle_{av_1}
\otimes\left|\phi^{+}(2)\right\rangle_{bv_3}
\otimes\left|\phi^{+}(2)\right\rangle_{cv_2}.
\end{aligned}
\]

Bob now performs
\[
\begin{aligned}
\mathcal{M}_8\equiv \bigl\{&
M_{8,1}=P\bigl[|0\rangle_B;|0\rangle_b\bigr]
+P\bigl[(|1\rangle,\ldots,|d_2'\rangle)_B;|1\rangle_b\bigr],\\
&M_{8,2}=I-M_{8,1}
\bigr\}.
\end{aligned}
\]
Suppose that outcome \(M_{8,1}\) occurs. Then
\[
\begin{aligned}
&H_1 \rightarrow |0\rangle_1|0\rangle_2|\beta_j'\rangle_3|\alpha_i\rangle_4|00\rangle_{av_1}|00\rangle_{bv_3}|11\rangle_{cv_2}\\
&\qquad\qquad+|0\rangle_1|0\rangle_2|0\rangle_3|\alpha_i\rangle_4|00\rangle_{av_1}|00\rangle_{bv_3}|00\rangle_{cv_2},\\
&H_2 \rightarrow |0\rangle_1|0\rangle_2|d_3'\rangle_3|\beta_j\rangle_4|00\rangle_{av_1}|00\rangle_{bv_3}|11\rangle_{cv_2},\\
&H_3 \rightarrow |\beta_j'\rangle_1|\alpha_i'\rangle_2|d_3'\rangle_3|d_4'\rangle_4|11\rangle_{av_1}|11\rangle_{bv_3}|11\rangle_{cv_2}\\
&\qquad\qquad+|\beta_j'\rangle_1|0\rangle_2|d_3'\rangle_3|d_4'\rangle_4|11\rangle_{av_1}|00\rangle_{bv_3}|11\rangle_{cv_2}\\
&\qquad\qquad+|0\rangle_1|\alpha_i'\rangle_2|d_3'\rangle_3|d_4'\rangle_4|00\rangle_{av_1}|11\rangle_{bv_3}|11\rangle_{cv_2}\\
&\qquad\qquad+|0\rangle_1|0\rangle_2|d_3'\rangle_3|d_4'\rangle_4|00\rangle_{av_1}|00\rangle_{bv_3}|11\rangle_{cv_2},\\
&H_8 \rightarrow |0\rangle_1|\gamma_m\rangle_2|0\rangle_3|0\rangle_4|00\rangle_{av_1}|11\rangle_{bv_3}|00\rangle_{cv_2},\\
&H_9 \rightarrow |0\rangle_1|\kappa_I\rangle_2|0\rangle_3|\gamma_m\rangle_4|00\rangle_{av_1}|11\rangle_{bv_3}|00\rangle_{cv_2},\\
&H_{11} \rightarrow |\kappa_I\rangle_1|d_2'\rangle_2|\beta_j'\rangle_3|d_4'\rangle_4|11\rangle_{av_1}|11\rangle_{bv_3}|11\rangle_{cv_2}\\
&\qquad\qquad+|\kappa_I\rangle_1|d_2'\rangle_2|0\rangle_3|d_4'\rangle_4|11\rangle_{av_1}|11\rangle_{bv_3}|00\rangle_{cv_2},\\
&H_{12} \rightarrow |\kappa_I\rangle_1|0\rangle_2|\gamma_m\rangle_3|0\rangle_4|11\rangle_{av_1}|00\rangle_{bv_3}|11\rangle_{cv_2},\\
&H_{13} \rightarrow |0\rangle_1|\gamma_m\rangle_2|\kappa_I\rangle_3|d_4'\rangle_4|00\rangle_{av_1}|11\rangle_{bv_3}|11\rangle_{cv_2},\\
&H_{15} \rightarrow |\kappa_I\rangle_1|d_2'\rangle_2|d_3'\rangle_3|\beta_j\rangle_4|11\rangle_{av_1}|11\rangle_{bv_3}|11\rangle_{cv_2},\\
&H_{16} \rightarrow |\kappa_I\rangle_1|0\rangle_2|0\rangle_3|\gamma_m\rangle_4|11\rangle_{av_1}|00\rangle_{bv_3}|00\rangle_{cv_2},\\
&H_{17} \rightarrow |0\rangle_1|d_2'\rangle_2|0\rangle_3|\kappa_I\rangle_4|00\rangle_{av_1}|11\rangle_{bv_3}|00\rangle_{cv_2}\\
&\qquad\qquad\pm|0\rangle_1|d_2'\rangle_2|d_3'\rangle_3|\kappa_I\rangle_4|00\rangle_{av_1}|11\rangle_{bv_3}|11\rangle_{cv_2},
\end{aligned}
\]
where
\[
|\alpha_i'\rangle_k=\sum_{u=1}^{d_k-1}\omega_{d_k}^{iu}|u\rangle
\]
for \(k=1,2,3,4\).

We then proceed as follows.

\textbf{Step 5.} Dave performs
\[
\begin{aligned}
\mathcal{M}_9\equiv \bigl\{&
M_{9,1}=P[|0\rangle_D;|0\rangle_{v_1};|1\rangle_{v_3};|0\rangle_{v_2}],\\
&M_{9,2}=P[|0\rangle_D;|1\rangle_{v_1};|0\rangle_{v_3};|1\rangle_{v_2}],\\
&M_{9,3}=P[(|0\rangle,\ldots,|d_4'-1\rangle)_D;|1\rangle_{v_1};|1\rangle_{v_3};|1\rangle_{v_2}],\\
&M_{9,4}=P[(|1\rangle,\ldots,|d_4'\rangle)_D;|1\rangle_{v_1};|0\rangle_{v_3};|0\rangle_{v_2}],\\
&M_{9,5}=P[(|1\rangle,\ldots,|d_4'\rangle)_D;|0\rangle_{v_1};|1\rangle_{v_3};|0\rangle_{v_2}]+P[(|1\rangle,\ldots,|d_4'-1\rangle)_D;|0\rangle_{v_1};|1\rangle_{v_3};|1\rangle_{v_2}],\\
&M_{9,6}=I-M_{9,1}-\cdots-M_{9,5}
\bigr\}.
\end{aligned}
\]
The corresponding subsets are
\[
\begin{aligned}
&M_{9,1}\Rightarrow H_8,\qquad &&M_{9,4}\Rightarrow H_{16},\\
&M_{9,2}\Rightarrow H_{12},\qquad &&M_{9,5}\Rightarrow H_9,H_{17},\\
&M_{9,3}\Rightarrow H_{15},\qquad &&M_{9,6}\Rightarrow H_1,H_2,H_3,H_{11},H_{13}.
\end{aligned}
\]
The pair \(\{H_9,H_{17}\}\) is LOCC distinguishable. Thus, if outcome \(M_{9,6}\) occurs, we move to the next step.

\textbf{Step 6.} Charlie performs
\[
\begin{aligned}
\mathcal{M}_{10}\equiv \bigl\{
M_{10,1}=|d_3'\rangle_C\langle d_3'|,
\;
M_{10,2}=I-M_{10,1}
\bigr\}.
\end{aligned}
\]
The corresponding results are \(\{H_2,H_3\}\) and \(\{H_1,H_{11},H_{15}\}\), respectively. These sets are LOCC distinguishable.

If outcome \(M_{1,2}\) in Step 1 or outcome \(M_{8,2}\) in Step 4 occurs, a similar argument gives an LOCC protocol for distinguishing the states in set~(\ref{set4.1}).

\section{The proof of Theorem 6}\label{T6}

First, we teleport subsystem \(D\) to Charlie and denote the joint subsystem of Charlie and Dave by \(\widetilde{C}\). A state \(|\phi^{+}(3)\rangle\) is shared between Alice and Charlie. The initial state is
\[
\begin{aligned}
\left|\psi\right\rangle_{AB\widetilde{C}}\otimes\left|\phi^{+}(3)\right\rangle_{a\tilde{c}}.
\end{aligned}
\]

\textbf{Step 1.} Alice performs the measurement
\[
\begin{aligned}
\mathcal{M}_1\equiv \bigl\{&
M_{1,1}=P\bigl[|0\rangle_A;|0\rangle_a\bigr]
+P\bigl[(|1\rangle,\cdots,|d_1'-1\rangle)_A;|1\rangle_a\bigr]
+P\bigl[|d_1'\rangle_A;|2\rangle_a\bigr],\\
&M_{1,2}=P\bigl[|0\rangle_A;|1\rangle_a\bigr]
+P\bigl[(|1\rangle,\cdots,|d_1'-1\rangle)_A;|2\rangle_a\bigr]
+P\bigl[|d_1'\rangle_A;|0\rangle_a\bigr],\\
&M_{1,3}=I-M_{1,1}-M_{1,2}
\bigr\}.
\end{aligned}
\]
Suppose that outcome \(M_{1,1}\) occurs. Then
\[
\begin{aligned}
H_{1,1} &\rightarrow |\kappa_I\rangle|\kappa_I\rangle|0\circ\gamma_m\rangle|11\rangle,\\
H_{1,2} &\rightarrow |\kappa_I\rangle|0\rangle|\gamma_m\circ\kappa_I\rangle|11\rangle,\\
H_{1,3} &\rightarrow |0\rangle|\gamma_m\rangle|\kappa_I\circ\kappa_I\rangle|00\rangle,\\
H_{1,4} &\rightarrow |\gamma_m'\rangle|\kappa_I\rangle|\kappa_I\circ0\rangle|11\rangle
+|d_1'\rangle|\kappa_I\rangle|\kappa_I\circ0\rangle|22\rangle,\\
H_{2,1} &\rightarrow |d_1'\rangle|0\rangle|0\circ\beta_j\rangle|22\rangle,\\
H_{2,2} &\rightarrow |0\rangle|0\rangle|\beta_j\circ d_4'\rangle|00\rangle,\\
H_{2,3} &\rightarrow |0\rangle|\beta_j\rangle|d_3'\circ0\rangle|00\rangle,\\
H_{2,4} &\rightarrow |\beta_j'\rangle|d_2'\rangle|0\circ0\rangle|11\rangle
+|0\rangle|d_2'\rangle|0\circ0\rangle|00\rangle,\\
H_{3,1} &\rightarrow |0\rangle|0\rangle|d_3'\circ\gamma_m\rangle|00\rangle,\\
H_{3,2} &\rightarrow |0\rangle|d_2'\rangle|\gamma_m\circ0\rangle|00\rangle,\\
H_{3,3} &\rightarrow |d_1'\rangle|\gamma_m\rangle|0\circ0\rangle|22\rangle,\\
H_{3,4} &\rightarrow |\gamma_m'\rangle|0\rangle|0\circ d_4'\rangle|11\rangle
+|d_1'\rangle|0\rangle|0\circ d_4'\rangle|22\rangle,\\
H_{4,1} &\rightarrow |\kappa_I\rangle|d_2'\rangle|d_3'\circ(0\pm d_4')\rangle|11\rangle,\\
H_{4,2} &\rightarrow |d_1'\rangle|d_2'\rangle|(0\pm d_3')\circ\kappa_I\rangle|22\rangle,\\
H_{4,3} &\rightarrow |d_1'\rangle|0\pm d_2'\rangle|\kappa_I\circ d_4'\rangle|22\rangle,\\
H_{4,4} &\rightarrow |0\rangle|\kappa_I\rangle|d_3'\circ d_4'\rangle|00\rangle
\pm|d_1'\rangle|\kappa_I\rangle|d_3'\circ d_4'\rangle|22\rangle,\\
H_{5,1} &\rightarrow |d_1'\rangle|d_2'\rangle|\kappa_I\circ\beta_j\rangle|22\rangle,\\
H_{5,2} &\rightarrow |d_1'\rangle|\kappa_I\rangle|\beta_j\circ d_4'\rangle|22\rangle,\\
H_{5,3} &\rightarrow |\kappa_I\rangle|\beta_j\rangle|d_3'\circ d_4'\rangle|11\rangle,\\
H_{5,4} &\rightarrow |\beta_j'\rangle|d_2'\rangle|d_3'\circ\kappa_I\rangle|11\rangle
+|0\rangle|d_2'\rangle|d_3'\circ\kappa_I\rangle|00\rangle,\\
H_{6,1} &\rightarrow |0\rangle|d_2'\rangle|d_3'\circ d_4'\rangle|00\rangle,\\
H_{6,2} &\rightarrow |d_1'\rangle|d_2'\rangle|d_3'\circ0\rangle|22\rangle,\\
H_{6,3} &\rightarrow |d_1'\rangle|d_2'\rangle|0\circ d_4'\rangle|22\rangle,\\
H_{6,4} &\rightarrow |d_1'\rangle|0\rangle|d_3'\circ d_4'\rangle|22\rangle,\\
H_{7,1} &\rightarrow |0\rangle|\kappa_I\rangle|d_3'\circ\kappa_I\rangle|00\rangle,\\
H_{7,2} &\rightarrow |\kappa_I\rangle|d_2'\rangle|\kappa_I\circ0\rangle|11\rangle,\\
H_{7,3} &\rightarrow |d_1'\rangle|\kappa_I\rangle|0\circ\kappa_I\rangle|22\rangle,\\
H_{7,4} &\rightarrow |\kappa_I\rangle|0\rangle|\kappa_I\circ d_4'\rangle|11\rangle,\\
H_{8,1} &\rightarrow |d_1'\rangle|\kappa_I\rangle|d_3'\circ\kappa_I\rangle|22\rangle,\\
H_{8,2} &\rightarrow |\kappa_I\rangle|d_2'\rangle|\kappa_I\circ d_4'\rangle|11\rangle,
\end{aligned}
\]
where \(|00\rangle\), \(|11\rangle\), and \(|22\rangle\) mean \(|00\rangle_{a\tilde{c}}\), \(|11\rangle_{a\tilde{c}}\), and \(|22\rangle_{a\tilde{c}}\), respectively;
\[
|\beta_j'\rangle=\sum_{u=1}^{d_1-2}\omega_{d_1-1}^{ju}|u\rangle,
\qquad
|\gamma_m'\rangle=\sum_{u=0}^{d_1-3}\omega_{d_1-1}^{mu}|u+1\rangle.
\]
The symbol ``\(\circ\)'' has the same meaning as in the proof of Theorem~3.

\textbf{Step 2.} Charlie performs
\[
\begin{aligned}
\mathcal{M}_2\equiv \bigl\{&
M_{2,1}=P[|d_3'0\rangle_{\widetilde{C}};|2\rangle_{\tilde{c}}],\\
&M_{2,2}=P[|(1,\ldots,d_3')\circ0\rangle_{\widetilde{C}};|0\rangle_{\tilde{c}}],\\
&M_{2,3}=P[|d_3'\circ(0,d_4')\rangle_{\widetilde{C}};|1\rangle_{\tilde{c}}],\\
&M_{2,4}=P[|0\circ(0,\ldots,d_4'-1)\rangle_{\widetilde{C}};|2\rangle_{\tilde{c}}]\\
&\qquad\qquad
+P[|d_3'\circ(1,\ldots,d_4'-1)\rangle_{\widetilde{C}};|2\rangle_{\tilde{c}}],\\
&M_{2,5}=P[|00\rangle_{\widetilde{C}};(|0\rangle,|1\rangle)_{\tilde{c}}],\\
&M_{2,6}=P[|(1,\ldots,d_3'-1)\circ(1,\ldots,d_4'-1)\rangle_{\widetilde{C}};|0\rangle_{\tilde{c}}],\\
&M_{2,7}=P[|(0,\ldots,d_3'-1)\circ d_4'\rangle_{\widetilde{C}};|0\rangle_{\tilde{c}}],\\
&M_{2,8}=I-M_{2,1}-\cdots-M_{2,7}
\bigr\}.
\end{aligned}
\]
The corresponding subsets are
\[
\begin{aligned}
&M_{2,1}\Rightarrow H_{6,2},\\
&M_{2,2}\Rightarrow H_{2,3},H_{3,2},\\
&M_{2,3}\Rightarrow H_{4,1},H_{5,3},\\
&M_{2,4}\Rightarrow H_{2,1},H_{3,3},H_{4,2},H_{7,3},H_{8,1},\\
&M_{2,5}\Rightarrow H_{2,4},\\
&M_{2,6}\Rightarrow H_{1,3},\\
&M_{2,7}\Rightarrow H_{2,2}.
\end{aligned}
\]
The pairs \(\{H_{2,3},H_{3,2}\}\) and \(\{H_{4,1},H_{5,3}\}\) are LOCC distinguishable. For the set
\[
\{H_{2,1},H_{3,3},H_{4,2},H_{7,3},H_{8,1}\},
\]
Bob performs
\[
\mathcal{M}_{2}'\equiv \bigl\{
M_{2,1}'=|0\rangle_B\langle 0|,
\;
M_{2,2}'=I-M_{2,1}'
\bigr\}.
\]
If outcome \(M_{2,1}'\) occurs, the subset is \(H_{2,1}\); otherwise, the remaining possibilities are \(\{H_{3,3},H_{4,2},H_{7,3},H_{8,1}\}\). Charlie then performs
\[
\mathcal{M}_{2}''\equiv \bigl\{
M_{2,1}''=P[|00\rangle_{\widetilde{C}};|2\rangle_{\tilde{c}}],
\;
M_{2,2}''=I-M_{2,1}''
\bigr\}.
\]
If outcome \(M_{2,1}''\) occurs, the subset is \(H_{3,3}\). Otherwise, Bob performs
\[
\mathcal{M}_{2}'''\equiv \bigl\{
M_{2,1}'''=|d_2'\rangle_B\langle d_2'|,
\;
M_{2,2}'''=I-M_{2,1}'''
\bigr\}.
\]
If outcome \(M_{2,1}'''\) occurs, the subset is \(H_{4,2}\); otherwise, the remaining possibilities are \(\{H_{7,3},H_{8,1}\}\), which are LOCC distinguishable.

If outcome \(M_{2,8}\) occurs, the remaining 17 subsets remain. We then move to the next step.

\textbf{Step 3.} Bob performs
\[
\begin{aligned}
\mathcal{M}_3\equiv \bigl\{
M_{3,1}=|0\rangle_B\langle 0|+|d_2'\rangle_B\langle d_2'|,
\;
M_{3,2}=I-M_{3,1}
\bigr\}.
\end{aligned}
\]
The corresponding results are
\[
\begin{aligned}
&M_{3,1}\Rightarrow
\begin{cases}
H_{1,2},H_{3,1},H_{3,4},H_{4,3},H_{5,1},H_{5,4},\\
H_{6,1},H_{6,3},H_{6,4},H_{7,2},H_{7,4},H_{8,2},
\end{cases}\\
&M_{3,2}\Rightarrow
\begin{cases}
H_{1,1},H_{1,4},H_{4,4},H_{5,2},H_{7,1}.
\end{cases}
\end{aligned}
\]

\textbf{Step 4.} For the branch corresponding to \(M_{3,1}\), Charlie performs
\[
\begin{aligned}
\mathcal{M}_4\equiv \bigl\{&
M_{4,1}=P[|(1,\ldots,d_3'-1)\circ0\rangle_{\widetilde{C}};|1\rangle_{\tilde{c}}],\\
&M_{4,2}=P[|(1,\ldots,d_3'-1)\circ d_4'\rangle_{\widetilde{C}};|1\rangle_{\tilde{c}}],\\
&M_{4,3}=P[|d_3'd_4'\rangle_{\widetilde{C}};|2\rangle_{\tilde{c}}],\\
&M_{4,4}=P[|(1,\ldots,d_3'-1)\circ d_4'\rangle_{\widetilde{C}};|2\rangle_{\tilde{c}}],\\
&M_{4,5}=P[|(1,\ldots,d_3'-1)\circ(0,\ldots,d_4'-1)\rangle_{\widetilde{C}};|2\rangle_{\tilde{c}}],\\
&M_{4,6}=P[|0d_4'\rangle_{\widetilde{C}};(|1\rangle,|2\rangle)_{\tilde{c}}],\\
&M_{4,7}=I-M_{4,1}-\cdots-M_{4,6}
\bigr\}.
\end{aligned}
\]
The corresponding subsets are
\[
\begin{aligned}
&M_{4,1}\Rightarrow H_{7,2},\qquad &&M_{4,2}\Rightarrow H_{7,4},H_{8,2},\\
&M_{4,3}\Rightarrow H_{6,4},\qquad &&M_{4,4}\Rightarrow H_{4,3},\\
&M_{4,5}\Rightarrow H_{5,1},\qquad &&M_{4,6}\Rightarrow H_{3,4},H_{6,3}.
\end{aligned}
\]
The pairs \(\{H_{7,4},H_{8,2}\}\) and \(\{H_{3,4},H_{6,3}\}\) are LOCC distinguishable. If outcome \(M_{4,7}\) occurs, the remaining possibilities are
\[
\{H_{1,2},H_{3,1},H_{5,4},H_{6,1}\}.
\]

\textbf{Step 5.} Bob performs
\[
\begin{aligned}
\mathcal{M}_5\equiv \bigl\{
M_{5,1}=|0\rangle_B\langle 0|,
\;
M_{5,2}=I-M_{5,1}
\bigr\}.
\end{aligned}
\]
If outcome \(M_{5,1}\) occurs, the remaining possibilities are \(\{H_{1,2},H_{3,1}\}\); otherwise, they are \(\{H_{5,4},H_{6,1}\}\). Each pair is LOCC distinguishable.

\textbf{Step 4$'$.} For the branch corresponding to \(M_{3,2}\) in Step 3, Charlie performs
\[
\begin{aligned}
\mathcal{M}_6\equiv \bigl\{&
M_{6,1}=P[|d_3'\circ(1,\ldots,d_4'-1)\rangle_{\widetilde{C}};|0\rangle_{\tilde{c}}],\\
&M_{6,2}=P[|0\circ(1,\ldots,d_4')\rangle_{\widetilde{C}};|1\rangle_{\tilde{c}}],\\
&M_{6,3}=P[|(0,\ldots,d_3'-1)\circ d_4'\rangle_{\widetilde{C}};|2\rangle_{\tilde{c}}],\\
&M_{6,4}=P[|d_3'd_4'\rangle_{\widetilde{C}};(|0\rangle,|2\rangle)_{\tilde{c}}],\\
&M_{6,5}=I-M_{6,1}-\cdots-M_{6,4}
\bigr\}.
\end{aligned}
\]
The corresponding subsets are
\[
\begin{aligned}
&M_{6,1}\Rightarrow H_{7,1},\qquad &&M_{6,2}\Rightarrow H_{1,1},\\
&M_{6,3}\Rightarrow H_{5,2},\qquad &&M_{6,4}\Rightarrow H_{4,4},\\
&M_{6,5}\Rightarrow H_{1,4}.
\end{aligned}
\]

If any other outcomes occur in Step 1, similar protocols can be constructed to distinguish the corresponding subsets perfectly by LOCC.

\section{The proof of Theorem 7}\label{T7}

Let Alice and Charlie share an EPR state \(\vert \phi^{+}(2)\rangle\), while Alice and Dave share a maximally entangled state \(\vert \phi^{+}(3)\rangle\). The initial state is
\[
\begin{aligned}
|\varphi\rangle_{ABCD}\otimes|\phi^{+}(2)\rangle_{a_1c}\otimes|\phi^{+}(3)\rangle_{a_2v},
\end{aligned}
\]
where \(a_1\) and \(a_2\) are Alice's ancillary systems, \(c\) is Charlie's ancillary system, and \(v\) is Dave's ancillary system.

\textbf{Step 1.} Charlie and Dave perform the measurements
\[
\begin{aligned}
\mathcal{M}_1\equiv \bigl\{&
M_{1,1}=P\bigl[|0\rangle_C;|0\rangle_c\bigr]
+P\bigl[(|1\rangle,\dots,|d_3'\rangle)_C;|1\rangle_c\bigr],\\
&M_{1,2}=I-M_{1,1}
\bigr\},
\end{aligned}
\]
and
\[
\begin{aligned}
\mathcal{M}_2\equiv \bigl\{&
M_{2,1}=P\bigl[|0\rangle_D;|0\rangle_v\bigr]
+P\bigl[(|1\rangle,\dots,|d_4'-1\rangle)_D;|1\rangle_v\bigr]
+P\bigl[|d_4'\rangle_D;|2\rangle_v\bigr],\\
&M_{2,2}=P\bigl[|0\rangle_D;|1\rangle_v\bigr]
+P\bigl[(|1\rangle,\dots,|d_4'-1\rangle)_D;|2\rangle_v\bigr]
+P\bigl[|d_4'\rangle_D;|0\rangle_v\bigr],\\
&M_{2,3}=I-M_{2,1}-M_{2,2}
\bigr\},
\end{aligned}
\]
respectively. Suppose that the outcomes corresponding to \(M_{1,1}\) and \(M_{2,1}\) occur. Then
\[
\begin{aligned}
&H_{1,1}\rightarrow |\kappa_I\rangle_1|\kappa_I\rangle_2|0\rangle_3|\gamma_m'\rangle_4|00\rangle_{a_1c}|11\rangle_{a_2v}\\
&\qquad\qquad+|\kappa_I\rangle_1|\kappa_I\rangle_2|0\rangle_3|d_4'\rangle_4|00\rangle_{a_1c}|22\rangle_{a_2v},\\
&H_{1,2}\rightarrow |\kappa_I\rangle_1|0\rangle_2|\gamma_m\rangle_3|\kappa_I\rangle_4|11\rangle_{a_1c}|11\rangle_{a_2v},\\
&H_{1,3}\rightarrow |0\rangle_1|\gamma_m\rangle_2|\kappa_I\rangle_3|\kappa_I\rangle_4|11\rangle_{a_1c}|11\rangle_{a_2v},\\
&H_{1,4}\rightarrow |\gamma_m\rangle_1|\kappa_I\rangle_2|\kappa_I\rangle_3|0\rangle_4|11\rangle_{a_1c}|00\rangle_{a_2v},\\
&H_{2,1}\rightarrow |d_1'\rangle_1|0\rangle_2|0\rangle_3|0\rangle_4|00\rangle_{a_1c}|00\rangle_{a_2v}\\
&\qquad\qquad+|d_1'\rangle_1|0\rangle_2|0\rangle_3|\beta_j'\rangle_4|00\rangle_{a_1c}|11\rangle_{a_2v},\\
&H_{2,2}\rightarrow |0\rangle_1|0\rangle_2|\beta_j'\rangle_3|d_4'\rangle_4|11\rangle_{a_1c}|22\rangle_{a_2v}\\
&\qquad\qquad+|0\rangle_1|0\rangle_2|0\rangle_3|d_4'\rangle_4|00\rangle_{a_1c}|22\rangle_{a_2v},\\
&H_{2,3}\rightarrow |0\rangle_1|\beta_j\rangle_2|d_3'\rangle_3|0\rangle_4|11\rangle_{a_1c}|00\rangle_{a_2v},\\
&H_{2,4}\rightarrow |\beta_j\rangle_1|d_2'\rangle_2|0\rangle_3|0\rangle_4|00\rangle_{a_1c}|00\rangle_{a_2v},\\
&H_{3,1}\rightarrow |0\rangle_1|0\rangle_2|d_3'\rangle_3|\gamma_m'\rangle_4|11\rangle_{a_1c}|11\rangle_{a_2v}\\
&\qquad\qquad+|0\rangle_1|0\rangle_2|d_3'\rangle_3|d_4'\rangle_4|11\rangle_{a_1c}|22\rangle_{a_2v},\\
&H_{3,2}\rightarrow |0\rangle_1|d_2'\rangle_2|\gamma_m\rangle_3|0\rangle_4|11\rangle_{a_1c}|00\rangle_{a_2v},\\
&H_{3,3}\rightarrow |d_1'\rangle_1|\gamma_m\rangle_2|0\rangle_3|0\rangle_4|00\rangle_{a_1c}|00\rangle_{a_2v},\\
&H_{3,4}\rightarrow |\gamma_m\rangle_1|0\rangle_2|0\rangle_3|d_4'\rangle_4|00\rangle_{a_1c}|22\rangle_{a_2v},\\
&H_{4,1}\rightarrow |\kappa_I\rangle_1|d_2'\rangle_2|d_3'\rangle_3|0\rangle_4|11\rangle_{a_1c}|00\rangle_{a_2v}\\
&\qquad\qquad\pm|\kappa_I\rangle_1|d_2'\rangle_2|d_3'\rangle_3|d_4'\rangle_4|11\rangle_{a_1c}|22\rangle_{a_2v},\\
&H_{4,2}\rightarrow |d_1'\rangle_1|d_2'\rangle_2|0\rangle_3|\kappa_I\rangle_4|00\rangle_{a_1c}|11\rangle_{a_2v}\\
&\qquad\qquad\pm|d_1'\rangle_1|d_2'\rangle_2|d_3'\rangle_3|\kappa_I\rangle_4|11\rangle_{a_1c}|11\rangle_{a_2v},\\
&H_{4,3}\rightarrow |d_1'\rangle_1|0\pm d_2'\rangle_2|\kappa_I\rangle_3|d_4'\rangle_4|11\rangle_{a_1c}|22\rangle_{a_2v},\\
&H_{4,4}\rightarrow |0\pm d_1'\rangle_1|\kappa_I\rangle_2|d_3'\rangle_3|d_4'\rangle_4|11\rangle_{a_1c}|22\rangle_{a_2v},\\
&H_{5,1}\rightarrow |d_1'\rangle_1|d_2'\rangle_2|\kappa_I\rangle_3|\beta_j'\rangle_4|11\rangle_{a_1c}|11\rangle_{a_2v}\\
&\qquad\qquad+|d_1'\rangle_1|d_2'\rangle_2|\kappa_I\rangle_3|0\rangle_4|11\rangle_{a_1c}|00\rangle_{a_2v},\\
&H_{5,2}\rightarrow |d_1'\rangle_1|\kappa_I\rangle_2|\beta_j'\rangle_3|d_4'\rangle_4|11\rangle_{a_1c}|22\rangle_{a_2v}\\
&\qquad\qquad+|d_1'\rangle_1|\kappa_I\rangle_2|0\rangle_3|d_4'\rangle_4|00\rangle_{a_1c}|22\rangle_{a_2v},\\
&H_{5,3}\rightarrow |\kappa_I\rangle_1|\beta_j\rangle_2|d_3'\rangle_3|d_4'\rangle_4|11\rangle_{a_1c}|22\rangle_{a_2v},\\
&H_{5,4}\rightarrow |\beta_j\rangle_1|d_2'\rangle_2|d_3'\rangle_3|\kappa_I\rangle_4|11\rangle_{a_1c}|11\rangle_{a_2v},\\
&H_{6,1}\rightarrow |0\rangle_1|d_2'\rangle_2|d_3'\rangle_3|d_4'\rangle_4|11\rangle_{a_1c}|22\rangle_{a_2v},\\
&H_{6,2}\rightarrow |d_1'\rangle_1|d_2'\rangle_2|d_3'\rangle_3|0\rangle_4|11\rangle_{a_1c}|00\rangle_{a_2v},\\
&H_{6,3}\rightarrow |d_1'\rangle_1|d_2'\rangle_2|0\rangle_3|d_4'\rangle_4|00\rangle_{a_1c}|22\rangle_{a_2v},\\
&H_{6,4}\rightarrow |d_1'\rangle_1|0\rangle_2|d_3'\rangle_3|d_4'\rangle_4|11\rangle_{a_1c}|22\rangle_{a_2v},\\
&H_{7,1}\rightarrow |0\rangle_1|\kappa_I\rangle_2|d_3'\rangle_3|\kappa_I\rangle_4|11\rangle_{a_1c}|11\rangle_{a_2v},\\
&H_{7,2}\rightarrow |\kappa_I\rangle_1|d_2'\rangle_2|\kappa_I\rangle_3|0\rangle_4|11\rangle_{a_1c}|00\rangle_{a_2v},\\
&H_{7,3}\rightarrow |d_1'\rangle_1|\kappa_I\rangle_2|0\rangle_3|\kappa_I\rangle_4|00\rangle_{a_1c}|11\rangle_{a_2v},\\
&H_{7,4}\rightarrow |\kappa_I\rangle_1|0\rangle_2|\kappa_I\rangle_3|d_4'\rangle_4|11\rangle_{a_1c}|22\rangle_{a_2v},\\
&H_{8,1}\rightarrow |d_1'\rangle_1|\kappa_I\rangle_2|d_3'\rangle_3|\kappa_I\rangle_4|11\rangle_{a_1c}|11\rangle_{a_2v},\\
&H_{8,2}\rightarrow |\kappa_I\rangle_1|d_2'\rangle_2|\kappa_I\rangle_3|d_4'\rangle_4|11\rangle_{a_1c}|22\rangle_{a_2v}.
\end{aligned}
\]
Here
\[
|\beta_j'\rangle_k=\sum_{u=1}^{d_k-2}\omega_{d_k-1}^{ju}|u\rangle,
\qquad
|\gamma_m'\rangle_k=\sum_{u=0}^{d_k-3}\omega_{d_k-1}^{mu}|u+1\rangle
\]
for \(k=1,2,3,4\).

\textbf{Step 2.} Alice performs
\[
\begin{aligned}
\mathcal{M}_3\equiv \bigl\{&
M_{3,1}=P[|0\rangle_A;|1\rangle_{a_1};|0\rangle_{a_2}],\\
&M_{3,2}=P[(|0\rangle,\dots,|d_1'-1\rangle)_A;|0\rangle_{a_1};|0\rangle_{a_2}],\\
&M_{3,3}=P[(|1\rangle,\dots,|d_1'-1\rangle)_A;|1\rangle_{a_1};|2\rangle_{a_2}]\\
&\qquad\qquad+P[(|1\rangle,\dots,|d_1'\rangle)_A;|1\rangle_{a_1};|0\rangle_{a_2}]\\
&\qquad\qquad+P[|d_1'\rangle_A;|1\rangle_{a_1};|1\rangle_{a_2}]\\
&\qquad\qquad+P[|d_1'\rangle_A;|0\rangle_{a_1};|1\rangle_{a_2}]\\
&\qquad\qquad+P[|d_1'\rangle_A;|0\rangle_{a_1};|0\rangle_{a_2}],\\
&M_{3,4}=I-M_{3,1}-M_{3,2}-M_{3,3}
\bigr\}.
\end{aligned}
\]
If outcome \(M_{3,1}\) occurs, the remaining possibilities are \(\{H_{2,3},H_{3,2}\}\), which are LOCC distinguishable. If outcome \(M_{3,2}\) occurs, the subset is \(H_{2,4}\). The branches corresponding to \(M_{3,3}\) and \(M_{3,4}\) yield
\[
\begin{aligned}
&M_{3,3}\Rightarrow
\begin{cases}
H_{1,4},H_{2,1},H_{3,3},H_{4,1},H_{4,2},H_{5,1},H_{5,3},\\
H_{6,2},H_{7,2},H_{7,3},H_{7,4},H_{8,1},H_{8,2},
\end{cases}\\
&M_{3,4}\Rightarrow
\begin{cases}
H_{1,1},H_{1,2},H_{1,3},H_{2,2},H_{3,1},H_{3,4},H_{4,3},\\
H_{4,4},H_{5,2},H_{5,4},H_{6,1},H_{6,3},H_{6,4},H_{7,1}.
\end{cases}
\end{aligned}
\]

\textbf{Step 3.} To distinguish the states in the branches corresponding to \(M_{3,3}\) and \(M_{3,4}\), let Alice and Bob share a maximally entangled state \(|\phi^{+}(2)\rangle_{a_3b}\). Bob then performs
\[
\begin{aligned}
\mathcal{M}_4\equiv \bigl\{&
M_{4,1}=P\bigl[|0\rangle_B;|0\rangle_b\bigr]
+P\bigl[(|1\rangle,\dots,|d_2'\rangle)_B;|1\rangle_b\bigr],\\
&M_{4,2}=I-M_{4,1}
\bigr\}.
\end{aligned}
\]

\textbf{Step 4.} Suppose that outcome \(M_{4,1}\) occurs. For the branch corresponding to \(M_{3,3}\), the subsets are transformed into
\[
\begin{aligned}
&H_{1,4}\rightarrow |\gamma_m\rangle_1|\kappa_I\rangle_2|\kappa_I\rangle_3|0\rangle_4|11\rangle_{a_3b}|11\rangle_{a_1c}|00\rangle_{a_2v},\\
&H_{2,1}\rightarrow |d_1'\rangle_1|0\rangle_2|0\rangle_3|\beta_j'\rangle_4|00\rangle_{a_3b}|00\rangle_{a_1c}|11\rangle_{a_2v}\\
&\qquad\qquad+|d_1'\rangle_1|0\rangle_2|0\rangle_3|0\rangle_4|00\rangle_{a_3b}|00\rangle_{a_1c}|00\rangle_{a_2v},\\
&H_{3,3}\rightarrow |d_1'\rangle_1|\gamma_m\rangle_2|0\rangle_3|0\rangle_4|11\rangle_{a_3b}|00\rangle_{a_1c}|00\rangle_{a_2v},\\
&H_{4,1}\rightarrow |\kappa_I\rangle_1|d_2'\rangle_2|d_3'\rangle_3|0\rangle_4|11\rangle_{a_3b}|11\rangle_{a_1c}|00\rangle_{a_2v}\\
&\qquad\qquad\pm|\kappa_I\rangle_1|d_2'\rangle_2|d_3'\rangle_3|d_4'\rangle_4|11\rangle_{a_3b}|11\rangle_{a_1c}|22\rangle_{a_2v},\\
&H_{4,2}\rightarrow |d_1'\rangle_1|d_2'\rangle_2|0\rangle_3|\kappa_I\rangle_4|11\rangle_{a_3b}|00\rangle_{a_1c}|11\rangle_{a_2v}\\
&\qquad\qquad\pm|d_1'\rangle_1|d_2'\rangle_2|d_3'\rangle_3|\kappa_I\rangle_4|11\rangle_{a_3b}|11\rangle_{a_1c}|11\rangle_{a_2v},\\
&H_{5,1}\rightarrow |d_1'\rangle_1|d_2'\rangle_2|\kappa_I\rangle_3|\beta_j'\rangle_4|11\rangle_{a_3b}|11\rangle_{a_1c}|11\rangle_{a_2v}\\
&\qquad\qquad+|d_1'\rangle_1|d_2'\rangle_2|\kappa_I\rangle_3|0\rangle_4|11\rangle_{a_3b}|11\rangle_{a_1c}|00\rangle_{a_2v},\\
&H_{5,3}\rightarrow |\kappa_I\rangle_1|\beta_j'\rangle_2|d_3'\rangle_3|d_4'\rangle_4|11\rangle_{a_3b}|11\rangle_{a_1c}|22\rangle_{a_2v}\\
&\qquad\qquad+|\kappa_I\rangle_1|0\rangle_2|d_3'\rangle_3|d_4'\rangle_4|00\rangle_{a_3b}|11\rangle_{a_1c}|22\rangle_{a_2v},\\
&H_{6,2}\rightarrow |d_1'\rangle_1|d_2'\rangle_2|d_3'\rangle_3|0\rangle_4|11\rangle_{a_3b}|11\rangle_{a_1c}|00\rangle_{a_2v},\\
&H_{7,2}\rightarrow |\kappa_I\rangle_1|d_2'\rangle_2|\kappa_I\rangle_3|0\rangle_4|11\rangle_{a_3b}|11\rangle_{a_1c}|00\rangle_{a_2v},\\
&H_{7,3}\rightarrow |d_1'\rangle_1|\kappa_I\rangle_2|0\rangle_3|\kappa_I\rangle_4|11\rangle_{a_3b}|00\rangle_{a_1c}|11\rangle_{a_2v},\\
&H_{7,4}\rightarrow |\kappa_I\rangle_1|0\rangle_2|\kappa_I\rangle_3|d_4'\rangle_4|00\rangle_{a_3b}|11\rangle_{a_1c}|22\rangle_{a_2v},\\
&H_{8,1}\rightarrow |d_1'\rangle_1|\kappa_I\rangle_2|d_3'\rangle_3|\kappa_I\rangle_4|11\rangle_{a_3b}|11\rangle_{a_1c}|11\rangle_{a_2v},\\
&H_{8,2}\rightarrow |\kappa_I\rangle_1|d_2'\rangle_2|\kappa_I\rangle_3|d_4'\rangle_4|11\rangle_{a_3b}|11\rangle_{a_1c}|22\rangle_{a_2v},
\end{aligned}
\]
where
\[
|\beta_j'\rangle_k=\sum_{u=1}^{d_k-2}\omega_{d_k-1}^{ju}|u\rangle
\]
for \(k=1,2,3,4\).

Alice then performs
\[
\begin{aligned}
\mathcal{M}_5\equiv \bigl\{&
M_{5,1}=P[|d_1'\rangle_A;|0\rangle_{a_1};|0\rangle_{a_2};|1\rangle_{a_3}],\\
&M_{5,2}=P[|d_1'\rangle_A;|0\rangle_{a_1};(|0\rangle,|1\rangle)_{a_2};|0\rangle_{a_3}],\\
&M_{5,3}=I-M_{5,1}-M_{5,2}
\bigr\}.
\end{aligned}
\]
If outcomes \(M_{5,1}\) and \(M_{5,2}\) occur, the corresponding subsets are \(H_{3,3}\) and \(H_{2,1}\), respectively. If outcome \(M_{5,3}\) occurs, we proceed.

\textbf{Step 5.} Bob performs
\[
\mathcal{M}_6\equiv \bigl\{
M_{6,1}=|d_2'\rangle_B\langle d_2'|,
\;
M_{6,2}=I-M_{6,1}
\bigr\}.
\]
The corresponding results are
\[
\begin{aligned}
&M_{6,1}\Rightarrow H_{4,1},H_{4,2},H_{5,1},H_{6,2},H_{7,2},H_{8,2},\\
&M_{6,2}\Rightarrow H_{1,4},H_{5,3},H_{7,3},H_{7,4},H_{8,1}.
\end{aligned}
\]

\textbf{Step 6.} For the branch corresponding to \(M_{6,1}\), Alice performs
\[
\mathcal{M}_7\equiv \bigl\{
M_{7,1}=|d_1'\rangle_A\langle d_1'|,
\;
M_{7,2}=I-M_{7,1}
\bigr\}.
\]
If outcome \(M_{7,1}\) occurs, the remaining possibilities are \(\{H_{4,2},H_{5,1},H_{6,2}\}\). Charlie then performs
\[
\mathcal{M}_8\equiv \bigl\{
M_{8,1}=|1\rangle_C\langle 1|+\cdots+|d_3'-1\rangle_C\langle d_3'-1|,
\;
M_{8,2}=I-M_{8,1}
\bigr\}.
\]
If outcome \(M_{8,1}\) occurs, the subset is \(H_{5,1}\); otherwise, the remaining possibilities are \(\{H_{4,2},H_{6,2}\}\), which are LOCC distinguishable.

If outcome \(M_{7,2}\) occurs, the remaining possibilities are \(\{H_{7,2},H_{8,2},H_{4,1}\}\). Applying the same measurement \(\mathcal{M}_8\), outcome \(M_{8,1}\) leaves \(\{H_{7,2},H_{8,2}\}\), which are LOCC distinguishable, while outcome \(M_{8,2}\) leaves \(H_{4,1}\).

For the branch corresponding to \(M_{6,2}\), Alice performs
\[
\begin{aligned}
\mathcal{M}_9\equiv \bigl\{&
M_{9,1}=P[(|1\rangle,\dots,|d_1'\rangle)_A;|1\rangle_{a_1};|0\rangle_{a_2};|1\rangle_{a_3}],\\
&M_{9,2}=P[|d_1'\rangle_A;|0\rangle_{a_1};|1\rangle_{a_2};|1\rangle_{a_3}],\\
&M_{9,3}=P[|d_1'\rangle_A;|1\rangle_{a_1};|1\rangle_{a_2};|1\rangle_{a_3}],\\
&M_{9,4}=I-M_{9,1}-M_{9,2}-M_{9,3}
\bigr\}.
\end{aligned}
\]
The corresponding subsets are
\[
\begin{aligned}
&M_{9,1}\Rightarrow H_{1,4},\qquad &&M_{9,2}\Rightarrow H_{7,3},\\
&M_{9,3}\Rightarrow H_{8,1},\qquad &&M_{9,4}\Rightarrow H_{7,4},H_{5,3}.
\end{aligned}
\]
The pair \(\{H_{7,4},H_{5,3}\}\) is LOCC distinguishable.

\textbf{Step 4$'$.} Suppose again that outcome \(M_{4,1}\) occurs, but now for the branch corresponding to \(M_{3,4}\). The subsets are transformed into
\[
\begin{aligned}
&H_{1,1}\rightarrow |\kappa_I\rangle_1|\kappa_I\rangle_2|0\rangle_3|\gamma_m'\rangle_4|11\rangle_{a_3b}|00\rangle_{a_1c}|11\rangle_{a_2v}\\
&\qquad\qquad+|\kappa_I\rangle_1|\kappa_I\rangle_2|0\rangle_3|d_4'\rangle_4|11\rangle_{a_3b}|00\rangle_{a_1c}|22\rangle_{a_2v},\\
&H_{1,2}\rightarrow |\kappa_I\rangle_1|0\rangle_2|\gamma_m\rangle_3|\kappa_I\rangle_4|00\rangle_{a_3b}|11\rangle_{a_1c}|11\rangle_{a_2v},\\
&H_{1,3}\rightarrow |0\rangle_1|\gamma_m\rangle_2|\kappa_I\rangle_3|\kappa_I\rangle_4|11\rangle_{a_3b}|11\rangle_{a_1c}|11\rangle_{a_2v},\\
&H_{2,2}\rightarrow |0\rangle_1|0\rangle_2|\beta_j'\rangle_3|d_4'\rangle_4|00\rangle_{a_3b}|11\rangle_{a_1c}|22\rangle_{a_2v}\\
&\qquad\qquad+|0\rangle_1|0\rangle_2|0\rangle_3|d_4'\rangle_4|00\rangle_{a_3b}|00\rangle_{a_1c}|22\rangle_{a_2v},\\
&H_{3,1}\rightarrow |0\rangle_1|0\rangle_2|d_3'\rangle_3|\gamma_m'\rangle_4|00\rangle_{a_3b}|11\rangle_{a_1c}|11\rangle_{a_2v}\\
&\qquad\qquad+|0\rangle_1|0\rangle_2|d_3'\rangle_3|d_4'\rangle_4|00\rangle_{a_3b}|11\rangle_{a_1c}|22\rangle_{a_2v},\\
&H_{3,4}\rightarrow |\gamma_m\rangle_1|0\rangle_2|0\rangle_3|d_4'\rangle_4|00\rangle_{a_3b}|00\rangle_{a_1c}|22\rangle_{a_2v},\\
&H_{4,3}\rightarrow |d_1'\rangle_1|0\rangle_2|\kappa_I\rangle_3|d_4'\rangle_4|00\rangle_{a_3b}|11\rangle_{a_1c}|22\rangle_{a_2v}\\
&\qquad\qquad\pm|d_1'\rangle_1|d_2'\rangle_2|\kappa_I\rangle_3|d_4'\rangle_4|11\rangle_{a_3b}|11\rangle_{a_1c}|22\rangle_{a_2v},\\
&H_{4,4}\rightarrow |0\pm d_1'\rangle_1|\kappa_I\rangle_2|d_3'\rangle_3|d_4'\rangle_4|11\rangle_{a_3b}|11\rangle_{a_1c}|22\rangle_{a_2v},\\
&H_{5,2}\rightarrow |d_1'\rangle_1|\kappa_I\rangle_2|\beta_j'\rangle_3|d_4'\rangle_4|11\rangle_{a_3b}|11\rangle_{a_1c}|22\rangle_{a_2v}\\
&\qquad\qquad+|d_1'\rangle_1|\kappa_I\rangle_2|0\rangle_3|d_4'\rangle_4|11\rangle_{a_3b}|00\rangle_{a_1c}|22\rangle_{a_2v},\\
&H_{5,4}\rightarrow |\beta_j\rangle_1|d_2'\rangle_2|d_3'\rangle_3|\kappa_I\rangle_4|11\rangle_{a_3b}|11\rangle_{a_1c}|11\rangle_{a_2v},\\
&H_{6,1}\rightarrow |0\rangle_1|d_2'\rangle_2|d_3'\rangle_3|d_4'\rangle_4|11\rangle_{a_3b}|11\rangle_{a_1c}|22\rangle_{a_2v},\\
&H_{6,3}\rightarrow |d_1'\rangle_1|d_2'\rangle_2|0\rangle_3|d_4'\rangle_4|11\rangle_{a_3b}|00\rangle_{a_1c}|22\rangle_{a_2v},\\
&H_{6,4}\rightarrow |d_1'\rangle_1|0\rangle_2|d_3'\rangle_3|d_4'\rangle_4|00\rangle_{a_3b}|11\rangle_{a_1c}|22\rangle_{a_2v},\\
&H_{7,1}\rightarrow |0\rangle_1|\kappa_I\rangle_2|d_3'\rangle_3|\kappa_I\rangle_4|11\rangle_{a_3b}|11\rangle_{a_1c}|11\rangle_{a_2v},
\end{aligned}
\]
where
\[
|\beta_j'\rangle_k=\sum_{u=1}^{d_k-2}\omega_{d_k-1}^{ju}|u\rangle,
\qquad
|\gamma_m'\rangle_k=\sum_{u=0}^{d_k-3}\omega_{d_k-1}^{mu}|u+1\rangle
\]
for \(k=1,2,3,4\).

Alice now performs
\[
\begin{aligned}
\mathcal{M}_{10}\equiv \bigl\{&
M_{10,1}=P[(|1\rangle,\dots,|d_1'-1\rangle)_A;|1\rangle_{a_1};|1\rangle_{a_2};|0\rangle_{a_3}],\\
&M_{10,2}=P[(|0\rangle,\dots,|d_1'-1\rangle)_A;|1\rangle_{a_1};|1\rangle_{a_2};|1\rangle_{a_3}],\\
&M_{10,3}=I-M_{10,1}-M_{10,2}
\bigr\}.
\end{aligned}
\]
If outcome \(M_{10,1}\) occurs, the subset is \(H_{1,2}\). If outcome \(M_{10,2}\) occurs, the remaining possibilities are \(\{H_{1,3},H_{5,4},H_{7,1}\}\). Charlie then performs
\[
\mathcal{M}_{11}\equiv \bigl\{
M_{11,1}=|1\rangle_C\langle 1|+\cdots+|d_3'-1\rangle_C\langle d_3'-1|,
\;
M_{11,2}=I-M_{11,1}
\bigr\}.
\]
If outcome \(M_{11,1}\) occurs, the subset is \(H_{1,3}\); otherwise, the remaining possibilities are \(\{H_{5,4},H_{7,1}\}\), which are LOCC distinguishable. If outcome \(M_{10,3}\) occurs, the state belongs to one of the remaining 10 subsets.

\textbf{Step 5$'$.} Bob performs
\[
\begin{aligned}
\mathcal{M}_{12}\equiv \bigl\{&
M_{12,1}=|1\rangle_B\langle 1|+\cdots+|d_2'-1\rangle_B\langle d_2'-1|,\\
&M_{12,2}=I-M_{12,1}
\bigr\}.
\end{aligned}
\]
The corresponding results are
\[
\begin{aligned}
&M_{12,1}\Rightarrow H_{1,1},H_{4,4},H_{5,2},\\
&M_{12,2}\Rightarrow H_{2,2},H_{3,1},H_{3,4},H_{4,3},H_{6,1},H_{6,3},H_{6,4}.
\end{aligned}
\]

\textbf{Step 6$'$.} For the branch corresponding to \(M_{12,1}\), Alice performs
\[
\begin{aligned}
\mathcal{M}_{13}\equiv \bigl\{&
M_{13,1}=|1\rangle_A\langle 1|+\cdots+|d_1'-1\rangle_A\langle d_1'-1|,\\
&M_{13,2}=I-M_{13,1}
\bigr\}.
\end{aligned}
\]
If outcome \(M_{13,1}\) occurs, the subset is \(H_{1,1}\); otherwise, the remaining possibilities are \(\{H_{4,4},H_{5,2}\}\), which are LOCC distinguishable.

For the branch corresponding to \(M_{12,2}\), Charlie performs
\[
\begin{aligned}
\mathcal{M}_{14}\equiv \bigl\{
M_{14,1}=|d_3'\rangle_C\langle d_3'|,
\;
M_{14,2}=I-M_{14,1}
\bigr\}.
\end{aligned}
\]

If outcome \(M_{14,1}\) occurs, the remaining possibilities are \(\{H_{3,1},H_{6,1},H_{6,4}\}\). Alice then performs
\[
\begin{aligned}
\mathcal{M}_{15}\equiv \bigl\{&
M_{15,1}=P[|0\rangle_A;|1\rangle_{a_1};|2\rangle_{a_2};|1\rangle_{a_3}],\\
&M_{15,2}=P[|d_1'\rangle_A;|1\rangle_{a_1};|2\rangle_{a_2};|0\rangle_{a_3}],\\
&M_{15,3}=I-M_{15,2}-M_{15,1}
\bigr\}.
\end{aligned}
\]
If outcome \(M_{15,1}\) occurs, the subset is \(H_{6,1}\); if outcome \(M_{15,2}\) occurs, the subset is \(H_{6,4}\); if outcome \(M_{15,3}\) occurs, the subset is \(H_{3,1}\).

If outcome \(M_{14,2}\) occurs, the remaining possibilities are \(\{H_{2,2},H_{3,4},H_{4,3},H_{6,3}\}\). Alice performs
\[
\begin{aligned}
\mathcal{M}_{16}\equiv \bigl\{&
M_{16,1}=P[(|1\rangle,\dots,|d_1'\rangle)_A;|0\rangle_{a_1};|2\rangle_{a_2};|0\rangle_{a_3}],\\
&M_{16,2}=P[|d_1'\rangle_A;|0\rangle_{a_1};|2\rangle_{a_2};|1\rangle_{a_3}],\\
&M_{16,3}=P[|0\rangle_A;(|0\rangle,|1\rangle)_{a_1};|2\rangle_{a_2};|0\rangle_{a_3}],\\
&M_{16,4}=I-M_{16,3}-M_{16,2}-M_{16,1}
\bigr\}.
\end{aligned}
\]
The corresponding subsets are
\[
\begin{aligned}
&M_{16,1}\Rightarrow H_{3,4},\qquad &&M_{16,2}\Rightarrow H_{6,3},\\
&M_{16,3}\Rightarrow H_{2,2},\qquad &&M_{16,4}\Rightarrow H_{4,3}.
\end{aligned}
\]

If other outcomes occur in Steps 1 or 3, similar protocols can again be constructed to distinguish the corresponding subsets perfectly by LOCC.

\section{The proof of Theorem 8}\label{T8}

Using the entanglement resources \(\vert \phi^{+}(2)\rangle_{ae_1}\), \(\vert \phi^{+}(2)\rangle_{be_2}\), \(\vert \phi^{+}(2)\rangle_{ce_3}\), and \(\vert \phi^{+}(2)\rangle_{ve_4}\), the initial state is
\[
\begin{aligned}
&\left|\psi\right\rangle_{ABCDE}
\otimes\vert \phi^{+}(2)\rangle_{ae_1}
\otimes\vert \phi^{+}(2)\rangle_{be_2}\\
&\qquad\qquad\otimes\vert \phi^{+}(2)\rangle_{ce_3}
\otimes\vert \phi^{+}(2)\rangle_{ve_4},
\end{aligned}
\]
where \(a\), \(b\), \(c\), and \(v\) are the ancillary systems of Alice, Bob, Charlie, and Dave, respectively, and \(e_1,e_2,e_3,e_4\) are Eve's ancillary systems.

\textbf{Step 1.} Alice, Bob, Charlie, and Dave perform the measurements
\[
\begin{aligned}
\mathcal{M}_1\equiv \bigl\{&
M_{1,1}=P\bigl[|0\rangle_A;|0\rangle_a\bigr]
+P\bigl[(|1\rangle,\dots,|d_1'\rangle)_A;|1\rangle_a\bigr],\\
&M_{1,2}=I-M_{1,1}
\bigr\},
\end{aligned}
\]
\[
\begin{aligned}
\mathcal{M}_2\equiv \bigl\{&
M_{2,1}=P\bigl[|0\rangle_B;|0\rangle_b\bigr]
+P\bigl[(|1\rangle,\dots,|d_2'\rangle)_B;|1\rangle_b\bigr],\\
&M_{2,2}=I-M_{2,1}
\bigr\},
\end{aligned}
\]
\[
\begin{aligned}
\mathcal{M}_3\equiv \bigl\{&
M_{3,1}=P\bigl[|0\rangle_C;|0\rangle_c\bigr]
+P\bigl[(|1\rangle,\dots,|d_3'\rangle)_C;|1\rangle_c\bigr],\\
&M_{3,2}=I-M_{3,1}
\bigr\},
\end{aligned}
\]
and
\[
\begin{aligned}
\mathcal{M}_4\equiv \bigl\{&
M_{4,1}=P\bigl[|0\rangle_D;|0\rangle_v\bigr]
+P\bigl[(|1\rangle,\dots,|d_4'\rangle)_D;|1\rangle_v\bigr],\\
&M_{4,2}=I-M_{4,1}
\bigr\},
\end{aligned}
\]
respectively. Suppose that outcomes \(M_{1,1}\), \(M_{2,1}\), \(M_{3,1}\), and \(M_{4,1}\) occur. Then the resulting post-measurement states are
\[
\begin{aligned}
&H_1 \to \vert \alpha_i' \rangle_1 \vert \alpha_i' \rangle_2 \vert p \rangle_3\vert 0 \rangle_4\vert 0 \rangle_5\vert 11 \rangle_{ae_1} \vert11 \rangle_{be_2} \vert 11 \rangle_{ce_3} \vert 00 \rangle_{ve_4}\\
&\qquad+\vert \alpha_i' \rangle_1 \vert 0 \rangle_2 \vert p \rangle_3 \vert 0 \rangle_4\vert 0 \rangle_5\vert 11 \rangle_{ae_1} \vert00 \rangle_{be_2} \vert 11 \rangle_{ce_3}\vert 00 \rangle_{ve_4}\\
&\qquad+\vert 0 \rangle_1\vert \alpha_i' \rangle_2 \vert p \rangle_3 \vert 0 \rangle_4\vert 0 \rangle_5\vert 00 \rangle_{ae_1} \vert11 \rangle_{be_2} \vert 11 \rangle_{ce_3}\vert 00 \rangle_{ve_4}\\
&\qquad+\vert 0 \rangle_1\vert 0 \rangle_2 \vert p \rangle_3 \vert 0 \rangle_4\vert 0 \rangle_5\vert 00 \rangle_{ae_1} \vert00 \rangle_{be_2} \vert 11 \rangle_{ce_3}\vert 00 \rangle_{ve_4},\\
&H_2 \to \vert \alpha_i' \rangle_1 \vert p \rangle_2\vert 0 \rangle_3\vert 0 \rangle_4\vert \alpha_i \rangle_5 \vert 11 \rangle_{ae_1} \vert11 \rangle_{be_2} \vert 00 \rangle_{ce_3} \vert 00\rangle_{ve_4}\\
&\qquad+\vert 0 \rangle_1\vert p \rangle_2\vert 0 \rangle_3\vert 0 \rangle_4\vert\alpha_i \rangle_5 \vert 00 \rangle_{ae_1} \vert11\rangle_{be_2} \vert 00 \rangle_{ce_3}\vert 00 \rangle_{ve_4},\\
&H_3 \to \vert p \rangle_1\vert 0 \rangle_2\vert 0 \rangle_3\vert \alpha_i' \rangle_4\vert \alpha_i \rangle_5 \vert 11 \rangle_{ae_1}\vert00 \rangle_{be_2}\vert 00 \rangle_{ce_3}\vert 11\rangle_{ve_4}\\
&\qquad+\vert p \rangle_1\vert 0 \rangle_2\vert 0 \rangle_3\vert 0 \rangle_4\vert\alpha_i \rangle_5 \vert11\rangle_{ae_1}\vert 00 \rangle_{be_2}\vert 00 \rangle_{ce_3}\vert 00 \rangle_{ve_4},\\
&H_4 \to \vert 0 \rangle_1\vert 0 \rangle_2\vert \alpha_i' \rangle_3\vert \alpha_i' \rangle_4\vert p \rangle_5 \vert 00 \rangle_{ae_1}\vert00 \rangle_{be_2}\vert 11 \rangle_{ce_3}\vert 11\rangle_{ve_4}\\
&\qquad+\vert 0 \rangle_1\vert 0 \rangle_2\vert \alpha_i' \rangle_3\vert 0 \rangle_4\vert p \rangle_5\vert 00 \rangle_{ae_1}\vert00 \rangle_{be_2}\vert 11 \rangle_{ce_3}\vert 00 \rangle_{ve_4}\\
&\qquad+\vert 0 \rangle_1\vert 0 \rangle_2\vert 0 \rangle_3\vert \alpha_i' \rangle_4\vert p \rangle_5\vert 00 \rangle_{ae_1}\vert00 \rangle_{be_2}\vert 00 \rangle_{ce_3}\vert 11 \rangle_{ve_4}\\
&\qquad+\vert 0 \rangle_1\vert 0 \rangle_2\vert 0 \rangle_3\vert 0 \rangle_4\vert p \rangle_5\vert 00 \rangle_{ae_1}\vert00 \rangle_{be_2}\vert 00\rangle_{ce_3}\vert 00 \rangle_{ve_4},\\
&H_5 \to \vert 0 \rangle_1\vert \alpha_i' \rangle_2\vert \alpha_i' \rangle_3\vert p \rangle_4\vert 0 \rangle_5\vert 00 \rangle_{ae_1}\vert11 \rangle_{be_2}\vert 11 \rangle_{ce_3}\vert 11\rangle_{ve_4}\\
&\qquad+\vert 0 \rangle_1\vert \alpha_i' \rangle_2\vert 0 \rangle_3\vert p \rangle_4\vert 0 \rangle_5\vert 00 \rangle_{ae_1}\vert11 \rangle_{be_2}\vert 00 \rangle_{ce_3}\vert 11 \rangle_{ve_4}\\
&\qquad+\vert 0 \rangle_1\vert 0 \rangle_2\vert \alpha_i' \rangle_3\vert p \rangle_4\vert 0 \rangle_5\vert 00 \rangle_{ae_1}\vert00 \rangle_{be_2}\vert 11 \rangle_{ce_3}\vert 11 \rangle_{ve_4}\\
&\qquad+\vert 0 \rangle_1\vert 0 \rangle_2\vert 0 \rangle_3\vert p \rangle_4\vert 0 \rangle_5\vert 00 \rangle_{ae_1}\vert00 \rangle_{be_2}\vert 00\rangle_{ce_3}\vert 11 \rangle_{ve_4},\\
&H_6 \to \vert \gamma_m \rangle_1\vert 0 \rangle_2\vert p \rangle_3\vert 0 \rangle_4\vert 1 \rangle_5\vert 11 \rangle_{ae_1}\vert00 \rangle_{be_2}\vert 11 \rangle_{ce_3}\vert 00\rangle_{ve_4},\\
&H_7 \to \vert 0 \rangle_1\vert p \rangle_2\vert 0 \rangle_3\vert 1 \rangle_4\vert \gamma_m \rangle_5\vert 00 \rangle_{ae_1}\vert11 \rangle_{be_2}\vert 00 \rangle_{ce_3}\vert 11\rangle_{ve_4},\\
&H_8 \to \vert p \rangle_1\vert 0 \rangle_2\vert 1 \rangle_3\vert \gamma_m \rangle_4\vert 0 \rangle_5\vert 11 \rangle_{ae_1}\vert00 \rangle_{be_2}\vert 11 \rangle_{ce_3}\vert 11\rangle_{ve_4},\\
&H_9 \to \vert 0 \rangle_1\vert 1 \rangle_2\vert \gamma_m \rangle_3\vert 0 \rangle_4\vert p \rangle_5\vert 00 \rangle_{ae_1}\vert11 \rangle_{be_2}\vert 11 \rangle_{ce_3}\vert 00\rangle_{ve_4},\\
&H_{10} \to \vert 1 \rangle_1\vert \gamma_m \rangle_2\vert 0 \rangle_3\vert p \rangle_4\vert 0 \rangle_5\vert 11 \rangle_{ae_1}\vert11 \rangle_{be_2}\vert 00 \rangle_{ce_3}\vert 11\rangle_{ve_4},
\end{aligned}
\]
where
\[
|\alpha_i'\rangle_k=\sum_{u=1}^{d_k-1}\omega_{d_k}^{iu}|u\rangle
\]
for \(k\in\{1,2,3,4,5\}\).

\textbf{Step 2.} Eve performs
\[
\begin{aligned}
\mathcal{M}_5\equiv \bigl\{&
M_{5,1}=P[(|1\rangle,\dots,|d_5'\rangle)_E;|1\rangle_{e_1};|0\rangle_{e_2};|1\rangle_{e_3};|0\rangle_{e_4}],\\
&M_{5,2}=P[(|1\rangle,\dots,|d_5'\rangle)_E;|0\rangle_{e_1};|1\rangle_{e_2};|0\rangle_{e_3};|1\rangle_{e_4}],\\
&M_{5,3}=P[|0\rangle_E;|1\rangle_{e_1};|0\rangle_{e_2};|1\rangle_{e_3};|1\rangle_{e_4}],\\
&M_{5,4}=P[(|1\rangle,\dots,|d_5'\rangle)_E;|0\rangle_{e_1};|1\rangle_{e_2};|1\rangle_{e_3};|0\rangle_{e_4}],\\
&M_{5,5}=P[|0\rangle_E;|1\rangle_{e_1};|1\rangle_{e_2};|0\rangle_{e_3};|1\rangle_{e_4}],\\
&M_{5,6}=P[(|1\rangle,\dots,|d_5'\rangle)_E;|0\rangle_{e_1};|0\rangle_{e_2};I_{e_3};I_{e_4}],\\
&M_{5,7}=P[(|0\rangle,\dots,|d_5'\rangle)_E;|1\rangle_{e_1};|0\rangle_{e_2};|0\rangle_{e_3};I_{e_4}],\\
&M_{5,8}=P[|0\rangle_E;I_{e_1};I_{e_2};|1\rangle_{e_3};|0\rangle_{e_4}],\\
&M_{5,9}=P[(|0\rangle,\dots,|d_5'\rangle)_E;I_{e_1};|1\rangle_{e_2};|0\rangle_{e_3};|0\rangle_{e_4}],\\
&M_{5,10}=I-M_{5,1}-\cdots-M_{5,9}
\bigr\}.
\end{aligned}
\]
The corresponding subsets are
\[
\begin{aligned}
&M_{5,1}\Rightarrow H_6,\qquad &&M_{5,2}\Rightarrow H_7,\\
&M_{5,3}\Rightarrow H_8,\qquad &&M_{5,4}\Rightarrow H_9,\\
&M_{5,5}\Rightarrow H_{10},\qquad &&M_{5,6}\Rightarrow H_4,\\
&M_{5,7}\Rightarrow H_3,\qquad &&M_{5,8}\Rightarrow H_1,\\
&M_{5,9}\Rightarrow H_2,\qquad &&M_{5,10}\Rightarrow H_5.
\end{aligned}
\]

If other outcomes occur in Step 1, similar LOCC protocols can be constructed.

\section{The proof of Theorem 9}\label{T9}

Let \(\vert G\rangle=\vert 000\rangle+\vert 111\rangle\). Using the entanglement resources \(\vert G\rangle_{av_1e_1}\), \(\vert G\rangle_{bv_2e_2}\), and \(\vert G\rangle_{cv_3e_3}\), the initial state is
\[
\begin{aligned}
\left|\psi\right\rangle_{ABCDE}\otimes \vert G\rangle_{av_1e_1}\otimes \vert G\rangle_{bv_2e_2}\otimes \vert G\rangle_{cv_3e_3},
\end{aligned}
\]
where \(a\), \(b\), and \(c\) are the ancillary systems of Alice, Bob, and Charlie, respectively; \(v_1,v_2,v_3\) are Dave's ancillary systems; and \(e_1,e_2,e_3\) are Eve's ancillary systems.

\textbf{Step 1.} Alice, Bob, and Charlie perform the measurements
\[
\begin{aligned}
\mathcal{M}_1\equiv \bigl\{&
M_{1,1}=P\bigl[|0\rangle_A;|0\rangle_a\bigr]
+P\bigl[(|1\rangle,\dots,|d_1'\rangle)_A;|1\rangle_a\bigr],\\
&M_{1,2}=I-M_{1,1}
\bigr\},
\end{aligned}
\]
\[
\begin{aligned}
\mathcal{M}_2\equiv \bigl\{&
M_{2,1}=P\bigl[|0\rangle_B;|0\rangle_b\bigr]
+P\bigl[(|1\rangle,\dots,|d_2'\rangle)_B;|1\rangle_b\bigr],\\
&M_{2,2}=I-M_{2,1}
\bigr\},
\end{aligned}
\]
and
\[
\begin{aligned}
\mathcal{M}_3\equiv \bigl\{&
M_{3,1}=P\bigl[|0\rangle_C;|0\rangle_c\bigr]
+P\bigl[(|1\rangle,\dots,|d_3'\rangle)_C;|1\rangle_c\bigr],\\
&M_{3,2}=I-M_{3,1}
\bigr\},
\end{aligned}
\]
respectively. Suppose that outcomes \(M_{1,1}\), \(M_{2,1}\), and \(M_{3,1}\) occur. Then
\[
\begin{aligned}
&H_1 \to \vert \alpha_i' \rangle_1 \vert \alpha_i' \rangle_2 \vert p \rangle_3\vert 0 \rangle_4\vert 0 \rangle_5\vert 111 \rangle_{av_1e_1} \vert111 \rangle_{bv_2e_2}\vert 111 \rangle_{cv_3e_3}\\
&\qquad+\vert \alpha_i' \rangle_1 \vert 0 \rangle_2 \vert p \rangle_3 \vert 0 \rangle_4\vert 0 \rangle_5\vert 111 \rangle_{av_1e_1} \vert000 \rangle_{bv_2e_2}\vert 111 \rangle_{cv_3e_3}\\
&\qquad+\vert 0 \rangle_1 \vert \alpha_i' \rangle_2 \vert p \rangle_3 \vert 0 \rangle_4\vert 0 \rangle_5\vert 000 \rangle_{av_1e_1} \vert111 \rangle_{bv_2e_2}\vert 111 \rangle_{cv_3e_3}\\
&\qquad+\vert 0 \rangle_1 \vert 0 \rangle_2 \vert p \rangle_3 \vert 0 \rangle_4\vert 0 \rangle_5\vert 000 \rangle_{av_1e_1} \vert000 \rangle_{bv_2e_2}\vert 111 \rangle_{cv_3e_3},\\
&H_2 \to \vert \alpha_i' \rangle_1 \vert p \rangle_2\vert 0 \rangle_3\vert 0 \rangle_4\vert \alpha_i \rangle_5\vert 111 \rangle_{av_1e_1}\vert111 \rangle_{bv_2e_2}\vert 000 \rangle_{cv_3e_3}\\
&\qquad+\vert 0 \rangle_1\vert p \rangle_2\vert 0 \rangle_3\vert 0 \rangle_4\vert\alpha_i \rangle_5\vert 000 \rangle_{av_1e_1}\vert111 \rangle_{bv_2e_2}\vert 000 \rangle_{cv_3e_3},\\
&H_3 \to \vert p \rangle_1\vert 0 \rangle_2\vert 0 \rangle_3\vert \alpha_i \rangle_4\vert \alpha_i \rangle_5\vert 111 \rangle_{av_1e_1}\vert000 \rangle_{bv_2e_2}\vert 000 \rangle_{cv_3e_3},\\
&H_4 \to \vert 0 \rangle_1\vert 0 \rangle_2\vert \alpha_i' \rangle_3\vert \alpha_i \rangle_4\vert p \rangle_5\vert 000 \rangle_{av_1e_1}\vert000 \rangle_{bv_2e_2}\vert 111 \rangle_{cv_3e_3}\\
&\qquad+\vert 0 \rangle_1\vert 0 \rangle_2\vert 0 \rangle_3\vert \alpha_i \rangle_4\vert p \rangle_5\vert 000 \rangle_{av_1e_1}\vert000 \rangle_{bv_2e_2}\vert 000 \rangle_{cv_3e_3},\\
&H_5 \to \vert 0 \rangle_1\vert \alpha_i' \rangle_2\vert \alpha_i' \rangle_3\vert p \rangle_4\vert 0 \rangle_5\vert 000 \rangle_{av_1e_1}\vert111 \rangle_{bv_2e_2}\vert 111 \rangle_{cv_3e_3}\\
&\qquad+\vert 0 \rangle_1\vert \alpha_i' \rangle_2\vert 0 \rangle_3\vert p \rangle_4\vert 0 \rangle_5\vert 000 \rangle_{av_1e_1}\vert111 \rangle_{bv_2e_2}\vert 000 \rangle_{cv_3e_3}\\
&\qquad+\vert 0 \rangle_1\vert 0 \rangle_2\vert \alpha_i' \rangle_3\vert p \rangle_4\vert 0 \rangle_5\vert 000 \rangle_{av_1e_1}\vert000 \rangle_{bv_2e_2}\vert 111 \rangle_{cv_3e_3}\\
&\qquad+\vert 0 \rangle_1\vert 0 \rangle_2\vert 0 \rangle_3\vert p \rangle_4\vert 0 \rangle_5\vert 000 \rangle_{av_1e_1}\vert000 \rangle_{bv_2e_2}\vert 000 \rangle_{cv_3e_3},\\
&H_6 \to \vert \gamma_m \rangle_1\vert 0 \rangle_2\vert p \rangle_3\vert 0 \rangle_4\vert 1 \rangle_5\vert 111 \rangle_{av_1e_1}\vert000 \rangle_{bv_2e_2}\vert 111 \rangle_{cv_3e_3},\\
&H_7 \to \vert 0 \rangle_1\vert p \rangle_2\vert 0 \rangle_3\vert 1 \rangle_4\vert \gamma_m \rangle_5\vert 000 \rangle_{av_1e_1}\vert111 \rangle_{bv_2e_2}\vert 000 \rangle_{cv_3e_3},\\
&H_8 \to \vert p \rangle_1\vert 0 \rangle_2\vert 1 \rangle_3\vert \gamma_m \rangle_4\vert 0 \rangle_5\vert 111 \rangle_{av_1e_1}\vert000 \rangle_{bv_2e_2}\vert 111 \rangle_{cv_3e_3},\\
&H_9 \to \vert 0 \rangle_1\vert 1 \rangle_2\vert \gamma_m \rangle_3\vert 0 \rangle_4\vert p \rangle_5\vert 000 \rangle_{av_1e_1}\vert111 \rangle_{bv_2e_2}\vert 111 \rangle_{cv_3e_3},\\
&H_{10} \to \vert 1 \rangle_1\vert \gamma_m \rangle_2\vert 0 \rangle_3\vert p \rangle_4\vert 0 \rangle_5\vert 111 \rangle_{av_1e_1}\vert111 \rangle_{bv_2e_2}\vert 000 \rangle_{cv_3e_3},
\end{aligned}
\]
where
\[
|\alpha_i'\rangle_k=\sum_{u=1}^{d_k-1}\omega_{d_k}^{iu}|u\rangle
\]
for \(k\in\{1,2,3,4,5\}\).

\textbf{Step 2.} Eve performs
\[
\begin{aligned}
\mathcal{M}_4\equiv \bigl\{&
M_{4,1}=P[(|0\rangle,\dots,|d_5'\rangle)_E;|1\rangle_{e_1};|0\rangle_{e_2};|0\rangle_{e_3}],\\
&M_{4,2}=P[|1\rangle_E;|1\rangle_{e_1};|0\rangle_{e_2};|1\rangle_{e_3}],\\
&M_{4,3}=P[(|1\rangle,\dots,|d_5'\rangle)_E;|0\rangle_{e_1};|1\rangle_{e_2};|1\rangle_{e_3}],\\
&M_{4,4}=P[(|1\rangle,\dots,|d_5'\rangle)_E;|0\rangle_{e_1};|0\rangle_{e_2};I_{e_3}],\\
&M_{4,5}=I-M_{4,1}-\cdots-M_{4,4}
\bigr\}.
\end{aligned}
\]
The corresponding subsets are
\[
\begin{aligned}
&M_{4,1}\Rightarrow H_3,\qquad &&M_{4,2}\Rightarrow H_6,\\
&M_{4,3}\Rightarrow H_9,\qquad &&M_{4,4}\Rightarrow H_4.
\end{aligned}
\]
If outcome \(M_{4,5}\) occurs, the state belongs to one of the remaining six subsets.

\textbf{Step 3.} Dave performs
\[
\begin{aligned}
\mathcal{M}_5\equiv \bigl\{&
M_{5,1}=P[(|1\rangle,\dots,|d_4'\rangle)_D;|1\rangle_{v_1};|1\rangle_{v_2};|0\rangle_{v_3}],\\
&M_{5,2}=P[(|1\rangle,\dots,|d_4'\rangle)_D;|1\rangle_{v_1};|0\rangle_{v_2};|1\rangle_{v_3}],\\
&M_{5,3}=P[|0\rangle_D;I_{v_1};I_{v_2};|1\rangle_{v_3}],\\
&M_{5,4}=P[|0\rangle_D;I_{v_1};|1\rangle_{v_2};|0\rangle_{v_3}],\\
&M_{5,5}=I-M_{5,1}-\cdots-M_{5,4}
\bigr\}.
\end{aligned}
\]
The corresponding subsets are
\[
\begin{aligned}
&M_{5,1}\Rightarrow H_{10},\qquad &&M_{5,2}\Rightarrow H_8,\\
&M_{5,3}\Rightarrow H_1,\qquad &&M_{5,4}\Rightarrow H_2.
\end{aligned}
\]
If outcome \(M_{5,5}\) occurs, the remaining possibilities are \(\{H_5,H_7\}\).

\textbf{Step 4.} Eve performs
\[
\mathcal{M}_6\equiv \bigl\{
M_{6,1}=|0\rangle_E\langle 0|,
\;
M_{6,2}=I-M_{6,1}
\bigr\}.
\]
If outcome \(M_{6,1}\) occurs, the subset is \(H_5\); otherwise, the subset is \(H_7\).

If other outcomes occur in Step 1, similar LOCC protocols can be constructed.

\section{The proof of Theorem 10}\label{T10}

Let \(\vert F\rangle=\vert 0000\rangle+\vert 1111\rangle\). Using the entanglement resources \(\vert F\rangle_{ac_1v_1e_1}\) and \(\vert F\rangle_{bc_2v_2e_2}\), the initial state is
\[
\begin{aligned}
\left|\psi\right\rangle_{ABCDE}\otimes \vert F\rangle_{ac_1v_1e_1}\otimes \vert F\rangle_{bc_2v_2e_2},
\end{aligned}
\]
where \(a\) and \(b\) are the ancillary systems of Alice and Bob, \(c_1,c_2\) are Charlie's ancillary systems, \(v_1,v_2\) are Dave's ancillary systems, and \(e_1,e_2\) are Eve's ancillary systems.

\textbf{Step 1.} Alice and Bob perform the measurements
\[
\begin{aligned}
\mathcal{M}_1\equiv \bigl\{&
M_{1,1}=P\bigl[|0\rangle_A;|0\rangle_a\bigr]
+P\bigl[(|1\rangle,\dots,|d_1'\rangle)_A;|1\rangle_a\bigr],\\
&M_{1,2}=I-M_{1,1}
\bigr\},
\end{aligned}
\]
and
\[
\begin{aligned}
\mathcal{M}_2\equiv \bigl\{&
M_{2,1}=P\bigl[|0\rangle_B;|0\rangle_b\bigr]
+P\bigl[(|1\rangle,\dots,|d_2'\rangle)_B;|1\rangle_b\bigr],\\
&M_{2,2}=I-M_{2,1}
\bigr\},
\end{aligned}
\]
respectively. Suppose that outcomes \(M_{1,1}\) and \(M_{2,1}\) occur. Then
\[
\begin{aligned}
&H_1 \to \vert \alpha_i' \rangle_1 \vert \alpha_i' \rangle_2 \vert p \rangle_3 \vert 0 \rangle_4\vert 0 \rangle_5\vert 1111 \rangle_{ac_1v_1e_1}\vert1111 \rangle_{bc_2v_2e_2}\\
&\qquad+\vert \alpha_i' \rangle_1 \vert 0 \rangle_2 \vert p \rangle_3 \vert 0 \rangle_4\vert 0 \rangle_5\vert 1111 \rangle_{ac_1v_1e_1}\vert0000 \rangle_{bc_2v_2e_2}\\
&\qquad+\vert 0 \rangle_1 \vert \alpha_i' \rangle_2\vert p \rangle_3 \vert 0 \rangle_4\vert 0 \rangle_5\vert 0000 \rangle_{ac_1v_1e_1}\vert1111 \rangle_{bc_2v_2e_2}\\
&\qquad+\vert 0 \rangle_1 \vert 0 \rangle_2 \vert p \rangle_3 \vert 0 \rangle_4\vert 0 \rangle_5\vert 0000 \rangle_{ac_1v_1e_1}\vert0000 \rangle_{bc_2v_2e_2},\\
&H_2 \to \vert \alpha_i' \rangle_1 \vert p \rangle_2\vert 0 \rangle_3\vert 0 \rangle_4\vert \alpha_i \rangle_5\vert 1111 \rangle_{ac_1v_1e_1}\vert1111 \rangle_{bc_2v_2e_2}\\
&\qquad+\vert 0 \rangle_1\vert p \rangle_2\vert 0 \rangle_3\vert 0 \rangle_4\vert \alpha_i \rangle_5\vert 0000 \rangle_{ac_1v_1e_1}\vert1111 \rangle_{bc_2v_2e_2},\\
&H_3 \to \vert p \rangle_1\vert 0 \rangle_2\vert 0 \rangle_3\vert \alpha_i \rangle_4\vert \alpha_i \rangle_5\vert 1111 \rangle_{ac_1v_1e_1}\vert0000 \rangle_{bc_2v_2e_2},\\
&H_4 \to \vert 0 \rangle_1\vert 0 \rangle_2\vert \alpha_i \rangle_3\vert \alpha_i \rangle_4\vert p \rangle_5\vert 0000 \rangle_{ac_1v_1e_1}\vert0000 \rangle_{bc_2v_2e_2},\\
&H_5 \to \vert 0 \rangle_1\vert \alpha_i' \rangle_2\vert \alpha_i \rangle_3\vert p \rangle_4\vert 0 \rangle_5\vert 0000 \rangle_{ac_1v_1e_1}\vert1111 \rangle_{bc_2v_2e_2}\\
&\qquad+\vert 0 \rangle_1\vert 0 \rangle_2\vert \alpha_i \rangle_3\vert p \rangle_4\vert 0 \rangle_5\vert 0000 \rangle_{ac_1v_1e_1}\vert0000 \rangle_{bc_2v_2e_2},\\
&H_6 \to \vert \gamma_m \rangle_1\vert 0 \rangle_2\vert p \rangle_3\vert 0 \rangle_4\vert 1 \rangle_5\vert 1111 \rangle_{ac_1v_1e_1}\vert0000 \rangle_{bc_2v_2e_2},\\
&H_7 \to \vert 0 \rangle_1\vert p \rangle_2\vert 0 \rangle_3\vert 1 \rangle_4\vert \gamma_m \rangle_5\vert 0000 \rangle_{ac_1v_1e_1}\vert1111 \rangle_{bc_2v_2e_2},\\
&H_8 \to \vert p \rangle_1\vert 0 \rangle_2\vert 1 \rangle_3\vert \gamma_m \rangle_4\vert 0 \rangle_5\vert 1111 \rangle_{ac_1v_1e_1}\vert0000 \rangle_{bc_2v_2e_2},\\
&H_9 \to \vert 0 \rangle_1\vert 1 \rangle_2\vert \gamma_m \rangle_3\vert 0 \rangle_4\vert p \rangle_5\vert 0000 \rangle_{ac_1v_1e_1}\vert1111 \rangle_{bc_2v_2e_2},\\
&H_{10} \to \vert 1 \rangle_1\vert \gamma_m \rangle_2\vert 0 \rangle_3\vert p \rangle_4\vert 0 \rangle_5\vert 1111 \rangle_{ac_1v_1e_1}\vert1111 \rangle_{bc_2v_2e_2},
\end{aligned}
\]
where
\[
|\alpha_i'\rangle_k=\sum_{u=1}^{d_k-1}\omega_{d_k}^{iu}|u\rangle
\]
for \(k\in\{1,2,3,4,5\}\).

\textbf{Step 2.} Eve performs
\[
\begin{aligned}
\mathcal{M}_3\equiv \bigl\{&
M_{3,1}=P[(|1\rangle,\dots,|d_5'\rangle)_E;|0\rangle_{e_1};|0\rangle_{e_2}],\\
&M_{3,2}=I-M_{3,1}
\bigr\}.
\end{aligned}
\]
If outcome \(M_{3,1}\) occurs, the subset is \(H_4\); otherwise, the state belongs to one of the remaining nine subsets.

\textbf{Step 3.} Dave performs
\[
\begin{aligned}
\mathcal{M}_4\equiv \bigl\{&
M_{4,1}=P[(|1\rangle,\dots,|d_4'\rangle)_D;|0\rangle_{v_1};I_{v_2}],\\
&M_{4,2}=P[(|1\rangle,\dots,|d_4'\rangle)_D;|1\rangle_{v_1};|1\rangle_{v_2}],\\
&M_{4,3}=I-M_{4,1}-M_{4,2}
\bigr\}.
\end{aligned}
\]
If outcome \(M_{4,1}\) occurs, the remaining possibilities are \(H_5,H_7\), which are LOCC distinguishable. If outcome \(M_{4,2}\) occurs, the subset is \(H_{10}\). Otherwise, the state belongs to one of the remaining six subsets.

\textbf{Step 4.} Charlie performs
\[
\begin{aligned}
\mathcal{M}_5\equiv \bigl\{&
M_{5,1}=P[|0\rangle_C;I_{c_1};|1\rangle_{c_2}],\\
&M_{5,2}=P[|0\rangle_C;|1\rangle_{c_1};|0\rangle_{c_2}],\\
&M_{5,3}=I-M_{5,1}-M_{5,2}
\bigr\}.
\end{aligned}
\]
If outcome \(M_{5,1}\) occurs, the subset is \(H_2\); if outcome \(M_{5,2}\) occurs, the subset is \(H_3\). If outcome \(M_{5,3}\) occurs, we proceed.

\textbf{Step 5.} Eve performs
\[
\begin{aligned}
\mathcal{M}_6\equiv \bigl\{&
M_{6,1}=P[(|1\rangle,\dots,|d_5'\rangle)_E;|1\rangle_{e_1};|0\rangle_{e_2}],\\
&M_{6,2}=P[(|1\rangle,\dots,|d_5'\rangle)_E;|0\rangle_{e_1};|1\rangle_{e_2}],\\
&M_{6,3}=I-M_{6,1}-M_{6,2}
\bigr\}.
\end{aligned}
\]
If outcome \(M_{6,1}\) occurs, the subset is \(H_6\); if outcome \(M_{6,2}\) occurs, the subset is \(H_9\). Otherwise, we continue.

\textbf{Step 6.} Dave performs
\[
\begin{aligned}
\mathcal{M}_7\equiv \bigl\{
M_{7,1}=P[|0\rangle_D;I_{v_1};I_{v_2}],
\;
M_{7,2}=I-M_{7,1}
\bigr\}.
\end{aligned}
\]
If outcome \(M_{7,1}\) occurs, the subset is \(H_1\); otherwise, the remaining subset is \(H_8\).

If other outcomes occur in Step 1, similar LOCC protocols can be constructed.

\end{appendix}

\end{document}